\documentclass[cmp]{svjourmod}

 \pdfoutput=1
 
\usepackage[pdftex]{graphicx}
\usepackage{amssymb}
\usepackage{bm}
\usepackage{psfrag}
\usepackage{amsmath}
\usepackage{amsfonts}
\usepackage{mathrsfs}
\usepackage{cite}

\newcommand{\p}{\partial}

\newcommand{\dd}{{\rm d}}

\begin{document}

\title{Affine sphere relativity}

\author{E. Minguzzi}
\institute{Dipartimento di Matematica e Informatica ``U. Dini'', Universit\`a degli
Studi di Firenze,  Via S. Marta 3,  I-50139 Firenze, Italy \\
\email{ettore.minguzzi@unifi.it} }
% \\ Phone: +39 055 4796 253, Fax: +39 055 471787 }
\authorrunning{E. Minguzzi}

%\author{E. Minguzzi\thanks{
%Dipartimento di Matematica e Informatica ``U. Dini'', Universit\`a
%degli Studi di Firenze, Via S. Marta 3,  I-50139 Firenze, Italy.
%E-mail: ettore.minguzzi@unifi.it} }

\date{}

\maketitle

\begin{abstract}
\noindent
We investigate spacetimes whose light cones could be anisotropic.
We prove the equivalence of the structures: (a) Lorentz-Finsler manifold for which the mean Cartan torsion  vanishes, (b) Lorentz-Finsler manifold for which the indicatrix (observer space) at each point is a convex hyperbolic affine sphere centered on the zero section, and (c) pair given by a spacetime volume and a sharp convex cone distribution.
The equivalence suggests to describe {\em (affine sphere) spacetimes} with this structure, so that no algebraic-metrical concept enters the definition.
As a result, this work shows how the metric features of spacetime  emerge from elementary concepts such as measure and order.
 Non-relativistic spacetimes are obtained replacing  {\em proper} spheres with {\em improper} spheres, so the distinction does not call  for group theoretical elements.
 In physical terms, in affine sphere spacetimes the light cone distribution and the spacetime measure determine the motion of massive and massless particles (hence the dispersion relation). Furthermore, it is shown that, more generally,  for Lorentz-Finsler theories non-differentiable at the cone, the lightlike geodesics and the transport of the particle momentum over them are well defined though the curve parametrization could be undefined. Causality theory is also well behaved.
  Several results for affine sphere spacetimes are presented. Some results in Finsler geometry, for instance in the characterization of Randers spaces, are also included.
\end{abstract}

\setcounter{tocdepth}{2}
\tableofcontents

\section{Introduction}

In recent years the Finslerian generalization of general relativity has made considerable progress. Several results including much of the edifice of causality theory and the famous singularity theorems have been generalized \cite{minguzzi13d,minguzzi14c,aazami14,minguzzi15}. Only a few but important difficulties still remain; this work is devoted to the solutions of some of those. As we  shall see their resolution will make us look at the spacetime concept in some novel ways.

%In order to introduce it let us recall that in

In general relativity it is possible to recover the Lorentz metric from  the spacetime volume form and the light cone distribution.  In fact,  it is well known that in spacetime dimension $n+1\ge 3$ two Lorentzian metrics on the same manifold share the same light cones if and only if they  are proportional, see e.g.\ \cite[App.\ D]{wald84}. The conformal factor can then be fixed to one imposing the equality of the volume forms.

This very simple property has prominent importance because it shows that the gravitational phenomena is encoded in two simple  concepts: the causal order and the spacetime measure. One could also add to this pair a further element, namely the spacetime topology.

This observation has led several researchers to believe that the quantization of gravity or better, of spacetime itself, must be formulated in terms of these structures. Among the theories that embody these ideas %apply this approach rather literally
we  might mention Causal Set Theory \cite{bombelli87}
%. Some theories try to distinguish between the dynamics of the conformal factor and %that of the cone structure, for instance
and unimodular gravity \cite{anderson71,henneaux89,bock03}.

% but we might also mention that there are several other theories which try to construct a dynamics of  the
%
%Proponents of the so called  have indeed  investigated a quantum theory of gravity in which these mathematical elements acquire a central role. The volume form is used in order to build a Lorentz invariant correspondence between manifolds and graphs (sprinkling process) while the causal order provided by the cones is used to direct the graphs so obtained. These authors investigate the possibility that the ultimate building block of spacetime could be  a directed graph, which is clearly a discrete and hence quantized object. Much of the theory concerns ways by which this process could be inverted since the concept of directed graph should be more fundamental than that of manifold.

We share the opinion that a fundamental theory should pass through the concepts of order, measure and topology and so that once the  manifold is given, one should be able to recover the metric from a volume form and a cone structure.
Unfortunately, this correspondence is lost for the so far proposed Finslerian generalizations of Einstein's general relativity, so this work aims to solve this problem.

It is perhaps worth to recall what is Finsler geometry before we become more specific. We might say that Riemannian spaces can be obtained from differentiable manifolds $M$ introducing a point dependent scalar product $g$ (Riemannian metric), which has the effect of converting each tangent space $T_xM$ into a (finite dimensional) Hilbert space. Similarly, Finsler spaces can be obtained  from manifolds $M$ by introducing a point dependent Minkowski norm $F_x$ or, which is the same,  a Finsler Lagrangian $\mathscr{L}=F^2/{2}$, which converts each tangent space $T_xM$ into a   Minkowski space, namely into a Banach space with strongly convex unit balls. These unit balls are also called {\em indicatrices}.

As the Minkowski ball is no more round (ellipsoidal), namely since it cannot be brought to a sphere through a linear change of coordinates on $T_xM$, Finsler geometry is essentially related to anisotropic features of the space.

In  Finslerian generalizations of Einstein's theory there are further complications related to the fact that the Finsler Lagrangian $\mathscr{L}$, having Lorentzian Hessian, induces non-compact unit balls (indicatrices).

We shall recognize that anisotropic theories of relativity can preserve the correspondence,
\begin{quote}
Finsler Lagrangian $\Leftrightarrow$ spacetime measure + light cone structure,
\end{quote}
provided the Finsler indicatrix is an affine sphere, or equivalently, provided the mean Cartan torsion (Tchebycheff form)  vanishes:
\begin{equation} \label{mei}
I_\alpha=0.
\end{equation}
This idea is the result of the physical interpretation of deep mathematical results by several distinguished mathematicians including Pogorelov, Calabi, Cheng and Yau. We shall also show that the coordinates introduced by Gigena to study affine spheres have a transparent physical interpretation. In particular, inhomogeneous projective coordinates should be used on the tangent space while homogeneous projective coordinates should be used on the cotangent space; in this way the former can be interpreted as velocity components while the latter as momentum components. The function $u({\bm v})$ solving the Monge-Amp\`ere equation of the affine sphere will receive the interpretation of observer Lagrangian of the theory.

In order to fully understand this solution we will have to introduce some concepts from {\em affine differential geometry}, as the reader might not be acquainted with this beautiful mathematical theory \cite{li93,nomizu94}. Thus  portions of the work will have a review character. We do not claim particular originality in this exposition, saved perhaps for the Finslerian point of view which at this stage is necessary in order to establish a connection with current literature on anisotropic gravity theories.

We recall that affine differential geometry originated with
Blaschke's  construction of a natural {transverse direction} - the {\em affine normal} - to any point on a  non-degenerate hypersurface immersed on  affine space. Remarkably, the construction does not require a scalar product, a fact which, ultimately,
%Blaschke's discovery that, without using a scalar product, it is possible to construct a natural {\em affine normal direction} to any point on a  non-degenerate hypersurface immersed on  affine space. Ultimately, this remarkable result
will allow us to give a definition of spacetime free from algebraic-metrical elements. For instance, the distinction between non-relativistic and relativistic physics will be devoid of group theoretical characterizations and related instead to the center of the affine sphere distribution, whether placed at infinity or not.

Since the vacuum equations of general relativity demand the proportionality between the Ricci tensor and the metric, one might ask whether the condition $I_\alpha=0$ has a similar characterization. We shall prove that the answer is affirmative in at least three different ways as $I_\alpha=0$ can be regarded as the K\"ahler-Einstein condition for the Lorentz-Finsler metric (Theor.\ \ref{poj}), as the K\"ahler-Einstein condition for the Monge-Amp\`ere (Cheng-Yau) Riemannian metric of the timelike cone (Theor.\ \ref{poh}), and also as a kind of  Einstein condition for the Blaschke structure of the indicatrix (Prop.\ \ref{ein}).

Much of this work will be devoted to the kinematics of the theory and to its interpretation. The many proposed Finslerian gravitational dynamics \cite{horvath50,horvath52,takano74,takano74b,ikeda79,ishikawa80,miron87,miron92,rutz93,storer00,stavrinos08,
stavrinos09,voicu10,vacaru12,castro12,pfeifer12,lammerzahl12,li14} can then be adapted to our kinematical model, adding the condition $I_\alpha=0$.  A dynamics proposed by the author which first suggested to consider a vanishing mean Cartan torsion can be found in \cite{minguzzi14c}. There it was shown that the (hh-)Ricci tensor $R^\mu_{\ \alpha \mu \beta}$ appearing in most dynamical proposals is symmetric if $I_\alpha =0$, and also that these spaces are weakly-Berwald and weakly-Landsberg.

In (positive definite) Finsler spaces the condition $I_\alpha=0$ was already considered by Cartan \cite{cartan23}, but later Deicke \cite{deicke53} discovered that Finsler spaces satisfying this condition were Riemannian and hence isotropic. Of course, the interest in this condition faded, since the many results obtained through its imposition were a consequence of the triviality of the Finsler space.
 %who found it geometrically attractive. Later  The interest in this condition faded since many interesting results obtained through its imposition were interpreted as a natural consequence of the isotropy and hence triviality of the Finsler space.
 Early authors working in Finsler gravity did not pay much attention to the signature of  the metric, so some of them  discarded this condition \cite{ikeda79} although  Deicke's theorem really holds only for positive definite metrics.

Exact solutions will not be considered in this work but in a related paper \cite{minguzzi16c} we shall provide examples of affine sphere  spacetimes which reduce themselves to the Schwarzschild, Kerr, FLRW  spacetimes in a suitable velocity limit, and hence which satisfy the Lorentzian Einstein equations in the same limit.
 %Independently of the full set of field equations adopted, exact solutions satisfying $I_\alpha=0$ have yet to be found saved, of course, for the Lorentzian solutions or for some Lorentz-Minkowski spaces with trivial dependence on the spacetime event $x$.
  In general, the affine sphere condition $I_\alpha=0$ is quite hard to solve
 %This difficulty should be expected since this condition is  hard to solve analytically
 though, as we shall recall, general theorems  guarantee the existence of solutions. Mathematicians are working to find new methods to generate affine spheres in closed form \cite{fox14,hildebrand14b}. A perturbative approach seems more amenable but will be pursued in a different work.

This work is organized as follows. In the first section we recall some elements of Lorentz-Finsler theory, we define the indicatrix and we introduce the quotient and the induced metrics on the indicatrix. We introduce the canonical Hessian metric of the timelike cone and relate the Finsler Lagrangian to the K\"ahler potential of the cone.
%, and study the causality properties of $T_xM$.
We also give some arguments pointing to a null mean Cartan torsion, which can be added to those already discussed in \cite{minguzzi14c}. This condition makes it possible to identify the spacetime volume form in the usual way and can be regarded as a  K\"ahler affine condition of Ricci flatness on the vertical degrees of freedom.

In the second section we introduce the mathematics of affine spheres, we characterize affine spheres through the mean Cartan torsion $I_\alpha$, we clarify the role of the volume form on spacetime, we show how to convert affine sphere theoretical results into Finslerian results (and conversely), and  prove some theorems required for the physical interpretation of affine spheres. We introduce both proper and improper spheres, the physical theory constructed from those leading respectively to relativistic and non-relativistic physics.

%is shown to dispense ourselves from the identification of a volume form in Lorentz-Finsler spaces and it allows us to interpret the vertical dynamics as a K\"ahler affine condition of Ricci flatness.
%, and reduces considerably certain indeterminacies in the horizontal dynamics.
%They include a study of volume forms on Lorentz-Finsler spaces, and the interpretation of the vertical dynamics as a K\"ahler affine condition of Ricci flatness.
%The third chapter is devoted to the theory of homogeneous convex cones and affine spheres and to their relationship with Lorentz-Finsler spaces.
The third section is devoted to the application of the results of the previous sections to the geometrical and physical interpretation of Lorentz-Finsler spaces having vanishing mean Cartan torsion. Here we use  a deep mathematical theorem, first conjectured by Calabi, in order to connect volume and conic order on spacetime with the affine sphere distribution on the tangent bundle. We are then able to give a definition of affine sphere spacetime that does not involve metrical or group theoretical elements.

In the fourth section   we return to the broader framework of Lorentz-Finsler theories. We show that the lightlike geodesic flow follows solely from the distribution of light cones and so does the transport of the photon momenta along the geodesic. These results are really independent of the Lagrangian and so do not use its differentiability at the light cone. They require just differentiability and convexity conditions on the distribution of light cones. Finally, we prove that the standard results of causality theory are preserved.

For space reasons the discussion of the relativity principle and a study of some models satisfying it will be given in a different work \cite{minguzzi16c}.
  %several other applications will be given in a different works.

%Here the relativity principle is identified with the homogeneity condition on the affine sphere.
%The last chapter applies the previous theory to those four dimensional models which respect an exact form of the relativity principle.

\subsection{Elements of Lorentz-Finsler theory}

Concerning notation and terminology, the Lorentz signature is $(-,+,\cdots, +)$. The wedge product between 1-forms is defined by $\alpha \wedge \beta=\alpha \otimes \beta-\beta\otimes \alpha$. On an affine space $E$ the Hessian metric of a function $f$ with respect to affine coordinates will be denoted, with some abuse of notation, $\dd^2 f$. The inclusion is reflexive: $X\subset X$. The manifold $M$  has dimension $m=n+1\ge 2$ and it will be physically interpreted as  the spacetime. Greek indices take values $0,1,\cdots,n$ while Latin indices take values in $1,\cdots,n$.  We shall often write ${\bm y}$ in place of $y^i$.  Local coordinates on $M$ are denoted $\{x^\mu\}$ while the induced local coordinate system on $TM$ is $\{x^\mu,y^\nu\}$, namely $y^\nu \p/\p x^\nu \in T_xM$.

A point in the space $TTM$ will be denoted with $(x,y,\dot x,\dot y)$. Observe that the canonical projection $\pi\colon TM\to M$, $(x,y)\mapsto x$, has pushforward $\pi_*\colon TTM\to TM$, $(x,y,\dot x,\dot y)\mapsto (x,\dot x)$ so the vertical space $VTM$ consists of the points $(x,y,0,\dot y)$ and it is naturally diffeomorphic to $TM\times_M TM$ (it can be easily checked calculating the cocyle after a change of coordinates $\tilde x=\tilde x(x)$ on the base \cite{godbillon69,minguzzi14c}).

We start giving a quite general setting for Finsler spacetime theory, which we call the  {\em rough model} \cite{minguzzi14h,asanov85,pimenov88}.

Let $\Omega$ be a  subbundle of  the slit tangent bundle, $\Omega \subset TM\backslash 0$, such that $\Omega_x$ is an open sharp convex cone for every $x$.
A {\em Finsler Lagrangian} is a map  $\mathscr{L}\colon \Omega \to \mathbb{R}$ which
is  positive homogeneous of degree two in the fiber coordinates
\[
\mathscr{L}(x,sy)=s^2 \mathscr{L}(x,y), \qquad \forall s>0.
  \]
  It is assumed that the fiber dependence  is at least $C^3(\Omega)$,
 that $\mathscr{L}<0$ on $\Omega$ and that $\mathscr{L}$ can be continuously extended setting $\mathscr{L}=0$ on $\p \Omega$. We might denote $\mathscr{L}_x:=\mathscr{L}|_{\Omega_x}$.
The matrix metric is defined as the Hessian  of $\mathscr{L}$ with respect to the fibers
\[
g_{\mu \nu}(x,y)= \frac{\p^2 \mathscr{L}}{\p y^\mu \p y^\nu}.
\]
This matrix can be used to define a metric in two different, but essentially equivalent  ways. The Finsler metric is typically defined as $g=g_{\mu \nu}(x,y) \dd x^\mu \dd x^\nu$ and is a map $g\colon \Omega \to  T^*M \otimes T^*M$.  For any given $x$ one could also use this matrix  and the mentioned diffeomorphism with the vertical space to define a vertical metric on $\Omega_x$ as follows $g_{\mu \nu}(x,y) \dd y^\mu \dd y^\nu$. Most often we shall use the latter metric, but should nevertheless be clear from the context which one is meant. In index free notation the metric will be also denoted  $g_y$ to stress the dependence on the fiber coordinates.
Given a non-linear connection one could also interpret these two metrics as two different restrictions, horizontal or vertical, of the Sasaki metric on $\Omega$
\[
g_S=g_{\mu \nu}(x,y) \,\dd x^\mu \dd x^\nu+g_{\mu \nu}(x,y)\, \delta y^\mu \delta y^\nu,
\]
where $\delta y^\mu=\dd y^\mu+N^\mu_\nu(x,y) \dd x^\nu$ and $N^\nu_\mu$ are the coefficients of the non-linear connection.

The manifold $(M,\mathscr{L})$ is called a {\em Finsler spacetime} whenever   $g_y$ is Lorentzian, namely of signature $(-,+,\cdots,+)$. By positive homogeneity we have $\mathscr{L}=\frac{1}{2} g_y(y,y)$ and $\dd \mathscr{L}=g_y(y,\cdot)$. The usual Lorentzian-Riemannian case is obtained for $\mathscr{L}$ quadratic in the fiber variables. The vectors belonging to $\Omega_x$ are called {\em timelike} while those belonging to $\p \Omega_x\backslash\{0\}$ are called {\em lightlike}.  We shall also denote the former set with $I^+_x$ and the latter set with $E^+_x$, often dropping the plus sign. A vector is {\em causal} if it is either timelike or lightlike, the set of causal vectors being denoted $J^+_x$.
%We might also add or drop the adjective {\em future directed}.
The plus sign is introduced for better comparison with  notations of Lorentzian geometry and general relativity and can be dropped in most parts of this work.

There are other approaches to Lorentz-Finsler geometry which are contrasted in \cite{minguzzi14h}. For instance, one might start with a Finsler Lagrangian  defined on the whole slit tangent bundle $TM\backslash \{0\}$, in which case it is possible to prove, for $n\ge 2$ and for reversible Lagrangians $\mathscr{L}(x,-y)=\mathscr{L}(x,y)$, that the timelike set $\{\mathscr{L}<0\}$ is the union of two convex sharp cones \cite{minguzzi13c} (see also \cite{beem70,beem74,perlick06}). A time orientability assumption allows one to select a future  $I_x^+$ and a past $I_x^-$ continuous cone distribution as in Lorentzian geometry.
%One can be identified with the future cone and the other with the past cone.
The present study applies to this framework as well provided the future cone is identified with $\Omega$ and the Finsler Lagrangian is there restricted. Observe that we do not demand the differentiability of the Finsler Lagrangian at the boundary $E^+$, nor that the metric can be continuously extended to it. This condition would make it possible to replace the Finsler Lagrangian with an extension defined over the whole slit tangent bundle \cite{minguzzi14h}.

The space indicatrix, or observer space, or simply {\em the indicatrix} is the set\footnote{Whenever the Lagrangian is defined over the whole slit tangent bundle it  can be useful to define  \cite{minguzzi14h} the  {\em light cone} indicatrix $\mathscr{I}^0$ or the {\em spacetime indicatrix}  $\mathscr{I}^+$ obtained for $\mathscr{L}=0$ or $2\mathscr{L}=1$. The names follow from the signature of the induced metrics.}
\[
\mathscr{I}^-=\{(x,y)\in \Omega \colon 2\mathscr{L}(x,y)=-1\}
\]%
Once again there will be no ambiguity in dropping the minus sign.

%
%%Suppose that $g_{\mu \nu}$ is Lorentzian.
%The vectors $y\in \Omega_x$ are called timelike, lightlike or spacelike depending on the sign of $\mathscr{L}(x,y)$ respectively negative, zero or positive. A vector is causal if it is either timelike or lightlike.  It can be shown that  the timelike subbundle $I=\{(x,y)\in \Omega\colon \mathscr{L}(x,y)<0\}$ is the union of open disjoint convex sharp  cones on $\Omega$ \cite{beem70,beem74,perlick06,minguzzi13c}.
%The {\em indicatrix} bundle $\mathscr{I}$ is the union of three sets
%\begin{itemize}
%\item[] $\mathscr{I}^-=\{(x,y)\in \Omega \colon 2\mathscr{L}(x,y)=-1\}$ \quad {\em observer (velocity) space indicatrix},
%\item[] $\mathscr{I}^0=\{(x,y)\in \Omega \colon  2\mathscr{L}(x,y)=0\}$\qquad  \ {\em light cone indicatrix},
%\item[] $\mathscr{I}^+=\{(x,y)\in \Omega \colon 2 \mathscr{L}(x,y)=1\}$\qquad  {\em spacetime indicatrix}.
%\end{itemize}
%The name for the last set follows from the Lorentzianity of the induced metric on this set (cf.\ Sect.\ \ref{kzo}). The first two sets can actually be disconnected on $\Omega$, so for a sound physical interpretation it is necessary to introduce a time orientation, that is, a special selection of one of the components (see below). For $n\ge 2$ and for $\Omega=TM\backslash 0$ the set $ \mathscr{I}^+$ is necessarily connected \cite{beem74,minguzzi13c} and if the Lagrangian is reversible  the timelike subbundle $I$ has just two components $I^+$ and $I^-$ as in Lorentzian geometry \cite{minguzzi13c}.

 Due to positive homogeneity the Finsler Lagrangian can be recovered from the indicatrix as follows
\begin{equation} \label{jui}
\mathscr{L}(x,y)=- s^2/2, \textrm{ where } s>0 \textrm{ and }  y/s \in \mathscr{I}^-.
\end{equation}

The Cartan torsion is defined by
\begin{equation} \label{mac}
C_{\alpha \beta \gamma}(x,y):=\frac{1}{2} \,\frac{\p}{\p y^\gamma} \,g_{\alpha \beta} ,
\end{equation}
it is symmetric and satisfies $C_{\alpha \beta \gamma} y^\gamma=0$. Its traceless part will be denoted with $M_{\alpha \beta \gamma}$. The Cartan curvature is $C_{\alpha \beta \gamma \delta}:=\frac{\p}{\p y^\delta} C_{\alpha \beta \gamma}$. For every $x$ the set $\Omega_x$ endowed with the vertical metric $g_{\mu \nu}(x,y) \dd y^\mu \dd y^\nu$ has Levi-Civita connection coefficients $C^\alpha_{\beta \gamma}$ in the coordinates $\{y^\mu\}$.
The mean Cartan torsion is
\begin{equation} \label{aqs}
I_\alpha:=g^{\mu \nu} C_{\mu \nu \alpha}=\frac{1}{2}\frac{\p}{\p y^\alpha} \log \vert \det g_{\mu \nu}\vert ,
\end{equation}
where for the last equality we used Jacobi's formula for the derivative of the determinant.
%The set $\Omega_x$ endowed with the vertical metric, namely with $g_{\mu \nu}(x,y) \dd y^\mu \dd y^\nu$, is a pseudo-Riemannian space. Its curvature is denoted $S^{\alpha}_{\ \beta \gamma \delta}$  (this is the $vv$-curvature of the Cartan Finsler connection \cite{minguzzi14c}) and a simple calculation shows that in the coordinates $\{y^\mu\}$ \cite{cartan34,minguzzi14c}
%\begin{equation}
%S^{\alpha}_{\ \beta \gamma \delta}=C^\alpha_{\mu \delta} C^\mu_{\beta \gamma}-C^\alpha_{\mu \gamma} C^\mu_{\beta  \delta}.
%\end{equation}
%This formula will be useful in the study of the vertical dynamics.

A well known problem in Finsler geometry is that of providing a natural notion of manifold volume form. There have been several proposals, the most popular being the Busemann's  and the Holmes-Thompson's volume forms \cite{paiva04}. Unfortunately, none of them can work in a Lorentz-Finsler framework since they rely on the compactness of the  indicatrix.

In pseudo-Riemannian geometry there is a simple volume form associated to any metric. In a local coordinate system it is given by
\begin{equation} \label{nos}
\dd \mu=\left| \sqrt{ \vert \det  g_{\mu \nu}\vert} \dd x^0\wedge \dd x^1\wedge \cdots\wedge \dd x^n\right|= \sqrt{ \vert \det  g_{\mu \nu}\vert}\,\dd^{n+1} x.
\end{equation}
where $\vert\,\vert$ reminds us that we are taking the equivalence class, that is, we are regarding as equivalent any  two $n$+1-forms differing by a sign.

 Since in Physics there seems to be the need of a well defined spacetime volume we find in Eq.\ (\ref{aqs})  a first motivation for imposing the condition $I_\alpha=0$. This is the simplest condition which assures that a natural volume form on spacetime could be defined. In fact, if it holds true we can adopt the usual pseudo-Riemannian expression for the volume form.

\subsection{Quotient  and  induced metrics} \label{kzo}
This section introduces the notion of quotient metric, and of induced (angular) metric on the indicatrix. It is known material  introduced here just  to fix the notation and terminology.

The pair $(\Omega_x,g)$ is a Lorentzian manifold. Let $Q_x$ be the quotient of $\Omega_x$ under the action of
homotheties.
%Let  us endow $\Omega_x$ with the metric $g=g_{\mu \nu}(x,y) \dd y^\mu \dd y^\nu$ so that $(\Omega_x,g)$ is a Lorentzian manifold. Sometimes we shall write $g_y$ in coordinate free notation in order to stress the dependence on the fiber point. Let $Q_x$ be the quotient of $\Omega_x$ under the action of
%homotheties.
%Clearly, $Q_x$ is diffeomorphic to $S^n$ if $\Omega_x=T_xM\backslash 0$.
The bundle $\pi^Q\colon
\Omega_x \to Q_x$ is principal, the group action on it being
the group of dilations $(\mathbb{R},+)$, where any homothety acts as
$y \mapsto e^s y$, for some $s\in \mathbb{R}$. The one-parameter
group of diffeomorphisms is generated by the Liouville vector field $k\colon: T_xM\to T T_xM$
\[k(y)=y=y^\mu \p/\p y^\mu.\]
The positive homogeneity of the metric $g_{\mu \nu}(x,sy)=g_{ \mu \nu}(x,y)$ implies
$\mathcal{L}_k g=2 g$,
where $\mathcal{L}$ is the Lie derivative, thus $k$ is a Killing
vector for the metric $g/ \vert 2\mathscr{L}(x,y)\vert$.
The principal bundle $\Omega_x$ can be endowed with a natural connection
1-form
%$\omega\colon T(V\backslash E)\sim V\to \mathbb{R}$,
\begin{equation}
\omega:= \frac{g_y(y,\cdot)}{g_y(y,y)}.
\end{equation}
%The principal subbundle of non-lightlike vectors  $\Omega_x\backslash E_x$
%$\pi\vert_{V\backslash E}\colon V\backslash
%E \to \pi(V\backslash E)$,
%is  endowed with a
Indeed, $\omega$ satisfies the defining conditions of a connection
1-form on a principal bundle \cite{kobayashi63} (recall that
$\mathcal{L}_k k=0$)
\begin{align*}
\mathcal{L}_k \omega &=0, \qquad \omega(k)=1.
\end{align*}
Let us define
\begin{equation}
F=\sqrt{2\vert \mathscr{L}\vert}.
\end{equation}
The connection 1-form is integrable and the principal bundle is
trivial because the connection is exact
\[
\omega=\dd \log F.
\]
 It is also possible to define a metric on $Q_x$ while
working with vectors on $T\Omega_x\simeq T_xM$. This process is quite well known in
relativity theory  \cite{geroch71} and has been called {\em
indicatorization} in the literature on Finsler spaces
\cite{matsumoto77}.

Let us consider the metric on $\Omega_x$
\begin{equation} \label{ten}
h=\frac{1}{\vert2 \mathscr{L} \vert}\Big(g-\frac{g_y(y,\cdot)\otimes  g_y(y,\cdot)}{g_y(y,y)}\Big),
\end{equation}
which in coordinates $\{y^\mu\}$ reads
\begin{equation} \label{dqa}
h_{ \mu \nu}=\frac{1}{\vert2 \mathscr{L} \vert}\big(\mathscr{L}_{,\mu,\nu}-\frac{1}{2 \mathscr{L}}\, \mathscr{L}_{,\mu}
\mathscr{L}_{,\nu}\big)=-\frac{1}{F} \,F_{,\mu,\nu}
\end{equation}
then
%by Eq.\ (\ref{koi})
\[
h(y,\cdot)=0, \qquad h_{s y}= h_y, \textrm{ for } s>0,
\]
where the last property can be written $\mathcal{L}_k h=0$.
Thus
$h$ depends only on the point of $Q_x$ and annihilates the radial position vector $y$, so  it
defines a metric $\tilde{h}$ on the quotient $Q_x$ through
\[
\tilde{h}_{\tilde{y}}(w_1,w_2):=h_y(W_1,W_2) ,
\]
where $w_1,w_2\in T_{\tilde{y}}Q_x$. Here $y\in \Omega_x$ is any vector such
that $\pi^Q(y)=\tilde y$,
 and $W_1,W_2\in
T_xM$ are representatives of $w_1,w_2\in T_{\tilde{y}}Q_x$ in the sense that
$\pi^Q_*(W_i)=w_i$, $i=1,2$. Since $h$ is homogeneous of zero degree  and
annihilates $y$,  the defining expression is well posed as it is
independent of the choice of representatives $(y,W_1,W_2)$. Observe
that $h=\pi^Q{}^{*}\tilde{h}$ but in what follows we might not be too
rigorous in distinguishing between $h$ and $\tilde h$.
%, as this
%usually does not cause ambiguities.

\begin{remark}
The metric (\ref{dqa}) is also the {\em induced metric} on the indicatrix $\mathscr{I}^-$ since there $2\mathscr{L}=-1$, and the vectors tangent to the indicatrix annihilate $\dd \mathscr{L}$, so over vectors tangent to the indicatrix $h_{ \mu \nu}=\mathscr{L}_{,\mu,\nu}$.  In Finsler geometry it is called  {\em angular metric} but in Lorentz-Finsler theory the name {\em acceleration metric} seems more appropriate. %In fact $x: I \to M$ is a timelike curve which is parametrized with respect to proper-time $g_{\dot x}(\dot x,\dot x)=-1$, then the acceleration is defined by $a=D \dot x$
%If  $\theta^i$, $i=1,\cdots, n$ are coordinates in $\mathscr{I}^{\pm}$, namely if $y^{\mu}\colon U \to T_xM$, $U\subset \mathbb{R}^n$, $\theta^i \mapsto y^\mu(\theta^i)$, is the embedding of the indicatrix we have $h_{i j}=h_{\mu \nu} y^\mu_{, i} y^\nu_{, j}=g_{\mu \nu} y^\mu_{, i} y^\nu_{, j}$, where $h_{i j}\dd \theta^i \dd \theta^j$ is, more precisely, the induced metric. If $h^{ij}$ is its inverse then at any point of $\mathscr{I}^{\pm}$
%\begin{equation} \label{inv}
%h^{\alpha \beta}:= h^{ij} y^\alpha_{, i} y^\beta_{, j}=g^{\alpha \beta}-\frac{y^\alpha y^\beta}{g_y(y,y)}, \qquad g_y(y,y)=\pm 1.
%\end{equation}
%\end{remark}
%This useful equation can be easily verified noticing that $\{y^\alpha \frac{\p}{\p y^\alpha}, y^\beta_{,i} \frac{\p}{\p y^\beta} \}$ provides a basis of $T_xM$ at the point and verifying the contraction of (\ref{inv}) with $g_{\alpha \mu} y^\mu$ and $g_{\alpha \mu} y^{\mu}_{, k}$.
\end{remark}
From (\ref{ten}) the metric $g$ reads
%conformal to (in physics terminology this is a Kaluza-Klein metric \cite{trautman84})
\begin{equation} \label{kkk}
g=\vert 2 \mathscr{L}\vert
\big(-\omega\otimes \omega+h\big).
\end{equation}
Since $g$ is Lorentzian, $h$ is  Riemannian over $\mathscr{I}^-$. This decomposition can also be written in polar form
\begin{equation} \label{war}
g=- \dd F^2+F^2 h.
\end{equation}
%
%\begin{remark}[Quotient metric over the light cone] Let us assume that $g$ is $C^2$ over $E$.
%Although the metric $\tilde h$ is not defined over $\tilde E:=\pi(E)$ we can define the following object. For every section $\sigma:\tilde E\to E$ we have a Riemannian metric over $\tilde E$, $\gamma_\sigma\colon \tilde E\to T\tilde E\otimes T\tilde E$, with the property $\gamma_{s\sigma}=s^2\gamma_\sigma$, for every $s>0$. Indeed we can define $\gamma_{\tilde y}(w_1,w_2):=g_{\sigma(\tilde y)}(W_1,W_2)$
%where $w_1,w_2\in T_{\tilde{y}}E$, and $W_i$ are representatives in the sense that for $i=1,2$, $\pi_*(W_i)=w_i$ and $g_y(y,W_i)=0$.
%\end{remark}

%Thus $g_v$ is the direct sum of the metric $2 L \omega_v^2$ on the
%line tangent to the fiber and the metric $2 L h_v$ on the plane
%$\textrm{Ker}\, \omega_v$.

\subsection{Riemannian Hessian metric on the timelike subbundle} \label{rho}
%In this section let $\Omega=I$ be the the open convex sharp conic subbundle  of the tangent bundle generated by the indicatrix $\mathscr{I}^{-}$.
The Lorentz-Finsler structure $\mathscr{L}$ on $\Omega$, induces a Lorentzian metric on each fiber $\Omega_x$ which is in one-to-one correspondence with a Riemannian Hessian structure on $\Omega_x$ induced by a $m$-logarithmically homogeneous potential.

Let us construct this correspondence (compare with recent work in \cite{hildebrand14,fox15}). From the previous section, the metric on $\Omega_x$ can be written
\[
g=\dd^2 \mathscr{L}
%\frac{\p^2 \mathscr{L}}{\p y^\alpha\p y^\beta} \,\dd y^\alpha \dd y^\beta
=\frac{1}{2 \mathscr{L}}\,(\dd \mathscr{L})^2\oplus (- 2 \mathscr{L}) h=-\dd F^2+F^2 h, \quad \mathscr{L}=-\frac{1}{2}\, F^2.
\]
This Lorentz metric is positive homogeneous of degree two. If we look for a scale invariant complete Riemannian metric on $\Omega_x$ it is natural to consider  the Hessian (K\"ahler affine)  metric ($m=n+1$)
\begin{equation} \label{nxo}
\hat{g}=m\Big[\frac{1}{(2 \mathscr{L})^2}\,(\dd \mathscr{L})^2\oplus h\Big]=m\big[(\dd \log F)^2\oplus h\big]
%=-\frac{n+1}{2}\frac{\p^2 \log \vert\mathscr{L} \vert}{\p y^\alpha \p y^\beta} \,\dd y^\alpha \dd y^\beta
=\dd^2  \log V
%=\frac{\p^2 \log V}{\p y^\alpha \p y^\beta} \,\dd y^\alpha \dd y^\beta ,
\end{equation}
where
\begin{equation} \label{dup}
V=\Big(\frac{-2\mathscr{L}}{m}\Big)^{-\frac{m}{2}}=\Big(\frac{F}{\sqrt{m}}\Big)^{-m}.
\end{equation}
Here the denominator has been chosen so as to get Eq.\ (\ref{nxp}).
The function $\log V$ is the K\"ahler potential \cite{cheng82}. It is ($-m$)-logarithmically homogeneous, namely
\[
\log V(sy)=\log V(y)-m\log s.
\]
Conversely, let $\log V$ be $(-m)$-logarithmically homogeneous with complete positive definite Hessian metric  on $\Omega_x$ for every $x$, then it is possible to define a Lorentz-Finsler structure on $\Omega$ inverting (\ref{dup}).

 Writing $g$ in place of $h$ in (\ref{nxo})
\begin{equation} \label{dis}
\hat{g}=m\Big[\frac{2}{(2 \mathscr{L})^2}\,(\dd \mathscr{L})^2\oplus \frac{1}{-2\mathscr{L}}\, g\Big] ,
\end{equation}
  and using the rank one update of the determinant  we get
\begin{equation} \label{nxp}
\det \hat{g}_y=-(\det g_y) V^2,
\end{equation}
thus
\begin{equation}
I_\mu=\frac{1}{2} \frac{\p}{\p y^\mu} \log \left( V^{-2} \det \frac{\p^2 \log V}{\p y^\alpha \p y^\beta} \right).
\end{equation}
By   positive homogeneity of degree $-1$ of $I_\alpha$, this identity is equivalent to
\begin{equation} \label{jis}
K_{\alpha \beta}=\hat K_{\alpha \beta}+2\hat g_{\mu \nu}.
\end{equation}
Here we have introduced the  K\"ahler Ricci tensor of a K\"ahler affine metric  (it is not the usual Ricci tensor) for both the Lorentzian and Riemannian metrics
\begin{align}
K_{\alpha \beta}&:=-\frac{\p^2}{\p y^\alpha \p y^\beta} \log \vert \det {g}_y\vert =-2\frac{\p}{\p y^\alpha} I_\beta,\\
\hat K_{\alpha \beta}&:=-\frac{\p^2}{\p y^\alpha \p y^\beta} \log \det \hat{g}_y. \label{ddp}
\end{align}
This definition is inspired by analogous definitions in K\"ahler geometry \cite{cheng82}. The connection with complex K\"ahler geometry can be made more precise introducing a tube domain, but this approach will not be pursued here. The Hessian metric $\hat g$ is K\"ahler-Einstein if
\begin{equation} \label{sis}
\hat K_{\alpha \beta}=\hat \kappa(x,y) \hat g_{\alpha \beta}.
\end{equation}
Observe that both $\hat K_{\alpha \beta}$ and $\hat g_{\alpha \beta}$ are Hessian metrics, thus their  vertical derivative is a symmetric tensor. A simple observation by Knebelman \cite{knebelman29b}, originally conceived for Finsler metrics but perfectly valid for Hessian metrics, shows that $\hat \kappa$ is actually independent of $y$, thus the previous equation is equivalent to
\[
-\log \det \hat{g}_y=\frac{\hat \kappa}{2} \log V^2+a+b_\alpha y^\alpha,
\]
for some $a,b_\alpha$ independent of $y$. However, (\ref{nxp}) shows that $\log \det \hat{g}_y$ and $\log V^2$ are $(-2m)$-logarithmically homogeneous, thus $b_\alpha=0$ and $\hat \kappa=-2$, namely $\det \hat{g}_y= e^{-a(x)} V^2.$ The comparison of this equation with (\ref{nxp}) shows that $\det g_y$ does not  depend on $y$. Conversely, if $\det g_y$ does not depend on $y$ then  (\ref{sis}) holds true, just use Eq.\ (\ref{jis}).
 % take the logarithm of (\ref{nxp}) and differentiate twice.
We conclude

\begin{theorem} \label{poh}
The complete, Riemannian, Hessian metric $\hat g$ on $\Omega_x$ is K\"ahler-Einstein if and only if the mean Cartan torsion vanishes: $I_\alpha=0$. In this case $\hat k=-2$ and
\[
\det  \frac{\p^2 \log V}{\p y^\alpha \p y^\beta}=\alpha V^2, \qquad \alpha =-\det g_{\alpha \beta}.
\]
\end{theorem}
 If this equation is satisfied, $\hat g$ is called the {\em Monge-Amp\`ere} or the {\em Cheng-Yau metric} of the cone $\Omega_x$.
%Observe that one could also define
%\begin{equation} \label{kat}
%K_{\alpha \beta}:=-\frac{\p^2}{\p y^\alpha \p y^\beta} \log \det {g}_y=-2\frac{\p}{\p y^\alpha} I_\beta.
%\end{equation}

We have a similar result for the Einstein condition, $K_{\alpha \beta}=\kappa(x,y) g_{\alpha \beta}$, on the Lorentzian metric (compare \cite[Sect.\ 5]{ishikawa81}).

\begin{theorem} \label{poj}
The Lorentzian Hessian metric $ g$ on $\Omega_x$ is K\"ahler-Einstein if and only if the mean Cartan torsion vanishes: $I_\alpha=0$. In this case $\kappa=0$.
\end{theorem}

\begin{proof}
Once again Knebelman observation implies that $\kappa$ does not depend on $y$. Thus multiplying the Einstein condition by $y^\alpha$ and using the positive homogeneity of degree $-1$ of $I_\alpha$, $2I_\beta=\kappa(x) y_\beta$. Applying $y^\gamma\frac{\p}{\p y^\gamma}$ to both sides gives $-2I_\beta=\kappa y_\beta$ thus $I_\alpha=0$. $\square$
\end{proof}

\section{Preliminaries on affine spheres and indicatrices}

Let us consider a pair $(E,\omega)$ where $E$ is an affine space modeled over a $n+1$-dimensional vector space $V$ and $\omega$ is a non-trivial alternating multilinear $n+1$-form over $V$, sometimes called {\em determinant} (not to be confused with the determinant of an endomorphism). In short we are considering an affine space with a translational invariant notion of oriented volume.
%The group of maps from $E$ to itself which preserves this structure is called equiaffine group.

Next let $f\colon N\to E$ be a $C^3$ immersion where $N$ is a $n$-dimensional manifold. The manifold $N$ is termed {\em hypersurface} and $f$ is called {\em hypersurface immersion}.  Let $\xi\colon N\to TE$, $p\mapsto \xi_p$, be a vector field over $f(N)$ and transverse to it. We have for $p\in N$,
\[
T_{f(p)}E=f_*(T_p N)\oplus \langle \xi_{p} \rangle.
 \]
 Furthermore, on $E$ we have a natural derivative $D$ due to its affine structure. Let $X,Y$ be vector fields on $N$ (so $f_*(X)$ and $f_*(Y)$ are tangent to $f(N)$). The next formulas are obtained splitting the left-hand side by means of the direction determined by $\xi$
\begin{align}
D_{f_*(X)} f_*(Y)&=f_*(\nabla_X Y)+h(X,Y) \xi, \qquad (Gauss) \label{gau}\\
D_{f_*(X)}\xi&=-f_*(S(X))+\tau(X) \xi.\qquad (Weingarten) \label{wei}
\end{align}
They define a torsion-less connection  $\nabla$, a symmetric bilinear form $h$ (the affine metric), an endomorphism $S$ of the tangent bundle $TN$ (the shape operator) and a one-form $\tau$ over $N$. These objects satisfy some differential equalities (Gauss, Codazzi) which the reader can find in \cite[Theor.\ 2.1]{nomizu94}.

Under a change of transverse field
\begin{equation} \label{onr}
\bar \xi=\phi \xi+f_*(Z)
\end{equation}
these objects change as follows\cite[Prop.\ 2.5]{nomizu94}
\begin{align}
\bar{h}&=\frac{1}{\phi} \, h, \label{dop}\\
\bar \nabla_XY&=\nabla_X Y-\frac{1}{\phi}\, h(X,Y) Z,\\
\bar{\tau}&=\tau+\frac{1}{\phi}\, h(Z, \cdot)+d \log \vert \phi\vert, \label{two}\\
\bar S&=\phi S-\nabla \! . \,Z+\bar \tau(\cdot )Z .
\end{align}

Observe that $h$ is definite if $f(N)$ is the boundary of a convex set. The change of transverse field redefines $h$ through multiplication by a conformal factor, thus the non-degeneracy of $h$ including the absolute value of its signature is really a property of $N$.
In what follows we shall assume that $N$ is non-degenerate. With some abuse of notation we shall often identify $N$ with $f(N)$ and $p$ with $f(p)$  in the next formulas. This is not source of confusion when $f$ is an embedding.

The affine metric induces a $n$-form $\omega_h$ on $N$. Let $\{e_i, i=1,\cdots,n\}$ be a basis of $T_pN$ such that $(\xi,f_*(e_1),\cdots,f_*(e_n))$ is $\omega$-positively oriented. Defined $h_{ij}=h(e_i,e_j)$ let
\[
\omega_h:=\sqrt{\vert \det h_{ij}\vert}\, \theta^1\wedge\cdots \wedge \theta^n ,
\]
where $\{\theta^i, i=1,\cdots,n\}$ is the dual basis of $\{e_i\}$.

 Blaschke has shown that it is possible to select a special transverse field on every non-degenerate hypersurface. The {\em Blaschke} or {\em affine} normal is determined up to a sign by the conditions
\begin{itemize}
\item[(i)] $\tau=0$, \qquad (equiaffine condition)
\item[(ii)] $\omega_h=f^*(i_\xi \omega)$.
\end{itemize}
If $h$ is definite the sign of $\xi$ is fixed so as to make $h$ positive definite. If $h$ is Lorentzian up to a sign, it is fixed in such a way that the signature is $(-,+,\cdots, +)$. Given the Blaschke normal the formulas of Gauss and Weingarten determine a Blaschke  metric, shape operator and torsion-less connection. The scalar $H=\frac{1}{n}\, \textrm{tr} S$ is called {\em affine mean curvature}. It can be shown that the equiaffine condition is equivalent to $\nabla [f^*(i_\xi \omega)]=0$ (see the next Prop.\ \ref{igf} or \cite[Prop.\ 1.4]{nomizu94}).

So far we have given a traditional introduction to affine differential geometry. Actually, it is interesting to notice that the affine normal can be defined already for the weaker structure given by $(E,\vert\omega\vert)$ where $\vert\omega\vert$ is a volume form rather than a $n$+1-form. It is sufficient to replace (ii) with
\begin{itemize}
\item[(ii')] $\vert\omega_h\vert=\vert f^*(i_\xi \omega)\vert$, \qquad (the affine volume equals the induced volume)
\end{itemize}
where $\omega$ is any local representative of $\vert\omega\vert$. In fact (ii) is not able, in any case, to fix the sign of $\xi$.

The Pick cubic form is a symmetric tensor on $N$ defined by
\begin{equation} \label{mic}
c(X,Y,Z)=\frac{1}{2}[(\nabla_X h)(Y,Z)+\tau(X) h(Y,Z)],
\end{equation}
where $X,Y,Z\in T_pN$ and where the symmetry follows from the Codazzi equations.
Actually, the usual definition  from affine differential geometry does not include the 1/2 factor. We included it for consistency with a related Finslerian definition (cf.\ Theor.\ \ref{rel}). Let $c^\sharp$ be the tensor on $N$ defined by $h(X,c^\sharp(Y,Z))=c(X,Y,Z)$, from Eq.\ (\ref{mic}) it follows that  the Levi-Civita connection of $h$ is given by\footnote{In the published version the last two terms are missing, a fact which does not affect the work.}
\begin{equation}
\nabla^h_X Y=\nabla_X Y+c^\sharp(X,Y)-\frac{1}{2} \tau(X) Y-\frac{1}{2} \tau(Y) X+\frac{1}{2} h(X,Y) \tau^\sharp.
\end{equation}
%The trace of $c^\sharp$ is the Tchebycheff form.
 We shall need
\begin{proposition} \label{igf}
On $N$ we have
\begin{align}
\nabla_X \omega_h&=\{\textrm{tr}  \,[Y\mapsto c^\sharp(X,Y)]-n \tau(X)/2 \}\, \omega_h,\\
\nabla_X f^*(i_\xi \omega)&= \tau(X) f^*(i_\xi \omega). \label{jdk}
\end{align}
\end{proposition}

\begin{proof}
Let $E_X$ be  $ Y \mapsto c^\sharp(X,Y)-\frac{1}{2} \tau(X) Y-\frac{1}{2} \tau(Y) X+\frac{1}{2} h(X,Y) \tau^\sharp$, we have
\begin{align*}
(\nabla_X \omega_h)&(Y_1,\cdots,Y_n)= \p_X \omega_h(Y_1,\cdots, Y_n)-\sum_i\omega_h(\cdots, \nabla_X Y_i,\cdots)\\
&=\p_X \{\omega_h(Y_1,\cdots, Y_n)\}-\sum_i\omega_h(\cdots, \nabla^h_X Y_i,\cdots)+\sum_i\omega_h(\cdots, E_X(Y_i),\cdots)\\
&=(\nabla^h_X \omega_h)(Y_1,\cdots,Y_n)+(\textrm{tr} E_X)\,  \omega_h(Y_1,\cdots, Y_n) .
\end{align*}
Concerning the second equation
\begin{align*}
\nabla_X f^*(i_\xi \omega)&=f^*(D_{f_*(X)} i_\xi \omega)=f^*( i_\xi D_{f_*(X)}\omega+i_{D_{f_*(X)}\xi} \omega).
\end{align*}
Since $\omega$ is translational invariant $D_{f_*(X)}\omega=0$, thus Eq.\ (\ref{jdk}) follows from (\ref{wei}). $\square$
\end{proof}

Observe that the equiaffine condition is preserved redefining $\xi \to \beta \xi$ where $\beta \ne 0$ is a constant while $\omega_h$ and $f^*(i_\xi \omega)$ can be made coincident with a suitable choice of $\beta$ provided they differ by a multiplicative constant. Furthermore, in the equiaffine case they differ by a multiplicative constant iff $\nabla \omega_h=0$ iff \begin{itemize}
\item[(iii)] $\textrm{tr} \,[Y\mapsto c^\sharp(\cdot,Y)]=0$. \qquad (apolarity condition)
\end{itemize}
Thus the transverse field is  Blaschke's up to a constant provided (i) and (iii) hold.

If the lines on $E$ generated by the Blaschke normals to $N$ meet at a point $o$, $N$ is said to be a proper {\em affine sphere} with center $o$, while if they are parallel it is an improper affine sphere \cite[Def.\ 3.3]{nomizu94}. If $S=H I$ then $H$ is constant over $N$ and $N$ is an affine sphere, proper if $H\ne 0$ and improper if $H=0$. The converse also holds: if $N$ is an affine sphere $S\propto I$. An affine sphere is called elliptic, parabolic or hyperbolic depending on the sign of $H$, respectively positive, zero or negative. For a proper affine sphere if $y\in N$ then $\xi(y)=-H(y-o)$.

Now suppose to have been given a pair $(M,\vert\omega\vert)$ where $M$ is a $n$+1-dimensional manifold and $\vert\omega\vert$ is a volume form. Then $(T_xM,\vert\omega\vert\vert_x)$ is a pair given by an affine space (actually a vector space) and a translational invariant volume form. Thus we can introduce, up to a sign, the affine normal to any non-degenerate immersions $N_x$ in $T_xM$. However, $T_xM$ is not an affine space but a vector space thus there is a point which plays a special role: the origin. The initial structure  $(M,\vert\omega\vert)$ naturally suggests to consider distributions $x\mapsto N_x$ of proper affine spheres with center the origin of the tangent spaces $T_xM$.

\begin{remark}
Let $\nabla$ and $h$ be the connection and affine metric induced by the Blaschke transverse field (one speaks of Blaschke structure), and let $Ric^\nabla$ be the Ricci tensor of the connection $\nabla$ on $N$. A characterization of the affine sphere condition is given by
\begin{proposition} \label{ein}
$N$ is an affine sphere if and only if the Blaschke structure satisfies $Ric^\nabla\propto h$, in which case $Ric^\nabla=H(n-1) h$, where $H$ is the affine mean curvature of the affine sphere.
\end{proposition}

We remark that the condition involved in this statement is not the usual Einstein condition since in general $\nabla$ does not coincide with $\nabla^h$.

\begin{proof}
For any Blaschke structure \cite[Prop.\ 3.4]{nomizu94}  \[Ric^\nabla(Y,Z)=\textrm{tr} S \,h(Y,Z)-h(S(Y),Z).\] The conclusion is easily reached upon taking the trace. $\square$
\end{proof}
\end{remark}

The next result, which will turn out to be useful in the next section, does not seem to have been previously noticed or stressed in the literature.
Let $m$ be the traceless part of the cubic form $c$ (where the trace is taken with $h$), namely
\begin{align*}
m(W,X,Y)&=c(W,X,Y)\\
& \quad -\frac{1}{n+2}\{\textrm{tr} c(W) h(X,Y)+\textrm{tr} c(Y) h(W,X)+\textrm{tr} c(X) h(Y,W)\},
\end{align*}
we have

\begin{theorem} \label{rai}
The tensor defined by $h(X,m^\sharp (Y,Z))=m(X,Y,Z)$ does not depend on the transverse field used to define $c$ and $h$. It coincides with the (one index raised) Pick cubic form $c^\sharp_B$ for the Blaschke normal.
\end{theorem}

\begin{proof}
Let us consider a change of transverse field parametrized as in (\ref{one}). Using the mentioned transformation rules and the corrected (Finslerian) definition of cubic form we arrive at
\begin{align*}
\bar c(W,X,Y)&=\frac{1}{\phi}\, c(W,X,Y)\\
&\quad +\frac{1}{2 \phi^2}\{h(Z,W) h(X,Y)+h(Z,Y) h(W,X)+h(Z,X) h(Y,W)\}
\end{align*}
taking the trace
\begin{equation} \label{thr}
\bar{\textrm{tr}}  \bar c (W)=\textrm{tr} c (W)+\frac{n+2}{2 \phi } h(Z,W)
\end{equation}
From here we arrive at $\bar{m}(W,X,Y)=\frac{1}{\phi} m(W,X,Y)$, and the first statement follows from (\ref{dop}). The Pick  cubic form for the Blaschke normal is traceless (apolarity condition), thus $m^\sharp$ coincides with $c_{B}^\sharp$. $\square$
\end{proof}

\subsection{Proper affine spheres} \label{pas}

Any embedding on a vector space which does not pass through the origin and which is  transverse to the position vector at every point is called   {\em centroaffine} \cite{nomizu94}.
It is instructive to prove the next known  result on centroaffine embeddings making particular attention to the role of the volume form. The parametrization of the affine sphere there introduced is due to Gigena \cite{gigena81,loftin10}. We shall see later on that ${\bm v}$ will be interpreted as observed velocity, while $u({\bm v})$ will be the observer Lagrangian of our theory.

\begin{theorem} \label{lpo}
Let $\{e_\alpha\}$ be a basis of $(V,\vert\omega\vert)$, let  $\rho=\vert\omega\vert(e_0,\cdots,e_n)>0$, and let $y^\alpha$ be the induced coordinates on the vector space $V$.
Let ${\bm v}$ be inhomogeneous projective coordinates on $\{y\in V: y^0>0\}$ so that \begin{equation} \label{inh}
y=(y^0,{\bm y})=-\frac{1}{u}\,(1,{\bm v}).
\end{equation}
Let $N$ be a centroaffine hypersurface with respect to the position vector with origin $p\in E$, then identifying $V$ with $E-p$,  let
\begin{equation} \label{ndu}
f\colon {\bm v} \to -\frac{1}{u({\bm v})}\,(1,{\bm v}) ,
\end{equation}
 be its local hypersurface immersion. Let $c\ne 0$, then relative to the transverse vector field $\xi:=-c y$ the affine metric is
\begin{equation} \label{dio}
h=h_{ij}\, d v^i d v^j=\frac{u_{ij}}{c u} \, d v^i d v^j,
\end{equation}
the connection coefficients are
%\begin{equation}
$(\nabla)^k_{ij}=-\frac{1}{u}(u_i \delta^k_j+u_j \delta^k_i)$,
%\end{equation}
the shape operator is $S=cI$, and $\tau=0$. The transverse field $\xi$ is the Blaschke normal and hence $N$ is a proper affine sphere with  affine mean curvature $H=c$  and center $p$ if and only if
\begin{equation} \label{mon}
\det u_{ij}=\epsilon \rho^2 \Big(\frac{H}{u}\Big)^{n+2}.
\end{equation}
where $\epsilon$ is the sign of the determinant of $h_{ij}$, i.e.\ the parity of the negative signature of $h$ (thus $\epsilon=1$ if $h$ is positive definite and $\epsilon=-1$ if Lorentzian). In particular, in the positive definite case
\begin{equation} \label{kod}
h_{ij}=\Big( \frac{\rho^2}{\det u_{ij}}\Big)^{\frac{1}{n+2}} \frac{1}{u^2}\, u_{ij}.
\end{equation}
\end{theorem}

\begin{remark}
If $E=T_xM$  and $p$ is not the origin of $E$ then $\{y^\alpha\}$ are not  canonical coordinates on the tangent bundle at $x$.
\end{remark}

\begin{proof}
Let us observe that (here $\tilde e_j$ is the canonical basis of $\mathbb{R}^n$)
\begin{equation} \label{pig}
f_{*}(\tilde e_j)=D_{f_{*}( \tilde e_j)} y=\p_j \{-\frac{1}{u({\bm v})}\,(1,{\bm v})\}=-\frac{u_j}{u} \, y-\frac{1}{u}\, e_j ,
\end{equation}
where $\p_j$ is a shorthand for $\p/\p {v^j}$. Thus
\begin{align*}
D_{f_{*}(\tilde e_i)} f_{*}( \tilde e_j)&=D_{f_{*}(\tilde e_i)}D_{f_{*}(\tilde e_j)} y=\p_i \p_j\{-\frac{1}{u({\bm v})}\,(1,{\bm v})\}\\
&=\frac{u_{ij}}{u}\, (-y)+\frac{u_j}{u^2}\, e_i+\frac{u_i}{u^2}\, e_j+2\frac{u_iu_j}{u^2}\, y\\
&=-\frac{u_i}{u}\,f_{*}(\tilde e_j)-\frac{u_j}{u}\,f_{*}(\tilde e_i)+ \frac{u_{ij}}{cu}\, (-cy).
\end{align*}
The first two terms are tangent to $N$, thus the last one gives the affine metric. From this same expression it is easy to read the connection coefficients. The statements concerning $S$ and $\tau$ are trivial since $D_{f_*(X)} y=f_*(X)$.

Now using Eq.\ (\ref{pig}) and $\xi=-cy$ we observe that
\begin{align*}
\vert f^*i_\xi \omega\vert (\tilde e_1,\cdots,\tilde e_n)&=\vert\omega\vert(\xi,f_*(\tilde e_1),\cdots, f_*(\tilde e_n))=\vert(-\frac{1}{u})^n\omega(\xi, e_1,\cdots, e_n)\vert\nonumber\\
&=\vert (-\frac{1}{u})^n(-c)\, \omega(y^0 e_0, e_1,\cdots, e_n)\vert\nonumber\\
&=\vert (-\frac{1}{u})^{n+1}(-c)\, \omega(e_0, e_1,\cdots, e_n)\vert =\rho\vert(-\frac{1}{u})^{n+1}(-c)\vert,
\end{align*}
while
\[
\vert \omega_h\vert (\tilde e_1,\cdots,\tilde e_n)= \sqrt{\vert\det h_{ij}\vert}=\sqrt{\epsilon\det h_{ij}},
\]
%where $\epsilon=1$ if $\tilde e_1\wedge\cdots\wedge \tilde e_n$ is positively $\omega_h$-oriented and $\epsilon=-1$ otherwise. But by definition of $\omega_h$, $\tilde e_1\wedge\cdots\wedge \tilde e_n$ is positively $\omega_h$-oriented if $\omega(\xi, f_*(\tilde e_1),\cdots,f_*(\tilde e_n))$ is positive thus $\epsilon=\textrm{sgn} [(-\frac{1}{u})^{n+1}(-c) ]$.
The vector $\xi$ is the Blaschke normal and hence $N$ is an affine sphere with affine mean curvature $H=c$ if and only if $\vert \omega_h\vert=\vert f^*i_\xi \omega\vert$ which reads $\det u_{ij}=\epsilon \rho^2(\frac{c}{u})^{n+2}$. $\square$
\end{proof}

\begin{remark}
Let us consider an affine sphere on $E$ with center $p$. Observe that
the rescaled affine sphere  $f\to f'=\lambda (f-p)+p$, $\lambda>0$ is determined by the function $u'=u/\lambda$ and from Eq.\ (\ref{mon}) it follows that it is still an affine sphere with  affine mean curvature $H'=\lambda^{-\frac{2n+2}{n+2}} H$.
 %$H'=\lambda^{-\frac{2}{n+2}} H$
 Without loss of generality we can study just affine spheres for which $H=-1,0,1$ since the others are obtained through rescaling. In the proper case  the transverse vector becomes either the position vector with origin $p$ or its opposite.
\end{remark}

\begin{remark} \label{naa}
Suppose that $N\subset E$ is an affine sphere with mean curvature $H$ for $(E,\vert \omega\vert)$ and suppose to change the volume form to $\vert \check \omega\vert=\alpha\vert \omega\vert$, $\alpha>0$. With the notation of Theorem (\ref{lpo}), $\check \rho=\alpha \rho$. Equation (\ref{mon}) clarifies that $N$ is still an affine sphere with mean curvature $\check H$ where
\begin{equation} \label{mpc}
\check \rho^2\check H^{n+2}=\rho^2 H^{n+2}.
\end{equation}
Thus the concept of elliptic or hyperbolic  affine sphere makes sense irrespective of the volume form and hence is well defined just on an affine space, while it is necessary to specify a volume form to talk about affine mean curvature $H$ of the affine sphere. In the proper case one can choose $\rho$ so as to get $\vert H\vert=1$.
\end{remark}

The next result is the crucial step which relates the affine sphere distributions with measures over $M$.

\begin{corollary} \label{kwx}
Given a proper affine sphere $N \subset E$ there is a unique translational invariant volume form on $E$ such that $\vert H\vert=1$.

Similarly, given a manifold $M$, and a point dependent distribution of proper affine spheres $x \mapsto N_x\subset T_xM$ (not necessarily centered at  the origin of $T_xM$), there is a unique volume form on $M$ for which the affine spheres satisfy $\vert H\vert=1$.
\end{corollary}

Of course the regularity of the dependence of $N_x$ on $x$ will be related to that of the volume form on $x$.

The metric (\ref{kod}) was first obtained by Loewner and Nirenberg \cite{loewner74} while searching for projective invariant metrics on convex sets. Let $p$ be the origin of $T_xM$.
Let $\{e_\alpha\}$ be any basis of $T_xM$ and let
%$\tilde e_\alpha=((A^{-1})^T)_{\alpha \beta} e_\beta$
$e_{\tilde \alpha}=A^\beta_{\ \tilde \alpha} e_\beta$ be another basis then the coordinates $(\tilde u, v^{\tilde i})$ are related to $(u,v^i)$ as follows (here $A^\beta_{\ \tilde \alpha} A^{\tilde \alpha}_{\ \gamma}=\delta^{\beta}_{\gamma}$, namely we are using a convention common in mathematical relativity in which the distinction between coordinates in made at the level of indices)
\begin{align*}
-\frac{1}{\tilde u} &=A^{\tilde 0}_{\ j}\Big(-\frac{ v^j}{ u}\Big)+A^{\tilde 0}_{\ 0} \Big(-\frac{1}{ u}\Big),\\
-\frac{ v^{\tilde i}}{\tilde u} &=A^{\tilde i}_{\ j}\Big(-\frac{v^j}{ u}\Big)+A^{\tilde i}_{\ 0} \Big(-\frac{1}{ u}\Big),
\end{align*}
which can be rewritten including also the transformation of the density under coordinate changes
\begin{align}
\tilde u &=  u \big[A^{\tilde 0}_{\ j}(x) v^j+A^{\tilde 0}_{\ 0}(x)\big]^{-1}, \label{pro1}\\
 v^{\tilde i} &=\big[A^{\tilde i}_{\ j}(x) v^j+A^{\tilde i}_{\ 0}(x)\big] \, \big[A^{\tilde 0}_{\ j}(x) v^j+A^{\tilde 0}_{\ 0}(x)\big]^{-1}, \label{pro2}\\
\tilde \rho&= \rho \, \big(\det A^{\tilde \alpha}_{\ \beta}(x)\big)^{-1}. \label{pro3}
\end{align}
In this expression we have  made explicit the dependence of the matrix $A$ on the point $x\in M$. If both frames are holonomic then $A^{\tilde \alpha}_{\ \beta}={\p x^{\tilde \alpha}}/{\p x^{ \beta}}$.

Since the metric (\ref{dio}) with $h_{ij}$ given by (\ref{kod}) and the Monge-Amp\`ere equation (\ref{mon}) hold irrespective of the starting basis $\{e_\alpha\}$ chosen, these objects are invariant under the projective changes (\ref{pro1})-(\ref{pro3}).

\subsection{Improper affine spheres}

In this section we introduce convenient coordinates for improper affine spheres \cite{loftin10}. They are chosen so as to simplify the Monge-Amp\`ere equation which describes these hypersurfaces.
 Let us notice that a (connected) hypersurface $N\subset E$ which is transverse to a direction $e_0\in V$ is a graph over a hyperplane transverse to $e_0$. We have

\begin{theorem} \label{lpp}
Let $\{e_\alpha\}$ be a basis of $(V,\vert\omega\vert)$, let  $\rho=\vert\omega\vert(e_0,\cdots,e_n)>0$,  let $y^\alpha$ be the induced coordinates on the vector space $V$ and let us denote  ${\bm v}={\bm y}$.
Let $N$ be a hypersurface  which is a graph $f\colon {\bm v} \to (-u({\bm v}),{\bm v})$ over the hyperplane $y^0=0$.
Let $c\ne 0$, then relative to the transverse vector field $\xi:=-c (1,{\bm 0})$ the affine metric is
\begin{equation} \label{nur}
h=h_{ij}\, d v^i d v^j=\frac{1}{c} \,u_{ij} \, d v^i d v^j,
\end{equation}
the connection coefficients are $(\nabla)^k_{ij}=0$ and $S=0,\tau=0$. The transverse field $\xi$ is the Blaschke normal and hence $N$ is an improper affine sphere if and only if
\begin{equation} \label{moo}
\det u_{ij}=\epsilon \rho^2 c^{n+2},
\end{equation}
where $\epsilon$ is the sign of the determinant of $h_{ij}$ (the parity of the negative signature of $h$).
\end{theorem}

\begin{proof}
Let us observe that (here $\tilde e_j$ is the canonical basis of $\mathbb{R}^n$)
\begin{equation} \label{pih}
f_{*}(\tilde e_j)=D_{f_{*}( \tilde e_j)} y=\p_j (-u({\bm v}),{\bm v})=-u_j e_0+ e_j ,
\end{equation}
where $\p_j$ is a shorthand for $\p/\p {v^j}$. Thus
\begin{align*}
D_{f_{*}(\tilde e_i)} f_{*}( \tilde e_j)&=D_{f_{*}(\tilde e_i)}D_{f_{*}(\tilde e_j)} y=\p_i \p_j(-u,{\bm v})=-u_{ij} \,e_0=\frac{1}{c}u_{ij} \,\xi.
\end{align*}
There is no term tangent to $N$ thus  the connection coefficients vanish. The statements concerning $S$ and $\tau$ are trivial since $D_{f_*(X)} e_0=0$.

Now using Eq.\ (\ref{pih}) and $\xi=-c e_0$ we observe that
\begin{align*}
\vert f^*i_\xi \omega\vert (\tilde e_1,\cdots,\tilde e_n)&=\vert \omega\vert(\xi,f_*(\tilde e_1),\cdots, f_*(\tilde e_n))=\vert c \omega(e_0, e_1,\cdots, e_n)\vert=\vert c\rho\vert,
\end{align*}
while
\[
\vert \omega_h\vert (\tilde e_1,\cdots,\tilde e_n)= \sqrt{\vert\det h_{ij}\vert}=\sqrt{\epsilon \det h_{ij}},
\]
%where $\epsilon=1$ if $\tilde e_1\wedge\cdots\wedge \tilde e_n$ is positively $\omega_h$-oriented and $\epsilon=-1$ otherwise. But by definition of $\omega_h$, $\tilde e_1\wedge\cdots\wedge \tilde e_n$ is positively $\omega_h$-oriented if $\omega(\xi, f_*(\tilde e_1),\cdots,f_*(\tilde e_n))$ is positive thus $\epsilon=\textrm{sgn} [(-\frac{1}{u})^{n+1}(-c) ]$.
The vector $\xi$ is the Blaschke normal and hence $N$ is an improper affine sphere  if and only if $\vert \omega_h\vert=\vert f^*i_\xi \omega\vert$ which reads $\det u_{ij}=\epsilon \rho^2 c^{n+2}$. $\square$
\end{proof}

\subsection{Centroaffine embeddings and Finsler indicatrices}

In this section we obtain some results on the relationship between the Finsler metric at a given point $x\in M$ and the affine metric of the indicatrix. %$\mathscr{I}^-_x\subset \Omega_x$.

Preliminarly, let us observe that the indicatrix %$\mathscr{I}^-$
is a centroaffine hypersurface with respect to the origin of $T_xM$ because it is transverse to the position vector $y$.
%Let $\mathscr{L}$ be a Finsler Lagrangian restricted to $T_xM$. The indicatrices $\mathscr{I}^+$ and $\mathscr{I}^-$  are centroaffine hypersurfaces because they are transverse to the position vector $y$.
Indeed, by positive homogeneity
\[
\dd \mathscr{L}(y)=\frac{\p \mathscr{L}}{\p y^\alpha }\, y^\alpha=2\mathscr{L}=- 1.
\]
%
%Contracting with $\dd \mathscr{L}$ and using $\frac{\p \mathscr{L}}{\p v^\alpha} \, v^\alpha=2\mathscr{L}$ and
%\[
%\dd \mathscr{L}(D_X Y)= \frac{\p \mathscr{L}}{\p v^\alpha} (X^\beta \frac{\p Y^\alpha}{\p v^\beta})= X^\beta \frac{\p }{\p v^\beta} (\frac{\p \mathscr{L}}{\p v^\alpha} Y^\alpha)-X^\beta Y^\alpha  \frac{\p^2 \mathscr{L}}{\p v^\alpha \p v^\beta}=-g_v(X,Y)
%\]
%we arrive at
%\begin{equation}
%g_{v}(X,Y)=h(X,Y) 2 \mathscr{L}.
%\end{equation}

For the first statement of the next theorem see also \cite{laugwitz11,beem70}, for the second statement see also \cite[Prop.\ 4.1]{mo10}.

%Let us denote with $y:=y^\alpha \p/\p y^\alpha$ the Liouville vector field on $TM$.
\begin{theorem} \label{rel}
%Let $\mathscr{I}^\pm_x=\mathscr{L}_x^{-1}(\pm 1/2)$.
The vertical Finsler metric  induces on the
indicatrix $\mathscr{I}^-$ a metric $h$ which coincides with the affine metric  for the transverse field $\xi= y$. Thus, on $\Omega_x$ Eqs. (\ref{dqa}) and (\ref{war}) hold
\begin{equation} \label{bcp}
g=- \dd F^2+F^2 h, \qquad h=-F^{-1} \dd^2 F, \qquad F=\sqrt{2 \vert\mathscr{L}\vert} .
%\frac{1}{2 \mathscr{L}}\,(\dd \mathscr{L})^2\oplus (\pm 2 \mathscr{L}) h ,
\end{equation}
%that is $h=\pm F^{-1} H F$ where $F=\sqrt{\pm 2 \mathscr{L}}$.

The Pick cubic form for the transverse  field $\xi$ is the restriction to the indicatrix of the Cartan torsion, that is $c=f^*C$. The Pick cubic form for the Blaschke transverse  field  is the restriction to the indicatrix of the traceless part of the Cartan torsion: $c_B^\sharp=m^\sharp=(f^*M)^\sharp$.

The indicatrix $\mathscr{I}^-_x$ is an affine sphere with  center at the origin of $T_xM$  iff the mean Cartan torsion vanishes on it (and hence on $\Omega_x$). In this case  with respect to the translational invariant volume form $\vert\omega\vert=\sqrt{\vert \det g_{\alpha \beta}\vert}\, d^{n+1} y$ (i.e.\ $\rho=\sqrt{\vert \det g_{\alpha \beta}\vert}$) the affine mean curvature of the indicatrix is such that  $H= - 1$, thus in projective coordinates
\begin{equation} \label{hhw}
\det u_{ij}= \vert \det g_{\alpha \beta}\vert \Big(-\frac{1}{u}\Big)^{n+2}.
\end{equation}
\end{theorem}

%Observe that if $g$ is Lorentzian the position vector $y$ is spacelike on $\mathscr{I}_x^+$ and timelike on $\mathscr{I}_x^-$,  thus Eq.\ (\ref{bcp}) implies that $h$ is Lorentzian  on $\mathscr{I}_x^+$ and positive definite on  $\mathscr{I}_x^-$.
%Notice that $c$ and $C$ provide the same information since $C(y,\cdot,\cdot)=0$.
%Moreover,
Observe that a zero mean Cartan torsion not only makes the indicatrix an affine sphere but also, by Eq.\ (\ref{aqs}) makes $\vert\omega\vert$ translational invariant and hence makes it possible to ask for the affine mean curvature of the affine sphere with respect to this volume form. Also notice that $c$ and $C$ provide the same information since $C(y,\cdot,\cdot)=0$.

\begin{remark} Suppose that the Finsler Lagrangian is defined over the whole slit tangent bundle.  Then a completely analogous theorem could be given for the spacetime indicatrix $\mathscr{I}^+$ where, however, $\xi=-y$, $g= \dd F^2+F^2 h$, $h= F^{-1} \dd^2 F$, the affine metric would be Lorentzian and the  affine mean curvature would be $H=1$.
\end{remark}

\begin{proof}
Let us contract $\dd \mathscr{L}$ with the Gauss equation $D_XY=\nabla_XY+h(X,Y) \xi$ where $X,Y \in T\mathscr{I}^-_x$ (here with some abuse of notation we identify $X$ with  $f_*(X)$ where $f$ is the embedding of the indicatrix) and use positive homogeneity
\[
\dd \mathscr{L}(D_X Y)= \!\frac{\p \mathscr{L}}{\p y^\alpha} \left(\!X^\beta \frac{\p Y^\alpha}{\p y^\beta}\!\right)\!= \!X^\beta \frac{\p }{\p y^\beta} \left(\frac{\p \mathscr{L}}{\p y^\alpha} Y^\alpha\!\!\right)-X^\beta Y^\alpha\!  \frac{\p^2 \mathscr{L}}{\p y^\alpha \p y^\beta}\!=-g_y(X,Y),
\]
thus
\[
-g_y(X,Y)=\dd \mathscr{L}(D_X Y)= h(X,Y)\dd \mathscr{L}( y)= 2 \mathscr{L}(x,y) h(X,Y)=-h(X,Y).
\]
 This calculation proves the first statement.
By positive homogeneity the indicatrix is $g$-orthogonal to $\xi$ indeed if $X\in T\mathscr{I}_x^-$, $g(X,\xi)= \dd \mathscr{L}(X)=0$. The equations  (\ref{dqa}) and (\ref{war}) follow from the just  established equality between  the affine metric and the induced metric (hence the same symbol $h$).

% The second order positive homogeneity of $g$ together with $g(y,y)=2\mathscr{L}$ gives Eq.\ (\ref{bcp}).

Recalling  that the induced metric is the affine metric we have for every $X,Y,Z\in T\mathscr{I}_x^-$
\begin{align*}
\nabla_Z h(X,Y)&=\nabla_Z[h(X,Y)]-h(\nabla_Z X,Y)-h(X,\nabla_Z Y)\\
&=D_Z[g(X,Y)]-g(\nabla_Z X,Y)-g(X,\nabla_Z Y)\\
&=D_Z[g(X,Y)]-g(D_Z X-h(X,Z)\xi,Y)-g(X,D_Z Y-h(X,Z)\xi)\\
&=D_Z[g(X,Y)]-g(D_Z X,Y)-h(X,D_Z Y)=(D_Z g)(X,Y)
\end{align*}
Since the immersion is centroaffine, $\tau=0$, thus we have the equality between Pick cubic form and pullbacked Cartan torsion (observe that the 1/2 factor must be present in Eq.\ (\ref{mic}) since it is included in Eq.\ (\ref{mac})).

Let $y\in \mathscr{I}^-_x$ and let $\{X_i, i=1,\cdots n\}$ be a $h$-orthogonal basis at $T_y \mathscr{I}^-_x$, then since $g_y(y,y)=2\mathscr{L}(x,y)=- 1$, $g_y(y,X_i)=0$,
\begin{align*}
\textrm{tr} \,[Y\mapsto c^\sharp(Z,Y)]&=\sum_i c(Z,X_i,X_i)=(g_y(y,y))^{-1} C(Z,y,y)+\sum_i C(Z,X_i,X_i)\\
&=f^*(\textrm{tr}_g C)(Z)=I(f_*(Z)) ,
\end{align*}
where $I$ is the mean Cartan torsion. From here the traceless part of the Pick cubic form, namely the Pick cubic form for the Blaschke normal c.f.\ Theor.\ \ref{rai}, is easily inferred to be the pullback of the traceless part of the Cartan torsion.

If the mean Cartan torsion vanishes then the apolarity condition holds thus the transverse field $y$ is Blaschke's up to a constant. But these normals generate lines which meet at the origin of $T_xM$ thus $\mathscr{I}^-$ is an affine sphere. Conversely, if $\mathscr{I}^-$ is an affine sphere with center the origin of $T_xM$ then $\xi=-H y$, $\textrm{sgn} (H)=- 1$, where the value of $H$ depends on the choice of volume form (Remark \ref{naa}). Since $ y$ coincides with the Blaschke normal up to a constant, the apolarity condition holds. As for every $Z$, $I(f_*(Z))=0$ and $I(y)=0$ we have $I=0$.

Now, suppose that the volume form is $\sqrt{\vert \det g_{\alpha \beta}\vert}\, d^{n+1} y$. Let $\{\chi^i,i=1,\cdots,n\}$ be a coordinate system on $\mathscr{I}^-$ and let $\chi^0=\mathscr{L}$. These definitions determine a coordinate system $\chi^{\alpha'}$ on the cone generated by $\mathscr{I}^-$ in such a way that the lines $\chi^i= cnst$ pass through the origin.  The position vector on the indicatrix reads $y=- \p/\p \chi^0$ and Eq.\ (\ref{bcp}) reads
\begin{equation} \label{jgk}
g_{\alpha ' \beta'} d \chi^{\alpha'} d \chi^{\beta'}=\frac{1}{2 \chi^0} (d \chi^0)^2- 2 \chi^0 h_{ij} d \chi^i d \chi^j.
 \end{equation}
  Thus $\omega= \sqrt{\vert \det g_{\alpha' \beta'}\vert}\, d \chi^0\wedge d \chi^1 \wedge \cdots \wedge d \chi^n$, and since $\{\xi, \p/\p \chi_1, \cdots, \p/\p \chi_n\}$ has the $\omega$-orientation given by $\textrm{sgn}( H)=-1$, we have
  \[
  \omega_h= -\sqrt{\vert \det h_{ij}\vert} \, d \chi^1 \wedge \cdots \wedge d \chi^n ,
  \]
     and finally
\[
i_\xi\omega=i_{-Hy}\omega= H \sqrt{\vert \det g_{\alpha' \beta'}\vert}\, d \chi^1 \wedge \cdots \wedge d \chi^n= - H  \vert2 \chi^0\vert^{\frac{n-1}{2}} \omega_h.
\]
Since on the indicatrix $\vert2 \chi^0\vert=1$ we conclude that $H=- 1$. $\square$
\end{proof}

%\begin{corollary} \label{omw}
%The Blaschke affine metric of $\mathscr{I}^-$ is Einstein if and only if $\mathscr{I}^-$ is an affine sphere  if and only if $I_\alpha=0$.
%\end{corollary}

%\begin{proof}
%The last equivalence has been already proved. The former equivalence follows from Prop.\ \ref{ein}.
%\end{proof}

\begin{remark} \label{njs}
Obviously Theorem \ref{rel} admits a  reformulation for positive definite $g$, it is sufficient to take the transverse field $\xi=-y$.
\end{remark}

We have established that the condition $I_\alpha=0$ characterizes those (Lorentz-)Finsler spaces for which the indicatrix is an affine sphere centered at the origin. One might ask what is the Finslerian characterization of an indicatrix which is an affine sphere arbitrarily centered. This question is answered by the next theorem

\begin{theorem} \label{doi}
Let $\{y^\alpha\}$ be canonical tangent coordinates on $T_xM$, where $(M,\mathscr{L})$ is a Lorentz-Finsler space and $\bar \Omega$ is the cone domain of $\mathscr{L}$. The indicatrix is an affine sphere (necessarily hyperbolic) centered at $p\in T_xM$, if and only if the mean Cartan torsion has the form
\begin{equation}
\begin{split}
I_\alpha&=\frac{n+2}{2} \sqrt{-g_y(y,y)} \left(1+\frac{g_y(y,p)}{\sqrt{-g_y(y,y)}}\right)^{\!-1} h_{\alpha \beta} p^\beta\\
&=\frac{n+2}{2} \frac{\p}{\p y^\alpha} \log \left(1-\frac{\p F}{\p y^\beta} p^\beta\right), \ with \ F=\sqrt{-g_y(y,y)} \label{or1}
\end{split}
\end{equation}
%there is a function $\phi\colon \Omega_x\to \mathbb{R}^+$ such that $I_\alpha=\frac{n+1}{2} \phi^{-1} h_{\alpha \beta} p^\beta$.
with $p^\beta$ independent of $y$. Let $C\subset E$ be the cone generated by the convex hull of the affine sphere with its center $p$. The  {\em vector} $p$ belongs to $\bar C-p$ and the domain of the Finsler Lagrangian is $\bar \Omega=\bar C-p$. Finally, let $\mu$ be the translational invariant volume form which assigns to the indicatrix the affine mean curvature $H=-1$, then
\[
\sqrt{\vert \det g_{\alpha \beta}\vert} \, \dd^{n+1} y= \left(1-\frac{\p F}{\p y^\beta} p^\beta\right)^{\!\frac{n+2}{2}} \!\!\mu,
\]
thus the $y$ dependence is all on the first factor.

%The equation (\ref{or1})  can also be satisfied in the hyperbolic affine sphere minus a point provided the space is Finslerian and the origin of $T_xM$ belongs to the indicatrix.

%Let $C\subset E$ be the cone generated by the convex hull of the affine sphere with its center $p$. The indicatrix can be complete and hence coinciding with the affine sphere only if the {\em vector} $p$ belongs to $\bar C-p$ in which case the domain of the Finsler Lagrangian is $\bar \Omega=\bar C-p$.
\end{theorem}

Observe that $p$ becomes a causal vector field over $M$ if the dependence on $x$ is considered. It can be called the {\em center vector field}. If timelike it selects a privileged observer on spacetime.

For Finsler spaces this result changes as follows.
\begin{theorem} \label{doo}
Let $\{y^\alpha\}$ be canonical tangent coordinates on $T_xM$, where $(M,\mathscr{L})$ is a Finsler space and where $\mathscr{L}$ has domain $T_xM \backslash 0$. The indicatrix is an affine sphere (necessarily elliptic hence an ellipsoid) centered at $p\in T_xM$, if and only if the mean Cartan torsion has the form
\begin{equation}
\begin{split}
I_\alpha&=-\frac{n+2}{2} \sqrt{g_y(y,y)} \left(1-\frac{g_y(y,p)}{\sqrt{g_y(y,y)}}\right)^{\!-1} h_{\alpha \beta} p^\beta\\
&=\frac{n+2}{2} \frac{\p}{\p y^\alpha} \log \left(1-\frac{\p F}{\p y^\beta} p^\beta\right), \ with \ F=\sqrt{g_y(y,y)} \label{or2}
\end{split}
\end{equation}
with $p^\beta$ independent of $y$. Let $C\subset E$ be the ellipsoid generated by the convex hull of the affine sphere with its center $p$. The  {\em vector} $p$ belongs to $\bar C-p$. Finally, let $\mu$ be the translational invariant volume form which assigns to the indicatrix the affine mean curvature $H=1$, then
\[
\sqrt{ \det g_{\alpha \beta}} \, \dd^{n+1} y= \left(1-\frac{\p F}{\p y^\beta} p^\beta\right)^{\!\frac{n+2}{2}} \!\!\mu,
\]
thus the $y$ dependence is all on the first factor.
\end{theorem}

Observe that $p$ becomes a vector field over $M$ if the dependence on $x$ is considered. In Eq.\ (\ref{or1}) we used Eq.\ (\ref{dqa}), while in Eq.\ (\ref{or2}) we used the analogous equation $h=F^{-1} \dd ^2 F$ which is valid for Finsler spaces (cf.\ Remark \ref{njs}).

\begin{remark}
For Finsler spaces the indicatrix can be an affine sphere (elliptic, parabolic or hyperbolic) in other ways if the domain of $\mathscr{L}$ is  a half space minus an open cone. This happens when the affine sphere passes through the origin of $T_xM$. The parabolic case will be considered in Sect.\ \ref{aae}. The elliptic case gives the Kropina Finsler spaces while the hyperbolic case gives yet another Finsler Lagrangian.
\end{remark}

%Of course the most interesting case is the Lorentzian one for in the other the indicatrix would be a non-centered ellipsoid.

\begin{proof}
We shall denote with $y$ a point in the indicatrix and we shall give the proof in the Lorentz-Finsler case, the Finsler case being obtained analogously for transverse fields $\xi=-y$, $\bar \xi=-(y-p)$.

Necessity. We make the change of transverse field from  $\xi=y$ to $\bar \xi=y-p$,  and we  parametrize it as in (\ref{one})
 \begin{equation*}
\bar \xi=\phi \xi+f_*(Z)
\end{equation*}
Both transverse fields are centroaffine thus the equiaffine condition holds $\tau=\bar \tau=0$. By assumption the indicatrix is an affine sphere with center $p$, so $\bar \xi$ is, up to a constant, the affine normal, hence the apolarity condition holds: $\bar{\textrm{tr}} \bar c=0$. From Eq.\ (\ref{thr}) we have
\begin{align*}
0&=\textrm{tr}  \, c (W)+\frac{n+2}{2 \phi } h(Z,W)
\end{align*}
which implies, denoting again with $h$ the restriction of the metric $g$ to the indicatrix, and using (\ref{one})
\begin{align*}
0&=\textrm{tr} C (f_*(W))+\frac{n+2}{2 \phi } h(f_*(Z),f_*(W))\\
&=\textrm{tr} C (f_*(W))+\frac{n+2}{2 \phi } h(\bar \xi-\phi \xi,f_*(W))\\
&=\textrm{tr} C (f_*(W))-\frac{n+2}{2 \phi } h(p,f_*(W)),
\end{align*}
which, given the arbitrariness of $W$, proves the  equation $I_\alpha=\frac{n+2}{2} \phi^{-1} h_{\alpha \beta} p^\beta$ with $\phi$ to be determined. Now observe that $f_*(Z)=\bar \xi-\phi \xi=y-p-\phi y$ is tangent to the indicatrix at $y$, thus $g_y(f_*(Z),y)=0$, which implies  by positive homogeneity $\phi=\frac{1}{\sqrt{-g_y(y,y)}}\big(1+\frac{g_y(y,p)}{\sqrt{-g_y(y,y)}}\big)$ (recall that on the indicatrix $g_y(y,y)=-1$ and $h_{\alpha \beta}$ is positive homogeneous of degree $-2$ cf.\ Eq.\ (\ref{dqa})).

Sufficiency. Define  $\phi=\frac{1}{\sqrt{-g_y(y,y)}}\big(1+\frac{g_y(y,p)}{\sqrt{-g_y(y,y)}}\big)$, $\bar \xi=y-p$, so that the equation reads  $I_\alpha=\frac{n+2}{2} \phi^{-1} h_{\alpha \beta} p^\beta$, next observe that for $y$ on the indicatrix the definition of $\phi$ can be recasted in the form $g_y(\bar \xi-\phi \xi,y)=0$ which allows us to define a vector field $Z$ over the indicatrix so that $f_*(Z)=\bar \xi-\phi \xi$. For every $W$ we have $0=\textrm{tr} C (f_*(W))-\frac{n+2}{2 \phi } h(p,f_*(W))$ and we can repeat the previous steps backwards till  $0= \textrm{tr} c (W)+\frac{n+2}{2 \phi } h(Z,W)$ which shows by Eq.\ (\ref{one}) that $\bar{\textrm{tr}} \bar c=0$, namely the transvese field $\bar \xi$ is centroaffine and satisfies the apolarity condition, thus the indicatrix is an affine sphere centered at $p$.
%
%Finally,  the Finslerian reverse Cauchy-Schwarz inequality \cite[Theor. 3]{minguzzi13c}
%implies that for every $y\in \Omega$ and $w\in \bar{\Omega}\backslash 0$, $-g_y(y,w) \ge  \sqrt{-g_y(y,y)} \sqrt{-g_w(w,w)}$,
%%with equality only if they are proportional,
%thus if $p\in -\bar{\Omega}$ we have $g_y(y,p)/\sqrt{-g_y(y,y)}> 0$ and the parenthesis in Eq.\ (\ref{or1}) is well defined for every $y$ on the indicatrix.
%
%If $p\in T_xM\backslash 0$ does not belong to $\bar{\Omega}\cup (-\bar{\Omega})$ then there is $w\in \Omega$ such that $g_w(w,w)=-1$, $g_w(w,p)=0$ it belongs to a hyperplane passing through the origin which

Suppose that the {\em vector} $p$ belongs to $\bar C-p$, then the domain of the Lagrangian is $\bar \Omega=\bar C-p$. For every $y\in \Omega$, we have $\frac{g_y(y,p)}{\sqrt{-g_y(y,y)}}=g_{\hat y}(\hat y, p)$ with $\hat y$ belonging to the indicatrix, moreover, the locus $\{w: g_{\hat y}(\hat y, w)> -1\}$ is the half space which includes the origin and is bounded by the hyperplane tangent to the indicatrix at $\hat y$. For every $\hat y$ this region includes the origin $p$ of the affine sphere thus $1+\frac{g_y(y,p)}{\sqrt{-g_y(y,y)}}>0$ and the  parenthesis in Eq.\ (\ref{or1}) is well defined for every $y\in \Omega$.

%If the {\em vector} $p$ belongs to $\bar C-p$ then the domain of the Lagrangian is $\bar \Omega=\bar C-p$ and since the indicatrix does not separate $p$ and the origin of $T_xM$,  we have $-1<g_p(p,p)\le 0$.  The Finslerian reverse Cauchy-Schwarz inequality \cite[Theor. 3]{minguzzi13c}  $-g_y(y,p) \ge  \sqrt{-g_y(y,y)} \sqrt{-g_p(p,p)}$, implies $1+\frac{g_y(y,p)}{\sqrt{-g_y(y,y)}} \le 1-\sqrt{-g_p(p,p)}>0$, thus the parenthesis in Eq.\ (\ref{or1}) is well defined for every $y\in \Omega$.

If $p \notin \bar C-p$ then it is easy to see that there is a half line starting from $0$ and tangent to the affine sphere, which means that the affine sphere cannot be used in its entirety to define a Finsler indicatrix. Furthermore, observe that if 0 stays on the indicatrix or on the opposite side of the indicatrix compared to $p$ then the space would be Finsler rather than Lorentz-Finsler.

The statement on the volume forms is proved as follows. Let $\{\chi^i\}$ be coordinates on the indicatrix. We extend them in two different ways. First we impose that the sets $\chi^i=cnst$ are half lines passing through the origin of $T_xM$, and add to the set the coordinate $\chi^0=\mathscr{L}$ so as to coordinatize $\Omega$. The Finsler metric reads (cf.\ Eq.\ (\ref{jgk}))
\begin{equation} \label{jgk2}
g_{\alpha ' \beta'} d \chi^{\alpha'} d \chi^{\beta'}=\frac{1}{2 \chi^0} (d \chi^0)^2- 2 \chi^0 h_{ij} d \chi^i d \chi^j.
 \end{equation}
Let $\xi=-\frac{\p}{\p \chi^0}=y$ be the Finslerian centroaffine transverse field. We have
\begin{equation} \label{uod}
f^*(i_\xi \sqrt{\vert \det g\vert}\, \dd \chi^0\wedge \dd \chi^1\wedge \cdots\dd \chi^n )=-\sqrt{\det h_{ij}} \,\dd \chi^1\wedge \cdots\dd \chi^n.
\end{equation}
Let us consider the similar equations that are obtained if the origin of $T_xM$ is moved on $p$. The coordinates $\{\chi^i\}$ are extended to coordinates $\chi^{\bar i}$ in such a way that their level sets are half lines originating from $p$. Adding as a further coordinate $\chi^{\bar 0}:=\bar{\mathscr{L}}$ gives a coordinatization of $C$. The barred Finsler metric reads
\begin{equation} \label{jgk3}
\bar g_{\bar \alpha  \bar \beta} d \chi^{\bar \alpha} d \chi^{\bar \beta}=\frac{1}{2 \chi^{\bar 0}} (d \chi^{\bar 0})^2- 2 \chi^{\bar 0} \bar h_{\bar i \bar j} d \chi^{\bar i} d \chi^{\bar j}.
 \end{equation}
 The affine normal is  $\bar \xi:=-\frac{\p}{\p \chi^{\bar 0}}=y-p$ thus at every point of the indicatrix we can find a tangent vector $Z$ such that $\bar \xi=\phi\xi+f_*(Z)$.
 We know from Theor.\ \ref{rel} that the $n+1$ form induced by $\bar g$ is such that \[
f^*( i_{\bar \xi} \sqrt{\vert \det \bar g_{\bar \alpha \bar \beta}\vert} \, \dd \chi^{\bar 0}\wedge\dd \chi^{\bar 1}\wedge\cdots \wedge \dd \chi^{\bar n})=\omega_{\bar h}=-\sqrt{\det \bar h_{\bar i\bar j}} \, \dd \chi^{\bar 1}\wedge\cdots \wedge \dd \chi^{\bar n}.
 \]
However, on the indicatrix the coordinate $\chi^i$ and $\chi^{\bar i}$ coincide and moreover $\bar h=h/\phi$, cf.\ Eq.\ (\ref{dop}), thus
\[
f^*( \phi i_{ \xi} \sqrt{\vert \det \bar g_{\bar \alpha \bar \beta}\vert} \, \dd \chi^{\bar 0}\wedge\dd \chi^{\bar 1}\wedge\cdots \wedge \dd \chi^{\bar n})=\omega_{\bar h}=-\phi^{-n/2} \sqrt{\det h_{ i j}} \, \dd \chi^{ 1}\wedge\cdots \wedge \dd \chi^{ n}.
\]
Let $\varphi$ be the ($y$-dependent but necessarily positive homogeneous of degree zero) factor such that in the canonical coordinates of the tangent bundle $\sqrt{\vert \det g\vert}= \varphi\sqrt{\vert \det \bar g\vert}$ or equivalently
\[
\sqrt{\vert \det  g_{ \alpha  \beta}\vert} \dd \chi^{ 0}\wedge\dd \chi^{ 1}\wedge\cdots \wedge \dd \chi^{ n}=\varphi \sqrt{\vert \det \bar g_{\bar \alpha \bar \beta}\vert} \, \dd \chi^{\bar 0}\wedge\dd \chi^{\bar 1}\wedge\cdots \wedge \dd \chi^{\bar n}
\]
 then
 \[
f^*(\frac{1}{\varphi} \phi^{\frac{n+2}{2}}) f^*(  i_{ \xi} \sqrt{\vert \det  g_{ \alpha  \beta}\vert} \dd \chi^{ 0}\wedge\dd \chi^{ 1}\wedge\cdots \wedge \dd \chi^{ n})=- \sqrt{\det h_{ i j}} \, \dd \chi^{ 1}\wedge\cdots \wedge \dd \chi^{ n},
\]
which using Eq.\ (\ref{uod}) proves the claim. $\square$
\end{proof}

\begin{corollary}
A Lorentz-Finsler space has an affine sphere indicatrix at $T_xM$ if and only if the Finsler Lagrangian $\mathscr{L}\colon \Omega \to \mathbb{R}$ satisfies the  vertical Monge-Amp\`ere equation
\begin{equation}\label{swi}
 -\det \dd^2 \mathscr{L}= \rho^2 \left(1+\frac{1}{\sqrt{- 2\mathscr{L}}}\,  \p_p \mathscr{L}\right)^{n+2}
\end{equation}
where $\rho \,\dd^{n+1} y$ is the translational invariant measure which assigns to the indicatrix the affine mean curvature $H=- 1$ and where $p\in \bar \Omega$ is the center of the sphere.
\end{corollary}

A similar but less interesting version holds for Finsler spaces. In Eq.\ (\ref{swi}) the minus signs have to be changed  to plus signs, and the plus signs have to be changed to minus signs.

%\begin{remark}
%In the next sections the omission of the $\pm$ superscript sign from  some mathematical objects will always  refer to the {\em future cone}. Thus we shall often replace $I^+_x \to I_x$, $\mathscr{I}^-_x\to \mathscr{I}_x$, $I^+(x)=I(x)$, $({I}^-_x)^*\to {I}^*_x$, $(\mathscr{I}^+_x)^*\to \mathscr{I}^*_x$ (see Sect.\ \ref{leg} for the definition of the last objects).
%\end{remark}

\subsection{Obtaining Finslerian results from affine differential geometry}

As mentioned previously, Theorem \ref{rel} admits a  reformulation for positive definite $g$. Since the affine metric and the Pick cubic form are the pullbacks of the Finsler metric and the Cartan torsion respectively, it is possible to obtain several results in (Lorentz-)Finsler geometry from results of affine differential geometry, and conversely.

For instance, the Maschke-Pick-Berwald theorem \cite[Theor.\ 4.5]{nomizu94} states that if the Pick cubic form $c_B$ of the Blaschke structure   vanishes, then the hypersurface (indicatrix) lies in a hyperquadric. From Theorem \ref{rai} this means that if the traceless part of the Pick cubic form $c$ vanishes (independently of the transverse field used) then the indicatrix lies in a hyperquadric (see also \cite[p.43, Lemma 3.2]{simon68}  \cite[Theor.\ 6.4]{nomizu94}). Using Theorem \ref{rai} this result can be translated to (Lorentz-)Finsler geometry as
\begin{theorem} \label{hqp}
Let $M$ be a (Lorentz-)Finsler space of dimension $n+1\ge 3$.
If the  traceless part of the Cartan torsion, namely
\[
M_{\alpha \beta \gamma}:=C_{\alpha \beta \gamma}-\frac{1}{n+2}(h_{\alpha \beta} I_\gamma+h_{\gamma \alpha } I_\beta+h_{ \beta \gamma } I_\alpha)
\]
vanishes, then the indicatrix lies in a hyperquadric.
\end{theorem}
%
%
%
%This fact is pretty obvious moving  from Finsler geometry since if the Cartan torsion vanishes then $\mathscr{L}$ is quadratic in $y$ and the indicatrix is just a level set of $\mathscr{L}$.
%
%In the other direction, it is well known in affine differential geometry that if the traceless part of the cubic form vanishes then the hypersurface lies in  a hyperquadric \cite[p.43, Lemma 3.2]{simon68}  \cite[Theor.\ 6.4]{nomizu94}. This result translated to (Lorentz-)Finsler geometry states that
%\begin{theorem}
%If the  traceless part of the Cartan torsion, namely
%\[
%M_{\alpha \beta \gamma}:=C_{\alpha \beta \gamma}-\frac{1}{n+1}(h_{\alpha \beta} I_\gamma+h_{\gamma \alpha } I_\beta+h_{ \beta \gamma } I_\alpha)
%\]
%vanishes, then the indicatrix lies in a hyperquadric.
%\end{theorem}
In the positive definite case it is necessarily an ellipsoid which need not be centered at the origin of $T_xM$, thus we have a Randers space if the origin of the tangent space lies in the interior of the ellipsoid and a Kropina space if it lies on the boundary. In Finsler geometry this result was stated by Matsumoto  \cite{matsumoto72,matsumoto78} but, as we have shown, it can conveniently  regarded as  the translation of a  classical theorem from affine differential geometry. This observation can also be found in \cite{mo10}.

%Joining this result with Theorems \ref{doi} and \ref{doo} we get
%\begin{theorem}
%Let a (Lorentz-)Finsler space have affine sphere indicatrix not intersecting the origin of $T_xM$, and assume $M_{\alpha \beta \gamma}=0$, namely that it is a hyperquadric,  then
%\[
%C_{\alpha \beta \gamma}=\frac{1}{n+1}(h_{\alpha \beta} I_\gamma+h_{\gamma \alpha } I_\beta+h_{ \beta \gamma } I_\alpha)
%\]

%\end{theorem}

\subsubsection{(Lorentz-)Randers and (Lorentz-)Kropina spaces}

Observe that  Theorem \ref{hqp} applied  to a metric $g$ of  Lorentzian signature gives that the indicatrix lies in  a hyperboloid whose center is not necessarily the origin of $T_xM$.  Let $C$ be the cone determined by the hyperboloid and its center $p$. We have shown above that the indicatrix is the whole hyperboloid only if  $p$ of is causal and  future directed, namely $p\in \bar C-p = \Omega$.
 This type of causally translated hyperboloids define the indicatrix for the Lorentzian analogs to the Randers and Kropina spaces.

The indicatrix is the locus $\{y: \eta(y-p,y-p)=-1, \eta(y,p)\le 0\}$ where $\eta$ is a quadratic form of Lorentzian signature on $T_xM$ and $p$ is causal,  $\eta(p,p)\le 0$. The  domain of the Lagrangian is then $\Omega=\{y: \eta(y,y)<0, \eta(y,p)\le 0\}$.

 %Suppose that  the indicatrix is (Euclidean) complete namely, suppose that it is the whole hyperboloid. Let  $o$ its center and let $O$ be the origin of  $T_xM$.
%Let $C$ the cone generated by the indicatrix and $o$, and let $\Omega$ be the cone generated by the indicatrix and $O$. The point $O$ cannot lie  outside the $o$-opposite of the closed cone $\bar C$, for in this case there would be  an half line starting from the origin of $T_xM$ and tangent to the indicatrix, which is impossible. Thus the vector $c:=o-O$ belongs\footnote{Actually there is a further option:  $c$ is  timelike past directed, namely it belongs to $-C$ and stays beyond the  indicatrix. Then the indicatrix may be used to define a Finsler, rather than Lorentz-Finsler space. We notice that in the upcoming notation $\eta(c,c)+1<0$ and $\eta(c,y)>0$. We shall not consider this case.} to $\bar C$. Observe that $\Omega$ is obtained translating the origin of $C$.

  %or the Kropina space, depending on whether the vector $c$  belongs  to $C$ (future timelike) or its  boundary (future lightlike), respectively.
  %There are coordinates $y^\alpha$ for which the indicatrix takes the form $\eta_{\alpha \beta}(y^\alpha -p^\alpha) (y^\alpha -p^\alpha)=-1$ where the future causality condition reads $p^0>0$ and $\eta_{\alpha \beta} p^\alpha p^\beta\le 0$, inequality sign in the Lorentz-Finsler case, and equality sign in the Lorentz-Kropina case.
  Using Eq.\ (\ref{jui}) for $1+\eta(p,p)\ne 0$ we arrive at the generalized Lorentz-Randers space
 \begin{align} \label{kdu}
 F&=\frac{\eta(p,y)+\sqrt{\eta(p,y)^2-(1+\eta(p,p)) \eta(y,y)}}{1+\eta(p,p)}, \qquad \textrm{gener.
 Lorentz-Randers space}
 \end{align}
 and for $1+\eta(p,p)= 0$ at
 \begin{align*}
 F&=\frac{\eta(y,y)}{2\eta(p,y)}, \qquad  \qquad \qquad \qquad  \qquad \qquad \qquad \qquad \quad\textrm{Lorentz-Kropina space}.
 \end{align*}
 Concerning the generalized Lorentz-Randers case, the argument of the square root is positive by the Lorentzian reverse Cauchy-Schwarz inequality $\eta(p,y)^2-\eta(p,p) \eta(y,y)> 0$, while $F$ is indeed positive independently of the sign of $1+\eta(p,p)$. In both generalized Randers and Kropina's cases, the expression of $F$ for Finsler spaces is analogous,  it is sufficient to replace $\eta$  by $-e$, the minus  Euclidean quadratic form.

 It is interesting to observe that in the generalized Lorentz-Randers case the argument in the square root of (\ref{kdu}) is a minus Lorentzian quadratic form  iff $ 1+\eta(p,p)>0$, in which case we  define the Lorentzian quadratic form
 \begin{equation}
 \bar \eta(y,y)=-\frac{1}{(1+\eta(p,p))^2} \big(\eta(p,y)^2-(1+\eta(p,p))  \eta(y,y)\big) .
 \end{equation}
We call the spaces in which $ 1+\eta(p,p)>0$ Lorentz-Randers.

Similarly, for Randers spaces the argument $e(p,y)^2+(1-e(p,p)) e(y,y)$ of the square root to the equation analogous to (\ref{kdu}) is a positive definite quadratic form (the generalized Randers case can only be Randers since it is assumed $1-e(p,p)>0$, for otherwise the ellipsoid cannot be interpreted as indicatrix since the origin of $T_xM$ would stay outside it) and we define
\begin{equation}
\bar e(y,y)= \frac{1}{(1-e(p,p))^2} \big(e(p,y)^2+(1-e(p,p)) e(y,y)\big)
\end{equation}
%Thus at the given point $x$ one could make a linear change of coordinates so as to diagonalize $\bar \eta$ rather than $\eta$ (resp. $\bar e$ rather than $e$). This diagonalization can be accomplished only at a point $x\in M$ and in general the Lorentzian (Euclidean) quadratic form $\eta$ (rep.\ $e$) or the overlined version would depend on $x$.
  In the Lorentz-Randers (Randers) case let us set  $\bar p=p (1+\eta(p,p))$ (resp.\ $\bar p=p (1-e(p,p))$).  Observe that $\bar \eta(\bar p,\bar p)=\eta(p,p) $ and $\bar \eta(\bar p,y)=\eta(p,y)/(1+\eta(p,p))$ (resp.\  $\bar e (\bar p,\bar p)=e(p,p) $, $\bar e(\bar p, y)=e(p,y)/(1-e(p,p))$).
The  Lorentz-Randers (Randers) case can be recognized as the function $F$ reads
 \begin{equation} \label{for}
 F=\bar\eta(\bar p, y)+\sqrt{-\bar \eta(y,y)} \quad \left(\textrm{resp}.\ F=-\bar e(\bar p, y)+\sqrt{\bar e(y,y)}\right)
 \end{equation}
 where $-1<\bar \eta(\bar p, \bar p)\le 0$ (resp.\ $\bar e(\bar p,\bar p)<1$). Traditionally the  Lorentz-Randers spaces are those given by the previous expression \cite{randers41,storer00,basilakos13}. However, the inequalities constraining $\bar p$ were not recognized and often the Lagrangian cone domain $\Omega$ had been incorrectly identified with (half) the locus $\bar \eta(y,y)<0$ rather than with the smaller set $F>0$ obtainable as $\eta(y,y)<0, \eta(y,p)\le 0$, where
 \[
 \eta(y,y)=\big(1+\bar \eta(\bar p,\bar p)\big)\big(\bar\eta(y,y)+\bar \eta(\bar p,y)^2\big) .
 \]
It is commonplace to regard particle Lagrangians of electromagnetic type as a manifestation of Finsler geometry. Our analysis shows that since the electromagnetic field does not satisfy any causality condition, it is inappropriate to mention (Lorentz-)Finsler geometry, for the indicatrix $F=1$ is not a  complete hyperquadric.
 %if the indicatrix $F=1$ is not a well defined complete hyperquadric or if it does not play the role of observer space.

\subsubsection{Relationship between connections}

Let $\nabla$ be the connection introduced in Eq.\ (\ref{gau}) and let  $ \nabla^h$ be the Levi-Civita connection of the affine metric $h$ for the centroaffine transverse field.
In order to translate some results from affine differential geometry  it is necessary to establish how the connections $\nabla$ and $\nabla^h$ should be expressed  in the language of Finsler geometry. In fact we can pass from tensors living on the indicatrix to tensors living on $T_xM$ by using positive homogeneity, and conversely we can restrict Finslerian tensors to the indicatrix provided they annihilate the position vector $y$.

We have
\begin{theorem}
The connection $\nabla$ is just the usual derivative $D$ (obtained through ordinary differentiation $\p/\p y^\alpha$ on $T_xM$) followed by the projection $P^\alpha_\beta=\delta^\alpha_\beta-y_\beta y^\alpha /g_y(y,y)$ on the tangent space to the indicatrix, namely $\nabla=P\circ D$, while $ \nabla^h$ is just the vertical Cartan covariant derivative $\nabla^{VC}$ followed by the same projection, $ \nabla^h=P\circ \nabla^{VC}$.
\end{theorem}
%\begin{align*}
%\nabla^{VC}_\gamma h_{\alpha \beta}&=-\frac{4}{(2 \mathscr{L})^3} y_\alpha y_\beta y_\gamma+\frac{1}{(2 \mathscr{L})^2}[2 g_{\alpha \beta} y_\gamma+g_{\gamma \beta} y_\alpha+g_{\alpha \gamma} y_\beta] ,
%=- \big(g_{\alpha \gamma} y_\beta+y_\alpha g_{\beta \gamma}-2 y_\alpha y_\beta  y_\gamma /g_y(y,y)\big )/g_y(y,y)
%\end{align*}

It is understood that the projection will act on every index of the derivative.
We recall here that the Cartan vertical derivative has connection coefficients $C^\gamma_{\alpha \beta}$ in the canonical coordinates of $T_xM$.
\begin{proof}
The first statement follows from the Gauss formula (\ref{gau}). For the second statement observe that $\nabla^{VC}_\gamma y^\alpha=\delta^\alpha_\gamma$ and $\nabla^{VC}_\gamma y_\beta=\nabla^{VC}_\gamma ( g_{\beta \alpha} y^\alpha)=g_{\beta \gamma}$, $\nabla^{VC}_\gamma (2\mathscr{L})=2y_\gamma$, as a consequence
 the vertical derivative of $h$ vanishes once projected, $P^{\alpha}_{\alpha'} P^{\beta}_{\beta '} P^\gamma_{\gamma'}\nabla^{VC}_\gamma h_{\alpha \beta}=0$. Thus $P\circ \nabla^{VC}$ provides a symmetric connection compatible with the affine metric, hence it is the Levi-Civita connection of $h$, namely $\nabla^h$. $\square$
\end{proof}

Let $h_b$ and $c_b$ be the affine metric and Pick cubic form for the Blaschke normal. In \cite{dillen94,hu11,hildebrand15} it is proved that a hyperbolic affine sphere satisfies $\nabla^{h_b} c_b=0$ if and only if it is homogeneous (hence asymptotic to a symmetric cone). If the indicatrix is an affine sphere centered at the origin (i.e.\ $I_\alpha=0$) the Blaschke structure coincides with the centroaffine structure (Theor.\ \ref{rel}), thus  $\nabla^{h_b} c_b=f^* P\circ \nabla^{VC} C $ where $P\circ \nabla^{VC} C$ is
\begin{align*}
%\nabla^h c=
P^{\mu'}_\mu P^{\alpha'}_\alpha P^{\beta'}_\beta P^{\gamma'}_\gamma \nabla^{VC}_{\mu'} C_{\alpha' \beta' \gamma'}&=C_{ \alpha \beta \gamma \mu}- C_{ \sigma \beta \gamma} C^\sigma_{ \alpha \mu}- C_{   \alpha \sigma \gamma} C^{\sigma}_{\beta \mu}  -   C_{\alpha \beta \sigma} C^\sigma_{ \gamma \mu}\\
&\!\!\!\!\!\!\! +\frac{1}{ g_y(y,y)}(C_{\alpha \beta \gamma} y_\mu+C_{\mu \beta \gamma} y_\alpha+C_{\alpha \mu \gamma} y_\beta+C_{\alpha \beta \mu} y_\gamma)
\end{align*}
As a consequence, we have the Finslerian result
\begin{theorem}
Let the Lorentz-Finsler space have  hyperbolic affine sphere indicatrices centered in the zero section. The indicatrix is homogeneous if and only if the  previous tensor in display vanishes. Under homogeneity the Lorentz-Finsler space is Berwald if and only if it is Landsberg.
\end{theorem}

\begin{proof}
We need only to prove the last statement. Under the Landsberg assumption $0=L_{\alpha \beta \gamma}=y^\mu \nabla^H_\mu C_{\alpha \beta \gamma}$ thus by \cite[Eq.\ (58)]{minguzzi14c} $G_{\alpha \beta \gamma \delta}= -y^\mu \nabla^H_\mu C_{\alpha \beta \gamma \delta}=-y^\mu \nabla^H_\mu (P\circ \nabla^{VC} C)_{\alpha \beta \gamma \delta}$. Under homogeneity the last tensor vanishes hence the thesis.
\end{proof}

\section{Applications to anisotropic relativity}

In this section we apply some deep mathematical results on affine sphere theory to spacetime physics.

\subsection{Classification of affine spheres and cone structures}
A non-degenerate hypersurface on  affine space having positive definite affine metric is said to be {\em affine complete} if it is complete with respect to the affine metric.
A non-degenerate hypersurface on  affine space   is {\em Euclidean complete} if it is complete with respect to the Euclidean metric induced by the coordinates placed on the affine space. Clearly, the latter notion is independent of the chosen Cartesian coordinate system.

These two notions of completeness are independent but Trudinger and Wang proved that for strictly convex smooth affine hypersurfaces and $n\ge 2$ affine completeness implies Euclidean completeness \cite[Theor.\ 5.1]{trudinger02}.

Observe that the notion of affine completeness makes sense only for definite hypersurfaces. The classification of definite affine spheres has been completed thanks to the work of several  mathematicians. We refer the reader to the reviews by Trudinger-Wang \cite{trudinger08} and Loftin \cite{loftin10} for more details on the theory of affine spheres.

%According to a theorem due to Titze and Nakajima \cite{burago74} a
%closed connected set $F\subset \mathbb{R}^n$ is convex if and only
%if it is locally convex, in the sense that for every $p\in F$ there
%is some convex neighborhood $U\ni p$ such that $U\cap F$ is convex.

A result due to Blaschke \cite{blaschke23} ($n=3$) and Deicke \cite{deicke53,brickell65} (any $n$) further extended by  Calabi \cite{calabi72}, and Cheng and Yau \cite{cheng86} reads
\begin{theorem}
Any definite elliptic affine sphere is an ellipsoid provided it  satisfies any among the following conditions: (a) compactness, (b) affine completeness, (c) Euclidean completeness.
\end{theorem}

The classification of definite parabolic affine spheres is due to  J\"orgens \cite{jorgens54}, Pogorelov \cite{pogorelov72}, Calabi \cite{calabi58}, Cheng and Yau \cite{cheng86}

\begin{theorem} \label{ksw}
Any definite parabolic affine sphere is an elliptic paraboloid provided it  satisfies any among the following conditions: (a) it is  a properly embedded\footnote{The embedding is proper if the inverse image of compact sets is compact. Roughly, the hypersurface has no `edge'. This is always the case for the indicatrices of Lorentz-Finsler geometry.} convex hypersurface, (b) affine completeness.
\end{theorem}

Finally, the classification of definite  hyperbolic affine spheres was conjectured by Calabi \cite{calabi72} and proved by Cheng and Yau \cite{cheng77,cheng86} (see also Calabi and Nirenberg, unpublished \cite{cheng86}). This proof was  improved and clarified by Gigena \cite{gigena81}, Sasaki \cite{sasaki80} and A.-M. Li \cite{li92}.

\begin{theorem} \label{hop}
For a definite hyperbolic affine sphere the following properties are equivalent: (a) Euclidean completeness, (b) affine completeness, (c) properly embedded.

Any such affine sphere $N$ is asymptotic to the boundary $\p C$ of an open convex sharp cone $C$ given by the convex hull of $N$ with its center. Conversely, any sharp open convex cone $C$ contains, up to rescalings, a unique properly embedded affine sphere $N$ which is asymptotic to $\p C$.
\end{theorem}

This result is extremely important because it shows that (definite) hyperbolic affine spheres and sharp convex cones are essentially the same object. It is based on Theor.\ \ref{lpo} and on the next result by Cheng and Yau \cite[Theor.\ 6]{cheng77}

\begin{theorem}[Cheng and Yau] \label{cyu}
Let $\rho >0$ and suppose that $D$ is a bounded convex domain in $\mathbb{R}^n$. Then there exists a unique continuous convex function $u$ on $\bar{D}$ such that $u \in C^\infty(D)$, $u$ satisfies $\det (u_{ij})=\rho^2(-1/u)^{n+2}$, and $u=0$ on $\p D$.
\end{theorem}

So suppose to have been given a cone $\Omega_x\subset T_xM$ and introduce coordinates $\{y^\alpha\}$ as in  Theor.\ \ref{lpo} in such a way that the hyperplane $y^0=1$ cuts the cone in a section $(1,D)$, with $D$ bounded  convex domain.
By Theor.\ \ref{lpo} given the solution $u$ on $ \bar D$, the embedding
\[
{\bm v} \mapsto f({\bm v})= -\frac{1}{u({\bm v})} (1,{\bm v}), \qquad {\bm v} \in D
\]
is an affine sphere with  affine mean curvature $H=-1$, which is asymptotic to $\p \Omega_x$. The corresponding Finsler Lagrangian is found imposing $-1=2\mathscr{L}(x,f({\bm v}))=2(-\frac{1}{u})^2\mathscr{L}(x,(1,{\bm v}))$, which gives $\mathscr{L}_x\vert_{(1,D)}=-u^2/2$, where the left-hand side is the restriction of the Finsler Lagrangian to the intersection between the hyperplane $y^0=1$ and the convex cone.  By positive homogeneity the Finsler Lagrangian is then determined on the whole cone $\bar{\Omega}_x$. From this equation  we read the regularity of $\mathscr{L}$ from that of $u$, in particular $\mathscr{L}_x\in C^\infty(\Omega_x)$   and it  can be continuously extended setting $\mathscr{L}=0$ on $\p \Omega$.
The relationship between $u$ and $\mathscr{L}$ is
\begin{align}
\mathscr{L}(x,(y^0,{\bm y}))&=-\frac{1}{2} (y^0)^2 u^2({\bm y}/y^0), \label{osd}\\
u(x,{\bm v})&=-\sqrt{-2\mathscr{L}(x,(1,{\bm v}))}. \label{sop}
\end{align}
We can summarize this result as follows
\begin{theorem} \label{cor}
Given an open convex sharp cone $\Omega_x\subset T_xM$ and a (vertically translational invariant) measure $\mu_x=\rho(x)\dd^{n+1} y$ there is one and only one Lorentz-Finsler Lagrangian $\mathscr{L}_x$ on $\Omega_x$ (so $\mathscr{L}_x<0$ on $\Omega_x$ and it converges to zero at $\p \Omega_x$) having vanishing mean Cartan torsion and such that $-\det \dd^2\mathscr{L}_x=\rho^2$. This Lagrangian is $C^\infty(\Omega_x)$ and its indicatrix is an affine sphere with affine mean curvature $H=-1$ with respect to $\mu_x$.
\end{theorem}

In Lorentz-Finsler gravity theories one obtains most scalars and tensors of physical interest from the Finsler Lagrangian. For instance, the proper time over a curve $x\colon I \to M$, $t \mapsto x(t)$, is $\int \!\!\sqrt{-2 \mathscr{L}(x(t), \dot x(t))}\, \dd t$.
The minus proper time over a curve  multiplied by the mass of the particle gives the action of the particle (it is locally minimized over geodesics). Let us write it for a unit mass particle
\begin{equation} \label{mos}
S=\int u(x(t),   {\bm v}(t)) \, \dd t.
\end{equation}
The projective coordinates $ {\bm v}$ and the function $u$ transform as in Eqs.\ (\ref{pro1})-(\ref{pro3}) under change of coordinates on $M$.

From the above equation we conclude that $u$ is
%, saved for the constant factor $m$,
the Lagrangian per unit mass. It has  to be distinguished from the Finsler  Lagrangian $\mathscr{L}$. We shall return on this interpretation in connection with the Legendre transform in Sect.\ \ref{leg}.
%\end{remark}

Theorem \ref{cor} leaves open some interesting related questions:
\begin{itemize}
\item[($\alpha$)] Can the regularity of $\mathscr{L}$ at the boundary be improved perhaps improving the regularity of the convex cone?
\item[($\beta$)] Can the Lorentz-Finsler Lagrangian be continuously extended beyond the cone $\bar{\Omega}$ preserving a vanishing mean Cartan torsion?
\item[($\gamma$)] Given a sufficiently regular convex cone is it possible to find a Lorentzian definite affine sphere which is asymptotic to the cone and its opposite?
\end{itemize}

As we mentioned most of causality theory and mathematical relativity depends on just the future causal cone. At first some regularity at the boundary of the cone seems desirable if not necessary in order to define a notion of lightlike geodesic. However, there  are approaches that do not demand such regularity, for instance,  lightlike geodesics could be defined as limits of timelike geodesics, or taking limits of Lorentzian metrics \cite[Remark 2]{minguzzi14h}. In Section \ref{qon} we shall show that it is possible to define the lightlike geodesics for a sufficiently differentiable light cone distribution without making reference to the differentiability of the Lagrangian at the boundary.

Nevertheless, it turns out that question ($\alpha$) above receives a partly positive answer. We recall first a result  by Cheng and Yau \cite[p.\ 53]{cheng77} according to which $u \in C^{0,1/2}(\bar{D})$ for $D$ strongly convex and  with  $C^2$ boundary. Since $u=0$ at $\p D$ this result implies that $\mathscr{L}$ is Lipschitz at the boundary of the cone.

Stronger results can be obtained looking at the graph of $u$, ${\bm v} \mapsto ({\bm v}, u({\bm v}))$. Lin and Wang  proved that though $u$ is not $C^1$ at $\p D$, its graph is a $C^{2,\alpha}$, $\alpha\in(0,1)$, hypersurface with boundary whenever the  boundary  $\p D$ is  $C^{2,\alpha}$ \cite[Theor.\ 4.6]{lin98}. Recently, for smooth boundary $\p D$, Jian and Wang \cite{jian13} proved that the graph of $u$ is a $C^{n+2,\alpha}$ hypersurface  with boundary for every $\alpha \in(0,1)$.

%The degree of differentiability $C^{k,\alpha}$, $k\ge 2$, $\alpha \in (0,1)$, of the  graph $Q$ of $u$ at $\p D$ is really the degree of differentiability of $\mathscr{L}$ at the cone.

 Now, we need
\begin{proposition} \label{qoi}
If the graph of $u$ is a  $C^{k,1/2}$, $k\ge 3$, hypersurface with boundary, then   $\mathscr{L}_x\in C^{1,1/2}(\bar \Omega_x)$ and $\dd \mathscr{L}_x|_{\p \Omega_x}\ne 0$.
\end{proposition}

\begin{proof}
Let us prove this fact near a point $p \in \p D$ at which the line parallel to the $v^n$ axis is transverse to $\p D$ and $v^n$ grows inwardly (the general case follows rotating the axes). Since the derivatives of $u$ diverge the graph of $u$ can be expressed as a graph $v^n=f(\bar v, z)$, where $f$ is convex (due to the convexity of the epigraph of $u$) and such that  $v^n= f(\bar v, u({\bm v}))$, $\bar v=(v^1,\cdots,v^{n-1})$, ${\bm v}=(\bar v,v^n)$, (see \cite[Sect.\ 4]{lin98}). By the assumption on the hypersurface, $f$ is $C^{k,1/2}$.  In other words, the $C^{k,1/2}$ embedding  can be taken to be $(\bar v,z)\mapsto (\bar v,f(\bar v, z),z)$. The Taylor expansion  of $f$ with respect to $z$ at $(\bar v,0)$  gives $\Delta v^n:=v^n-w(\bar v)=h(\bar v) u({\bm v})+ R(\bar v, v^n ) u^2({\bm v})$ where $R$ is a bounded Lagrange remainder which converges to $\p^2_z f (\bar v,0)/2=:R(\bar v, 0 )$ for $u\to 0$ and $w(\bar v)$ is the $C^{k,\alpha}$ graphing function of $\p D$ on $u=0$. Now, $u({\bm v})$ is not Lipschitz, for its derivatives diverge, but we mentioned that $u^2({\bm v})$ is Lipschitz, thus dividing by $\Delta v^n $ and letting $ v^n \to w(\bar v)$ we see that this identity can only hold with $h=0$. Inserting back $h=0$ we also get that $R(\bar v,0)\ne 0$. The argument proves that $\p_z f(\bar v,0)=0$  while $\p_z^2 f(\bar v,0)> 0$, so we have $f(\bar v, z)=w(\bar v)+b(\bar v) z^2+ z^2 c(\bar v, z)$ where $b,c \in C^{k-2,1/2}$, $b>0$. %Actually, $b$ is really $C^2$ (if there is a diverging sequence of second derivatives, differentiate twice with respect to $\bar v$, replace $z$ with the inverse of these second derivatives and take the limit to get a contradiction).
%Setting $z=0$ we see that $w$ is $C^{k,\alpha}$ bringing it to the
Let $ d$ be defined by $d(\bar v, y)=y c(\bar v, \sqrt{y})$ so that $\p_y d(\bar v, 0)=D_{\bar v} d(\bar v, 0)=0$,
% $d\in C^{1,\alpha/2}$ for $k=2$ and
  $d\in C^{1,1/2}$.
  %for $k\ge 3$.
  The equation $v^n=w(\bar v)+b(\bar v) y+ d(\bar v, y)$ with $y=u^2({\bm v})$ can now be differentiated
giving
\begin{align*}
\p_n u^2(\bar v,v^n)&= [b(\bar v) +D_{y} d(\bar v, u^2({\bm v}))]^{-1},\\
D_{\bar v} u^2(\bar v,v^n)&=-b^{-1}(\bar v) [ D_{\bar v} w(\bar v)+(D_{\bar v} b) u^2+D_{\bar v} d(\bar v, u^2)]
\end{align*}
which are $C^{0,1/2}(\bar D)$. Moreover, on the boundary $u=0$ thus $\p_n u^2=b^{-1}\ne 0$ which implies $\dd \mathscr{L}_x\ne 0$ there. $\square$
\end{proof}

By the mentioned result by Jian, Lin and Wang,  since on spacetime $n\ge 1$, we have
\begin{corollary} \label{hud}
If the light cone $\p \Omega_x$ is a  smooth hypersurface, then   $\mathscr{L}_x\in C^{1,1/2}(\bar \Omega_x)$ and $\dd \mathscr{L}_x|_{\p \Omega_x}\ne 0$.
\end{corollary}
Thus this result shows that some regularity of the Lagrangian at the light cone can be accomplished taking sufficiently regular cones.

This regularity does not allow us to define the lightlike geodesics in the usual way. An alternative and satisfactory method will be given  in Section \ref{qon} where we shall show that this regularity helps to define the affine parameter over lightlike geodesics.

%
%
%
%In fact, $Q$ is, at least locally, the image of a map $i: O \to \mathbb{R}^{n+1}$, $O\subset \mathbb{R}^n$, ${\bm q} \to i({\bm q})$. But the map $k: \mathbb{R}^{n+1} \to \mathbb{R}^{n+1}$, $({\bm v}, z) \to ({\bm v}, z^2)$ is $ C^\infty$, thus the composition $k \circ i$ has the degree of differentiability of $i$. It is injective because $i$ has image in $z\le 0$. Its Jacobian is also non-degenerate (also at the boundary), in fact in order to be degenerate the $n$-dimensional image of $i_*$ would need to have some vector in the ker of $k_*$ which is given by the only vector $({\bm 0}, 1)$
%
%provides an immersion which . This is really an embedding since  $i$ has image in $z\le 0$.
%
%
%This can be easily seen since $u \in C^{0,1/2}(\bar{D})$ and $u=0$ on $\p D$ imply that $u^2$ and hence $\mathscr{L}$ are Lipschitz[COMPLETARE]
%

%In any case, it is useful to recall a result by Cheng and Yau \cite[p.\ 53]{cheng77} according to which $u \in C^{0,1/2}(\bar{D})$ for $D$ strongly convex and a $C^2$ boundary $\p D$. Since $u=0$ at $\p D$ this result implies that $\mathscr{L}$ is Lipschitz at the boundary of the cone.

Though the Lagrangian could be non-smooth at the cone, it would be nice if we could work directly with a Lagrangian defined all over $TM\backslash 0$ for which the mean Cartan torsion vanishes everywhere, since traditionally the theory of Finsler connections and sprays has been developed on the slit tangent bundle (hence question $\gamma$).
As far as we know this  question has not been investigated, possibly  because the Lorentzian affine sphere asymptotically approaching the cones $\p \Omega_x$ and $-\p \Omega_x$ from outside would be described by a {\em non-elliptic}  Monge-Amp\`ere equation for which maximum principles are not available.

The next result and its interpretation are the main objectives of this work.
\begin{theorem} \label{equ}
Over a manifold $M$ the following three concepts are equivalent:
\begin{itemize}
\item[(a)] %A maximally negative extended
Lorentz-Finsler Lagrangian with vanishing mean Cartan torsion,
\item[(b)] Volume form and sharp cone distribution over $M$,
\item[(c)] Affine complete, definite, hyperbolic affine sphere subbundle  of the tangent bundle with center in the zero section.
\end{itemize}
\end{theorem}

%By now the relation between these concepts should be clear.

%Although confined in a proof environment
The reader might want to check the next proof environment for details on the correspondence.

%more closely how the correspondence is obtained looking at the next proof .

\begin{proof}
Given (a) we obtain (c) taking the indicatrix subbundle $\mathscr{I}^-$. Conversely, given (c) we obtain (a) using Eq.\ (\ref{jui}).

Given (c) we obtain (b) selecting first the volume form which assigns to the affine sphere the  affine mean curvature $H=-1$ according to Theorem \ref{kwx}. The sharp cone distribution is that given by the asymptotic cones to the affine spheres as determined by Theorem \ref{hop}.

Conversely, given (b) we get (c) as follows. We determine a distribution of properly embedded hyperbolic affine spheres up to rescalings thanks to Theorem \ref{hop}, and among those we select that for which $H=-1$ according to the volume form provided by (b).

From (a) we can also obtain directly the measure of (b) using Eq.\ (\ref{nos}). This possibility follows from Theorem \ref{rel}. $\square$
\end{proof}

The problem of the determination of spacetime from a volume form and a cone distribution is solved if we use in place of Lorentz-Finsler spacetimes the next more specialized objects.

\begin{definition} \label{ksp}
An {\em affine sphere spacetime} is any of the equivalent structures mentioned in Theor.\ \ref{equ}.
\end{definition}

%What we find remarkable in this definition is its
The primitiveness of this definition seems remarkable. Indeed, point (c) clarifies that it relies on just the manifold structure of $M$, as not even a volume form is required. Nevertheless,  through  the equivalence with (a) we can recover all the metrical aspects which are needed in calculations: given the Finsler Lagrangian we can define the metric, the length of timelike curves and hence the proper time of observers. We can define the spray and its geodesics and so have a natural notion of free fall; we can construct tensors over $M$, ask for the validity of generalized Einstein's equations, and so on. Still at the core of these algebraic objects there is just an affine sphere distribution.

The sharpness condition appearing in (b) is simply the request that the speed of light be finite in any direction. We shall see in a moment the physical meaning of  the affine completeness appearing in (c).

The spacetime of general relativity is recovered whenever one of the following equivalent condition holds: (a) the Cartan torsion vanishes, (b) the cones are round, (c) at each point the  hyperbolic affine sphere is a quadric.

\begin{remark} \label{ric}{{\em Physical meaning of affine completeness}}\\
In any Lorentz-Finsler theory the affine metric coincides with the restriction of the Lorentz-Finsler metric to the indicatrix (Theor.\ \ref{rel}). We mentioned that this object measures the length of  vectors tangent to the indicatrix, namely the length of accelerations. Physically, the affine completeness of the indicatrix  reflects the fact that an ideal rocket having bounded proper acceleration cannot reach the boundary of the indicatrix in a finite proper time, namely that the speed of light cannot be experienced by massive particles, and hence that there is a meaningful distinction between massive and massless particles. This requirement appears to be physically motivated  so it is natural to demand the affine completeness of the  indicatrix as in characterization (c).

We shall see later on that the cotangent space admits a dual indicatrix. Its  affine completeness represents the impossibility of reaching the speed of light by applying a bounded force for a finite proper time. The physical equivalence of these conditions is nothing but Newton's second law: the proper acceleration (measured by a non-accelerating local observer) is proportional to the force.
%experienced force of inertia (measured onboard by a dynamometer).
\end{remark}

\begin{remark}{{\em Tangent and cotangent translated affine sphere spacetimes}}\\
A generalization of the notion of affine sphere spacetime can be obtained introducing a causal vector field as a further ingredient and translating the affine sphere distribution on the tangent space as done in Theor.\ \ref{doo}. These spaces might be called {\em tangent translated affine sphere spacetimes} and have to be distinguished from the {\em cotangent translated affine sphere spacetimes}, for which the translation takes place on the cotangent space. In fact, the dual of a translated affine sphere (c.f.\ Sect.\ \ref{leg}) is not necessarily a translated affine sphere. If the affine sphere indicatrix is a quadric then the translation on the tangent space induces a translation on the cotangent space and conversely. This  restricted family includes the Lorentz-Randers spaces. We expect that the real light cones could be (possibly translated) affine spheres departing slightly from isotropy. The presence of matter could in principle modify the affine sphere equation. We shall investigate the relevant modifications in the next works.
\end{remark}

\begin{remark} \label{one}
Let $V$ be a 1-dimensional vector space, and define a hyperbolic affine sphere as any point $p$ different from the origin. The cone generated by it is clearly one half $V^+\ni p$ of the vector space. At the Lagrangian level the 0-dimensional hyperbolic affine sphere  is determined by $\mathscr{L}\colon V^+\to \mathbb{R}$ where $\mathscr{L}$ is positive homogeneous of degree two and negative, then $p$ is determined by  $-2\mathscr{L}(p)=1$. Correspondingly, a 1-dimensional affine sphere spacetime is any 1-dimensional manifold $M$ endowed with a function $\mathscr{L}(x,y)$ defined on one half of $TM$, positive homogeneous of degree two in $y$ and negative. Observe that its Hessian is negative definite so this is not a Lorentz-Finsler spacetime. Still this extension of the affine sphere spacetime definition is useful  when considering Calabi products \cite{calabi72}.
\end{remark}

Let $(M,\mathscr{L})$ and $(M',\mathscr{L}')$ be pseudo-Finsler spaces. Let us consider a diffeomorphism $f\colon M \to M'$ which induces a diffeomorphism, denoted with abuse of notation in the same way, $f\colon  E\to E'$, $(x,v) \mapsto (f(x),f_*(v))$ where $E$ denotes the slit tangent bundle. A diffeomorphism $f$ is said to be a {\em conformal transformation}
if there is a positive function $\varphi\colon M\to \mathbb{R}$  such that
 $f^*\mathscr{L}'=\varphi \mathscr{L}$ that is, for every $(x,y) \in \Omega$ where $\Omega$ is the conic domain of $\mathscr{L}$,  we have $(f(x),f_*(y))\in \Omega'$ and
\begin{equation}
\mathscr{L}'(f(x),f_*(y))=\varphi(x) \mathscr{L}(x,y),
\end{equation}
which is equivalent to\footnote{An observation due to Knebelman \cite{knebelman29b}  shows that  the definition of the conformal transformations with this equation for $\varphi$ dependent also on the fiber $y$ brings no more generality.} $f^*g'= \varphi g$. Observe that  $f^*\mathscr{L}'$ is a Finsler Lagrangian with the same signature of the Hessian of $\mathscr{L}'$. The map $f$ is an isomorphism (or an isometry) if $\varphi=1$.

%In the next theorem $\mathscr{L}$ and $\mathscr{L}'$ are assumed to have connected domains $\Omega, \Omega'$ in which they take negative value.
\begin{theorem}
Let $f\colon M \to M'$ be a diffeomorphism and suppose that $(M,\mathscr{L})$ and $(M',\mathscr{L}')$ are affine sphere spacetimes.  The map $f$ preserves the cone structure, namely for every $x\in M$, $f_*(\Omega_x)=\Omega_x'$, if and only if $f$ is a conformal transformation. It also preserves the volume form if and only if it is an isometry.
%
%
%If  $f$ preserves the cone structure, namely for every $x\in M$ $f_*(I_x^+)={I_x^+}'$ then  $(M',\mathscr{L}')$ is an affine sphere spacetime and $f$ is a conformal transformation. Conversely, if $f$ is a conformal transformation then it preserves the cone structure.
\end{theorem}

\begin{proof}
Suppose that  $f$ preserves the cone structure, namely  $\Omega'=f_*\Omega$. Since $\mathscr{L}'$ has vanishing mean Cartan torsion the same is true for $f^*\mathscr{L}'$. The Lorentz-Finsler Lagrangians $\mathscr{L}$ and $f^*\mathscr{L}'$ have the same domain, vanish at the boundary of $\Omega$ and have vanishing mean Cartan torsion. As a consequence they have affine sphere indicatrices, and since there is just one such indicatrix up to homotheties they are conformal.
%If $f$ preserves the cone structure then $\mathscr{L}$ and $f^*\mathscr{L}'$ have the same indicatrix at every point up to  scaling. In fact since $\mathscr{L}$ and $\mathscr{L}'$ are both maximally extended in the domain where they take negative value,  we have $\Omega'=f_*\Omega$. From here the implication is obvious.
If the volume form is preserved  the scaling factor is fixed to one since for a given volume form there is only one hyperbolic sphere with affine mean curvature $-1$.
The converse implication is trivial. $\square$
\end{proof}

\subsection{Physical interpretation of projective coordinates} \label{rem}

%\begin{remark}[Physical interpretation of coordinates]
At this stage we can also understand the physical interpretation of the inhomogeneous projective coordinates ${\bm v}$ introduced in Eq.\ (\ref{inh}) in the general framework of Lorentz-Finsler theories.
First we say that $\hat y\in \Omega_x$ is a {\em covariant velocity} if it is {\em on shell}, namely if $\hat y\in \mathscr{I}_x^-$.
Let $\hat y\in \mathscr{I}_x^-$ and let $P_{\hat y}$ be the hyperplane passing through the origin such that $\hat y+P_{\hat y}$ is  tangent to the indicatrix at $\hat y$. We can choose a $g_{\hat y}$-orthonormal frame so that $e_0=\hat y$. Since $\dd
\mathscr{L}=g_{\hat y}(\hat y, \cdot)$ the basis  $\{e_i, i=1,\cdots,n\}$ spans $P_{\hat y}$. The basis induces a coordinate system $\{y^\alpha\}$ on $T_xM$. Notice that at $\hat y$, $\dd \mathscr{L}=-\dd y^0$, furthermore at $\hat y$ we have $\det g_{\alpha \beta}=-1$.

 The coordinates $y^\alpha$  at $T_xM$ determine through the exponential map a local coordinate system in a neighborhood of $x$ which represents the coordinate system of the {\em observer} $\hat y$. Let $D_{\hat y}$ be that subset of $P_{\hat y}$ such that $\hat y+D_{\hat y}=(\hat y+P_{\hat y})\cap \Omega_x$. The set $D_{\hat y}$ is an open bounded convex set which represents the domain of allowed velocities of massive particles as seen from observer $\hat y$. Its boundedness expresses the finiteness  of the speed of light as measured by the observer.

Since $\hat y$ belongs to the indicatrix and $y=-\frac{1}{u({\bm v})}(1,{\bm v})$ over it, we have  with this choice of coordinates $u({\bm 0})=-1$. Let $x(s)$ be a timelike curve passing through $x=x(0)$ where the observer with covariant velocity ${\hat y}$ is also passing. Let $w^\alpha$ be the normal coordinates constructed by the observer in a neighborhood of $x$. Let $y\in \mathscr{I}_x^-$ be the covariant velocity of the particle at $x$. If $w^0$ changes of $\dd t$, ${\bm w}$ changes of ${\bm v} \dd t$, thus ${\bm v}$ is the velocity of a particle with covariant velocity $y$ as seen from the observer $\hat y$.

%If the indicatrix is an affine sphere the latter equality holds at any point of $T_xM$ since the determinant of the spacetime metric is independent of $y$.

The inhomogeneous projective coordinate ${\bm v}$ is just the {\em velocity} of the particle (see Fig.\ \ref{ind}) as measured by the observer $\hat y$, $D_{\hat y}$ is the domain of allowed velocities for massive particles as measured by $\hat y$, while $u({\bm v})$ plays the role of Lagrangian for the observer $\hat y$ (cf.\ Eq.\ (\ref{mos})). This result holds in any Lorentz-Finsler theory, so it is remarkable that
these physically relevant coordinates are at the same time the best coordinates in order to express the Monge-Amp\`ere PDE for $u$ for affine sphere spacetimes.
In the study of affine spheres they were introduced by Gigena \cite{gigena81}, and their usefulness has been  advocated  by Loftin \cite{loftin02}. The projective invariance of the affine sphere metric emphasized in Loewner and Nirenberg's work \cite{loewner74} is nothing but the well-posedness  of the affine sphere indicatrix geometry under projective changes, namely under changes of observer. More precisely, a change of observer is given by (\ref{pro1})-(\ref{pro3}) where, however, the matrix $A_{\alpha \beta}$ is not arbitrary since the basis $\{e_\alpha\}$ for the observer has to satisfy some conditions, namely it has to imply $\tilde g_{\alpha \beta}=\eta_{\alpha \beta}$ at ${\bm \tilde v}={\bm 0}$ (observe that $g_{\alpha \beta}=\eta_{\alpha \beta}$ at ${\bm  v}={\bm 0}$ as this is the observer condition on the coordinates).

We mentioned that in the observer coordinates of $\hat y$,  $\det g_{\alpha \beta}(\hat y)=-1$.  But in an affine sphere spacetime this determinant is independent of the point thus $\det g_{\alpha \beta}=-1$. In other words $\rho=1$.  Under a change of observer we also have $\det \tilde g_{\alpha \beta}=-1$, thus $\det A=1$. As a consequence,  for affine sphere spacetimes the  changes between observer coordinates are  unimodular.

%\end{remark}

The expansion of the Lagrangian $u$ in the observer coordinates is
\begin{equation} \label{huz}
u({\bm v})=-1+\frac{{\bm v}^2}{2} +\frac{1}{3}\, C_{ijk}(\hat y) v^i v^j v^k+\frac{1}{4!}\big(2 C_{ijkl}(\hat y)v^i v^j v^k v^l+3({\bm v}^2)^2\big)+\ldots
\end{equation}
The  quadratic term gives the usual classical kinetic energy for low speeds, see Eq.\ (\ref{mos}). The observer coordinates can be characterized as those coordinates for which the Taylor expansion up to second order of $u$ is $u({\bm v})=-1+{\bm v}^2/2+\cdots$. The expansion of the Lagrangian can be more suggestively written
\begin{equation} \label{eao}
u({\bm v})=-\sqrt{1-{\bm v}^2}+\frac{1}{3}\, C_{ijk}(\hat y) v^i v^j v^k+\frac{1}{12}\, C_{ijkl}(\hat y)v^i v^j v^k v^l+o(\vert v\vert^4) .
\end{equation}

The next result establishes that the isotropy of the speed of light is a property independent of the observer.

\begin{proposition} \label{njd}
In an affine sphere spacetime if the velocity domain $D$ is ellipsoidal at least for one observer  $\hat y\in\mathscr{I}_x$,
 then the same is true for every observer.\footnote{By a similar argument, taking into account the results of \cite{minguzzi16c} we have that  in four spacetime dimensions an analogous result holds for a domain of conical or tetrahedral shape.} In fact, $\mathscr{I}_x$ is actually the usual  hyperboloid (a quadric) as in special relativity thus the velocity domains in observer coordinates are balls.
\end{proposition}

In the hypothesis we do not require $D$ to be centered at the origin ${\bm v}=0$.

\begin{proof}
This result follows from the uniqueness of the Cheng-Yau solution $u$, see Theor.\ \ref{cyu}. The solution for a spherical domain of radius one centered at the origin is $u=-k\sqrt{1-{\bm v}^2}$ for some constant $k>0$ (see also \cite{minguzzi16c}). As a consequence, the solution for an elliptical domain centered at ${\bm c}$ is \[u=-k\sqrt{1- a_1^{-2} (v_1-c_1)^2+\cdots+ a_n^{-2} (v_n-c_n)^2}.\] The expansion in observer coordinates is $-1+{\bm v}^2/2+o(v^2)$ thus ${\bm c}=0$, $k=1$, ${\bm a}=1$, which proves that $u=-\sqrt{1-{\bm v}^2}$, the special relativistic solution. $\square$
\end{proof}

\begin{theorem}
Suppose that the Taylor expansion of the Lagrangian $u$ is that of non-relativistic physics (and for the matter of special relativity) up to order $\vert {\bm v}\vert^3$ for any observer, namely $u({\bm v})=-1+\frac{{\bm v}^2}{2}+o(v^3)$, then the Lorentz-Finsler manifold is a Lorentzian manifold.
\end{theorem}

\begin{proof}
Using Eq. (\ref{huz}) we have $C_{ijk}(\hat y) v^i v^j v^k=0$ for every observer $\hat y$ and ${\bm v}$ thus by polarization the Cartan torsion vanishes on the indicatrix and hence on $\Omega$. $\square$
\end{proof}

We are going to give a similar characterization  for affine sphere spacetimes. To that end it is convenient to pass from the mass normalized Lagrangian $u$ to an unnormalized Lagrangian denoted in the same way by replacing $u\to u/m$ in the previous formulas. In the normalized notation an affine sphere spacetime has an indicatrix having $H=-1$ thus
satisfying Eq.\ (\ref{hhw}). In the non-normalized formulation this equation reads
\begin{equation} \label{qhw}
\det u_{ij}= \vert \det g_{\alpha \beta}\vert \Big(-\frac{1}{u}\Big)^{n+2} m^{2(n+1)},
\end{equation}
thus we can identify $-m^{\frac{2(n+1)}{n+2}}$ with the non-normalized affine mean curvature  $H$. Let us continue working with the non-normalized notation till  the end of this section.

Every observer $\hat y$ in the observer coordinates determines a Lagrangian $u({\bm v})$, a Legendre map ${\bm v} \mapsto {\bm p}:=\p_{{\bm v}} u$ and a {\em mass matrix} $m_{ij}:=\frac{\p^2 u}{\p v^i \p v^j}$. The normalized trace of the mass matrix is $\frac{1}{n}\,\delta^{ij} m_{ij}$ and in the Lorentzian case it equals the mass of the particle $m$ (it would be 1 in the mass normalized approach). In Finslerian gravity theories this constant is necessarily $m$ for ${\bm v}=0$ but can run linearly in the velocity for small ${\bm v}$.

\begin{theorem}
If for every observer $\hat y\in \mathscr{I}$ the normalized trace of the mass matrix is a constant $m$ (the mass of the particle) up to quadratic corrections in $v$, then  the Lorentz-Finsler manifold is an affine sphere spacetime.
\end{theorem}

\begin{proof}
From (\ref{huz}) the mass matrix is $m_{ij}(v)=m\delta_{ij}+m2C_{ijk}(\hat y) v^k+o(v)$, thus its normalized trace is $m+m \frac{2}{n}I_k(\hat y) v^k+o(v)$ where we used $\delta^{ij}=g^{ij}+O(v)$, $C_{0jk}(\hat y)=C_{00k}(\hat y)=0$. Since $I_0(\hat y)=0$, the assumption implies $I_\mu(\hat y)=0$ for every observer, and hence all over $\Omega_x$. $\square$
\end{proof}

\begin{remark}
In other words the theorem states that affine sphere spacetimes are characterized by the property that for every test particle there is a constant $m$  such that for every close observer that looks at the test particle  $m_{ij}-m\delta_{ij}$ is  traceless not only at the zeroth order in $v$ but also at the first order. The constant $m$ is the rest mass of the particle. This characterization is reminiscent of the characterization of inertial coordinates systems which are those for which the apparent force acting on the particle has no linear contribution in ${\bm v}$, however, this is really a condition on the kinematical structure of the space.
\end{remark}

Joining the assumptions of Prop.\ \ref{njd} and the previous theorem we obtain  a characterization of Lorentzian spacetimes in Lorentz-Finsler theory.

\begin{proposition}
If at every event all observers experience  the non-relativistic characterization of mass as in the previous theorem, and if at least one observer measures an ellipsoidal (e.g.\ isotropic) speed of light then the spacetime is Lorentzian.
\end{proposition}

Another characterization of Lorentzian spacetimes can be obtained looking at those spacetimes which satisfy the {\em relativity principle}. These are the affine sphere spacetimes for which the group of linear non-degenerate endomorphisms $G(\Omega_x)$ which leaves invariant $\Omega_x$ is independent of $x$ and acts transitively on $\Omega_x$. In other words the cone $\Omega_x$ is homogeneous \cite{sasaki80}. The action descends to a transitive isometric action on the affine sphere indicatrix $\mathscr{I}_x$ (if we had selected an arbitrary indicatrix, it  would not be the case). This property is the mathematical realization of the idea that all observers are kinematically equivalent.  The Finsler Lagrangian is really some power of the characteristic function of the cone, but we shall not enter on this correspondence here. The important point is that for homogeneous cones the domain of allowed velocities $D_{\hat y}$ is really independent of $\hat y$ (up to space rotations of the observer coordinates) and the boundary  $\p D$ is $C^2$ only for the ellipsoid \cite{vinberg67,benoist01,jo03}. Physically, we can now interpret this result as follows
\begin{theorem} \label{muf}
For an affine sphere spacetime which satisfies the relativity principle, the speed of light has a $C^2$ dependence on the direction if and only if the spacetime is Lorentzian.
\end{theorem}
%This result implies that the speed of light behaves rather oddly in spacetimes which satisfy the relativity principle if they are not those of general relativity.
It can be shown that the spacetimes which satisfy the relativity principle have light cones which depart very much from isotropy \cite{minguzzi16c}, so all these odd features on the speed of light are not presents in spacetimes which have light cones obtained from small perturbations of the round cones. Of course, they will not satisfy the relativity principle, namely the perturbation spoils the Lorentz group and without restoring any other symmetry groups makes it possible to kinematically distinguish the   observers.
% observers
%allows one, in principle, to kinematically distinguish the  the observers.

%In fact, for applications we are interested in  light cones which are just small perturbations of the round cones, and those do not

\subsection{Legendre transform and dispersion relations} \label{leg}
Let $x\in M$ and let  $V=T_xM$ (we recall that  we might write $I_x$ for $\Omega_x$). The {\em Legendre map}  $\ell\colon I_x\to V^*$ is defined by
\begin{equation} \label{mia}
y \mapsto g_y(y,\cdot)=\dd \mathscr{L}.
\end{equation}
A study in the context of Lorentz-Finsler geometry can be found in \cite{minguzzi13c}. By the Finslerian reverse Cauchy-Schwarz inequality \cite{minguzzi13c}   the Legendre map is a bijection between $I_x$ and the polar cone
\begin{equation} \label{pol}
I_x^{*}=\{p\in V^*\backslash 0: p(w)< 0 \ \textrm{for} \ \textrm{every} \ w\in I_x \}.
\end{equation}
Since $I_x$ is
sharp and  non-empty  so is $I^{*}_x$. Let $g_p$ (of components $g^{\alpha \beta}$) denote the inverse of $g_y$, where $y\in I_x$ is such that $\ell(y)=p$. On the polar cone we define the Finsler Hamiltonian
\[
\mathscr{H}(p):=\frac{1}{2}\,g_p(p,p)=\mathscr{L}(\ell^{-1}(p)) .
\]
It is the Legendre transform of $\mathscr{L}$, it is positive homogeneous of degree two and its Hessian is $g_p$, a metric of  Lorentzian signature.
Clearly, $\ell$ provides a bijection between $\mathscr{I}_x$ and  the dual indicatrix $\mathscr{I}_x^{*}:=\{p\colon -2\mathscr{H}(p)=1\}$.

Now, for every volume form $\mu=\rho(x) \dd^{n+1} y$ on $T_xM$ there is a dual volume form $\mu^*=\rho^{-1}(x) \dd^{n+1} p$ on $T^*_xM$ so that $\mu \mu^*=\vert(\dd p_\mu \wedge \dd y^\mu)^{n+1}\vert$, is the canonical volume form induced from the symplectic 2-form.

Since $\mathscr{H}$ plays the same role for $I^{*}_x$ that $\mathscr{L}$ plays for $I_x$, in order to establish whether $\mathscr{I}_x^{*}$ is an affine sphere we have just to calculate the mean Cartan torsion on $(V^*,\mathscr{H})$ in place of $(V,\mathscr{L})$. By analogy this is given by
\begin{equation} \label{suz}
(I^*)^\alpha=\frac{1}{2}\frac{\p}{\p p_\alpha} \log \vert \det g_{ p}\vert =-\frac{1}{2}\frac{\p y^\beta}{\p p_\alpha}  \frac{\p }{\p y^\beta}\log \vert \det g_{ y }\vert=-\frac{\p y^\beta}{\p p_\alpha}  I_\beta=-I^\alpha.
\end{equation}
Taking into account Eq.\ (\ref{mpc}) and Theorem \ref{rel}  we have just provided a rather simple Finslerian proof of the next duality result (known to Calabi \cite{gigena78,gigena81,loftin10})
\begin{proposition} \label{dus}
$\mathscr{I}_x$ is a hyperbolic affine sphere if and only if  $\mathscr{I}_x^{*}$ is a hyperbolic affine sphere. In this case $\mathscr{I}_x$ has  affine mean curvature $H=-1$ with respect to $\mu=\sqrt{\vert\det g_y\vert} \dd^{n+1} y$ and $\mathscr{I}^{*}$ has  affine mean curvature $H=-1$ with respect to $\mu^*=\sqrt{\vert\det g_p\vert} \dd^{n+1} p$. Given these volumes on $V$ and $V^*$ each affine sphere of affine mean curvature $H$ asymptotic to $\p {I}_x$ is mapped by $\ell$ to an affine sphere of affine mean curvature $H^{-1}$ asymptotic to $\p  {I}_x^{*}$.
\end{proposition}

Let us study the Legendre map using inhomogeneous projective coordinates (Sect.\ \ref{pas}). Let $\{e_\alpha\}$ be a basis of $V$ so that $e_0\in I_x$ and $(e_0+\textrm{Span}\{e_i\})\cap I_x$ is bounded,  and let $y^\alpha$ be the induced coordinates on the vector space $V$.
Let ${\bm v}$ be inhomogeneous projective coordinates on $\{y\in V: y^0>0\}$ so that \begin{equation}
y=(y^0,{\bm y})=-\frac{1}{u}\,(1,{\bm v}).
\end{equation}
The indicatrix $\mathscr{I}_x\subset I_x$ is a graph whose radial graphing function is $u({\bm v})$, that is  $f\colon {\bm v} \mapsto -\frac{1}{u({\bm v})}\,(1,{\bm v})$ is the  hypersurface immersion of the indicatrix. Let $p=g_y(y,\cdot)=p_\mu \dd y^\mu$. In order to calculate $\ell(y)$ for $y\in \mathscr{I}_x$ we have to determine the value of $p_\mu$ such that $p(y)=-1$ and $p(f_*(\tilde e_j))=0$ where $\tilde e_j$ is the basis of $\mathbb{R}^n$. Recalling Eq.\ (\ref{pig}), namely $-uf_*(\tilde e_j)=u_j y+e_j$, the latter condition reads $p_j=u_j$. Thus the former condition reads $p_0+{\bm v}\cdot {\bm p}=u$. That is $p_0=-u^*$, the minus Legendre transform of $u$.

Let us calculate the affine metric and affine connection for the immersion $h\colon {\bm p}\mapsto (-u^*({\bm p}), {\bm p}) $ on $I^{*}_x$. Let $\xi^*({\bm p})= p$. To start with we have  $h_*(e^j)=D_{h_*(e^j)} \xi^*=(-\frac{\p u^*}{\p p_j}, e^j)=(-v^j,e^j)$. Let us denote for shortness $(u^*)^{ij}=\frac{\p^2 u^*}{\p p_i \p p_j}$
\begin{align*}
D_{h_*(e^i)} h_*(e^j)&=D_{h_*(e^i)} D_{h_*(e^j)} \xi^* =\frac{\p^2}{\p p_i \p p_j}(-u^*({\bm p}), {\bm p})=-\frac{(u^*)^{ij} }{u}\, (u,{\bm 0})\\
&=-\frac{(u^*)^{ij}}{u}\,  \{-(-{\bm p}\!\cdot\! {\bm v},{\bm p})+(-u^*, {\bm p})\}=\frac{(u^*)^{ij}}{u}\,  p_k h_*(e^k) -\!\frac{(u^*)^{ij}}{u}\,  \xi^*.
\end{align*}
Since the last equation is split  into a term tangent to the dual indicatrix and a term proportional to $\xi^*$, we can easily read the affine metric and connection coefficients.

\begin{theorem}
The dual indicatrix $\mathscr{I}^{*}_x$ is the image of the immersion
\begin{equation} \label{dua}
h\colon {\bm p}\mapsto (-u^*({\bm p}), {\bm p})
\end{equation}
where $u^*({\bm p})$ is the Legendre transform of $u({\bm v})$. With respect to the transverse field given by the position vector the affine metric is $h^*=-\frac{(u^*)^{ij}}{u}\,\dd p_i \dd p_j$ while the connection coefficients, defined by $\nabla^*_{e^j} e^i=(\nabla^*)_k^{ij} e^k$, are $(\nabla^*)_k^{ij}=\frac{(u^*)^{ij}}{u}\, p_k$.
\end{theorem}

The physical interpretation is obvious. We have seen that in observer coordinates $u({\bm v})$ is the Lagrangian (per unit mass) and ${\bm v}$ is the velocity, thus ${\bm p}$ is the momentum (per unit mass) and $u^*({\bm p})$ is the Hamiltonian (per unit mass). Thus the dual indicatrix is the Cartesian graph of the minus energy. Since observer coordinates are characterized by the expansion $u({\bm v})=-1+{\bm v}^2/2+\cdots$, they are also characterized by the expansion $u^*({\bm p})=1+{\bm p}^2/2+\cdots$. By the properties of the Legendre map, as $u^*$ is the energy ${\bm v}=\p E/\p {\bm p}$, thus the dependence $E({\bm p})$ might also be called {\em dispersion relation} and the usual velocity might be called {\em group velocity}.

\begin{center}

\begin{tabular}{@{} *5l @{}}
\emph{Math.\ objects} & \emph{Meaning} &  \\ \hline \hline
$\mathscr{L}(x,y)$    & Finsler Lagrangian \\
$\mathscr{H}(x,p)$   & Finsler Hamiltonian \\
${\bm v}$& velocity with respect to observer\\
${\bm p}$ & momenta as measured by observer\\
  $u({\bm v})$    & (observer) Lagrangian\\
 $u^*({\bm p})$ & (observer) Hamiltonian \\
 \hline
\end{tabular}

\end{center}

\begin{remark}
Observe that coordinates  ${\bm v}$ and ${\bm p}$ provide quite different parametrizations of the indicatrix and its dual (compare Eqs.\ (\ref{ndu}) and (\ref{dua})). The former are projective inhomogeneous coordinates while the latter coordinates, being Cartesian, might be called homogeneous. The coordinates ${\bm v}$ will be bounded while ${\bm p}$ certainly is not as the dual indicatrix is asymptotic to the dual cone.

Historically,  the first coordinates to be introduced in the study of affine spheres have been the dual coordinates  \cite{calabi72}. Of course, mathematically, one could use inhomogeneous projective coordinates on $T_x^*M$ rather than $T_xM$. Here it is the physical interpretation which dictates to use inhomogeneous coordinates on $T_xM$. In this way the boundedness of the domain $D$ expresses the finiteness of the (direction dependent) speed of light. Also observe that the sum of  velocities,  ${\bm v}+{\bm w}$, makes no sense since it is not projectively covariant. In fact the sum of velocities does not have the classical mechanics interpretation of change of observer. However, the sum of momenta $p_\alpha+q_\alpha$ must indeed be well defined as we need it in order to express, for instance, the conservation of momentum in the collisions of particles. Since the  components $p_\alpha$ are identified with the homogeneous coordinates on the cotangent space, this addition operation is indeed well defined.
\end{remark}

We can complement Def.\ \ref{ksp} with
\begin{proposition} \label{ksq}
An affine sphere spacetime is equivalently determined by:
\begin{itemize}
\item[(d)] affine complete, definite, hyperbolic affine sphere subbundle  of the cotangent bundle with center in the zero section.
\end{itemize}
\end{proposition}

The relationship between $\mathscr{H}$ and $u^*$ is not as simple as that between $\mathscr{L}$ and $u$. Using positive homogeneity of $\mathscr{H}$ it is easy to verify that $\mathscr{H}((p_0,{\bm p}))=-\frac{1}{2} s^2$ where $s>0$ is such that $\frac{p_0}{s}=-u^*(\frac{{\bm p}}{s})$, however in general it is not possible to find $s$ from this equation.

In observer coordinates we can  easily expand $u^*$. In fact we already know the expansion up to quadratic order, and taking into account that $(u^*)^{ij}$ is the inverse of $u_{ij}$ we get (here $a,b,c=1,2,3$)
\begin{align*}
\frac{\p^3 u^*}{\p p_k \p p_j \p p_i}=\frac{\p  (u^*)^{ij}}{\p p_k}=- (u^*)^{ia} (u^*)^{jb} \frac{\p v^c}{\p p_k} \frac{\p u_{ab}}{\p v^c}=- (u^*)^{ia} (u^*)^{jb} (u^*)^{kc} \frac{\p u_{ab}}{\p v^c}.
\end{align*}
We are interested on this value for ${\bm p}=0$ for which $(u^*)^{ij}=\delta^{ij}$ thus arguing similarly for the fourth order term we arrive at the expansion (here $s=1,2,3$ and the Cartan torsion and curvature are evaluated at the observer $\hat y$)
\begin{align}
\begin{split}
u^*({\bm p})=&\sqrt{1+{\bm p}^2}-\frac{1}{3} \,C_{ijk} p_i p_j p_k +\frac{1}{4!}(12 C_{sil} C_{sjk}-2 C_{ijkl}) p_i p_j p_k p_l + o(\vert {\bm p}\vert^4) . \label{esp}
%o(\vert {\bm p}\vert^4),
\end{split}
\end{align}
We recall that the we have used a mass normalized notation. In order to get the non-normalized versions, it is necessary to replace $u\to u/m$, ${\bm p} \to {\bm p}/m$, $u^* \to u^*/m$ in the previous formulas.
%which can be also rewritten
%\begin{align}
%u^*({\bm p})=\sqrt{1+{\bm p}^2}-\frac{1}{3} C_{ijk}(\hat{y}) p_i p_j p_k .
%\end{align}

\subsubsection{ K\"ahler-Einstein condition on the cotangent space} \label{kst}
Proposition \ref{ksq} implies the validity of cotangent versions of the K\"ahler-Einstein characterizations of the affine sphere condition, cf.\ Theorems  \ref{poh} and \ref{poj}.

It is convenient to define a map $*\colon \Omega \to \Omega^*$, as $*=m \ell \circ i_{\mathscr{I}}$, where $\mathscr{I}$ is the indicatrix of a Finsler Lagrangian on $\Omega$, $\ell$ is the Legendre map, and $i_{\mathscr{I}}$ is an inversion with respect to the indicatrix (for homogeneous cones this map was introduced by Vinberg \cite{vinberg63}). In other words, $y^*:=m \ell(y)/[-2 \mathscr{L}(y)]$. Since we have analogous ingredients in $\Omega^*$, we define $p^*:=m \ell^{-1}(p)/[-2 \mathscr{H}(p)]$. It can be easily checked that $*$ is an involutive bijection, that $ 4 \mathscr{H}(y^*) \mathscr{L}(y)=m^2$, and that, defined the K\"ahler potential of $\Omega^*$, $\log V^*$, with $ V^*=(\tfrac{-2\mathscr{H}}{m})^{-m/2}$, we have $V^*(y^*) V(y)=1$.

The Cheng-Yau metric is defined as before through $\hat g_p=\dd ^2 \log V^*$, and the formula analogous to (\ref{nxp}), namely  $\det \hat g_p=-(\det g_p) (V^*)^2$ implies $\det \hat g_p\vert_{p=y^*}=(\det \hat g)^{-1}\vert_y$ (recall that $g(s p)=g(p)$ for $s>0$).

% implies
%Let us define $ V^*=(\tfrac{-2\mathscr{H}}{\sqrt{m}})^{-m/2}$ so that $\log V^*$ is the K\"ahler potential of $\Omega^*$. Since $\mathscr{H}(\ell(y))=\mathscr{L}(y)$ we have $V^*(\ell(y))=V(y)$. The formula analogous to \ref{nxp} implies
%\[
%\det \hat g_p\vert_{p=y^*}=-(\det g_p) (V^*)^2=-(\det g)^{-1} V^{-2}=(\det \hat g)^{-1}\vert_y ,
%\]
From here we can introduce the
 K\"ahler Ricci tensor  for both the Lorentzian and Riemannian metrics
\begin{align}
K^*{}^{\alpha \beta}&:=-\frac{\p^2}{\p p_\alpha \p p_\beta} \log \vert \det {g}_p\vert ,  \qquad
%=-2\frac{\p}{\p p_\alpha} I^*{}^\beta=2 g^{\alpha \mu} \frac{\p}{\p y^\mu} ( g^{\beta \nu} I_\nu),\nonumber\\
%&=-4  C^{\alpha \beta \nu } I_\nu- K^{\alpha \beta},\\
\hat K^*{}^{\alpha \beta}:=-\frac{\p^2}{\p p_\alpha \p p_\beta} \log \det \hat{g}_p. \label{dlp}
\end{align}
then the analogous results to Theorems  \ref{poh} and \ref{poj} state that the K\"ahler-Einstein condition involving any of these tensors is equivalent to the affine sphere condition.
%\end{remark}
%Observe that $\hat K^*{}^{\alpha \beta}(y^*)=-\hat K^{\alpha \beta}(y)$.
%\begin{theorem} \label{pzh}
%The complete, Riemannian, Hessian metric $\hat g$ on $\Omega_x$ is K\"ahler-Einstein if and only if the mean Cartan torsion vanishes: $I_\alpha=0$. In this case $\hat k=-2$ and
%\[
%\det  \frac{\p^2 \log V}{\p p_\alpha \p p_\beta}=\alpha V^2, \qquad \alpha =-\det g_p.
%\]
%\end{theorem}
% If this equation is satisfied, $\hat g$ is called the {\em Monge-Amp\`ere} or the {\em Cheng-Yau metric} of the cone $\Omega_x$.
%%Observe that one could also define
%%\begin{equation} \label{kat}
%%K_{\alpha \beta}:=-\frac{\p^2}{\p y^\alpha \p y^\beta} \log \det {g}_y=-2\frac{\p}{\p y^\alpha} I_\beta.
%%\end{equation}
%
%
%We have a similar result for the Einstein condition, $K_{\alpha \beta}=\kappa(x,y) g_{\alpha \beta}$, on the Lorentzian metric (compare \cite[Sect.\ 5]{ishikawa81}).
%
%\begin{theorem} \label{pzj}
%The Lorentzian Hessian metric $ g$ on $\Omega_x$ is K\"ahler-Einstein if and only if the mean Cartan torsion vanishes: $I_\alpha=0$. In this case $\kappa=0$.
%\end{theorem}

\subsection{Group--phase duality}
We have argued that the velocity ${\bm v}$ is an inhomogeneous projective coordinate on the tangent space and that ${\bm p}$ is a homogeneous projective coordinate on the cotangent space. This approach has given a description of $\mathscr{I}_x$ through an observer Lagrangian $u({\bm v})$.

We might ask what is the description of the indicatrices if we introduce coordinates with reversed roles: inhomogeneous on the cotangent space and homogeneous on the tangent space.
%Let us define
%\begin{align}
%\check {\bm v}&=-(u^*({\bm p}))^{-1} {\bm p}, \qquad &&\textrm{ phase velocity}\\
%\check {\bm p}&=-(u({\bm v}))^{-1} {\bm v}, \qquad &&\textrm{ phase momenta}\\
%\check u(\check {\bm v})&=(u^*({\bm p}))^{-1}, \qquad &&\textrm{ phase Lagrangian}\\
%\check u^*(\check {\bm p})&=(u({\bm v}))^{-1}, \qquad &&\textrm{ phase Hamiltonian}
%\end{align}

Let us define
\begin{align}
\check {\bm v}&=(u^*({\bm p}))^{-1} {\bm p}, \qquad &&\textrm{ phase velocity}\\
\check {\bm p}&=-(u({\bm v}))^{-1} {\bm v}, \qquad &&\textrm{ phase momenta}\\
\check u(\check {\bm v})&=-(u^*({\bm p}))^{-1}, \qquad &&\textrm{ phase Lagrangian}\\
\check u^*(\check {\bm p})&=-(u({\bm v}))^{-1}, \qquad &&\textrm{ phase Hamiltonian}
\end{align}
so that  $\check u^*, u^*>0$; $\check u, u<0$.
We recall that $E=u^*$ is the energy so the definition of phase velocity is the usual one. The immersion of the dual indicatrix ${\bm p}\mapsto (-u^*({\bm p}),{\bm p})$ reads in the new coordinates
\begin{equation}
\check {\bm v}\mapsto -\frac{1}{\check u(\check {\bm v})}\, (-1,\check {\bm v})
\end{equation}
while the immersion of the indicatrix ${\bm v}\mapsto -\frac{1}{ u( {\bm v})} (1,{\bm v})$ reads in the new coordinates
\begin{equation}
\check {\bm p} \mapsto (\check u^*(\check {\bm p}),\check {\bm p}).
\end{equation}
Recalling the definition of polar cone $I^*_x$,
%that $u<0$ and $u^*>0$ on the cone and its polar,
and noticing that  $p_\alpha y^\alpha < 0$ iff $ {\bm v}\cdot \check {\bm v}< 1$, we arrive at
\begin{theorem}
The phase velocity domain is $\check D=D^*$, where $D^*=\{{\bm z}\colon {\bm z}\cdot {\bm v}< 1, \forall {\bm v} \in D\}$ is the dual of the velocity domain. In particular, the velocity is bounded if and only if the phase velocity is bounded. This is so in the physical observer coordinates.
\end{theorem}
We recall that the domain $D$ and hence $D^*$ depends on the choice of observer, namely on the chosen point on the indicatrix, unless the indicatrix is homogeneous in which case, up to rotations, all the domains $D$ coincide.

It is not difficult to check (e.g.\ \cite{minguzzi16c}) that for the isotropic theory the group-phase duality is trivial: the phase quantities coincide with the usual (group) quantities. On the contrary under anisotropy they differ and there will be  directions for which the phase velocity is larger than the (group) velocity and conversely.

We showed in Prop.\ \ref{dus} that $\mathscr{I}_x$  is an affine sphere with  affine mean curvature $H$ if and only if $\mathscr{I}^*_x$ is an affine sphere with  affine mean curvature $H^{-1}$, thus in an affine sphere spacetime  we have the phase dual of Eq.\ (\ref{mon})
\begin{equation} \label{mond}
\det \frac{\p^2 \check u}{\p \check v_i \p \check v_j}=\rho^{-2} \Big(\frac{1}{H \check u}\Big)^{n+2},
\end{equation}
where $\rho^2=\vert \det g_{\alpha \beta}\vert$ and in the mass normalized notation, $H=-1$.

%We can pass to the non-normalized notation through the replacements $u\to u/m$, ${\bm p}\to {\bm p}/m$, $u^*\to u^*/m$, $\check {\bm v} \to \check {\bm v}$, $\check u\to m \check u$, $\check u^* \to m \check u^*$, $\check {\bm p} \to m \check {\bm p}$, thus the previous equation becomes
%\begin{equation} \label{monr}
%\det \frac{\p^2 \check u}{\p \check v_i \p \check v_j}=\frac{1}{\vert\det g_{\alpha \beta}\vert }\Big(-\frac{1}{ \check u}\Big)^{n+2} \frac{1}{m^{2 (n+1)}}.
%\end{equation}

%

\subsubsection{Momentum of a photon}
In this section we explain how to represent a photon. In general Lorentz-Finsler theories the momentum of a photon is given by a point $p\in \p \Omega_x^*$, where $\Omega_x^*$ is the polar cone of $\Omega_x$, the vertical domain of the Lorentz-Finsler Lagrangian at $x$.  In the observer coordinates of an observer $\hat{y}$, the momentum reads $p_{\mu}=h \nu(-1,\check {\bm p})$ where $\check {\bm p}\in \p D^*_{\hat y}$,  $\nu$ is the frequency and $h$ is Planck's constant. The frequency can also be written $-p_\mu \hat y^\mu=h \nu$. Observe that $\check {\bm p}$ is not necessarily normalized unless the speed of light is isotropic.

Suppose that the boundary $\p \Omega_x$ is $C^2$ and with strongly convex sections.
There is  a duality between the projective images of $\p \Omega^*_x$ and $\p \Omega_x$, in fact the hyperplane tangent to $\p \Omega^*_x$ at $p$ defines the direction of a null vector belonging to $\p \Omega_x$, and conversely. In the last section we shall show that this is all is required in order to have  a lightlike geodesic flow on spacetime and a well defined transport of momenta.

If we have more than that, namely if
$g$ is defined and non-degenerate at the boundary of $\Omega_x$, then the Legendre map establishes a one-to-one correspondence between $\p \Omega_x$ and $\p \Omega_x^*$ (cf.\ \cite{minguzzi13c,minguzzi14h}), namely between lightlike vectors and lightlike momenta, which, as we argue in the last section, is not really required for the observational interpretation of the theory.

Suppose we have less than that. Strong convexity of the cone guarantees that at the projective level the Legendre map $P\p \Omega_x\to P\p \Omega^*_x$ is injective, and differentiability guarantees that it is single valued. So  $\p \Omega_x$ could be
  non-differentiable and so have edges or  have  have flat (or non-strongly convex in the projective sense of $D$) parts. Of course, the cone $\Omega_x^*$ would have  a dual behavior. As a consequence, in the  rough theory a point on an edge of $\p \Omega_x$ corresponds to many momenta, and a point on an edge of $\p \Omega_x^*$ corresponds to many velocities. In general, if there is still strong convexity of $\Omega_x^*$ the lack of non-differentiability could be a minor problem since a convex function is almost everywhere differentiable, so the pathological momenta would form a set of vanishing measure. In any case in the general rough theory there is no more a one-to-one correspondence between $P\p \Omega_x$ and $P\p \Omega_x^*$ and the best way to represent a photon is probably as a suitable equivalence class of pairs $(p,y)\in \p \Omega_x^*\times \p \Omega_x$ with  $p_\mu y^\mu=0$.
%such that\footnote{This possibility is suggested by the following argument. If $g$ is well defined and non-degenerate at the boundary of $\Omega_x$ then the Finslerian  reverse Cauchy-Schwarz inequality holds also for the lightlike vectors \cite{minguzzi13c,minguzzi14h}. The equality case of this inequality shows that $p_\mu y^\mu=0$ if and only if $p\propto \ell(y)$.} $p_\mu y^\mu=0$ where $(p,y)\sim (p',y')$ iff $p_\mu {y'}^\mu=0$ and $p'_\mu y^\mu=0$.
These photons can then be regarded as superpositions of extremal photons, which correspond to values of $p,y$ belonging to the extremal points of $\p \Omega_x^*$ and $\p \Omega_x$.
%The photon's frequency is read from its energy $-p_0$ as before.

If one is interested in alternative models which retain most of the good features of general relativity then the rough cone models should be discarded at least on a first study. In fact there are plenty of models with smooth but non-isotropic cones that might attract the researcher's attention. Rough models  have been included in our analysis because they make their appearance in the study of alternative Finslerian realizations of the relativity principle \cite{minguzzi16c} (see also Theor.\ \ref{muf}).

\subsection{Non-relativistic spacetimes are improper affine sphere bundles} \label{aae}
The new formulation of spacetime in terms of affine spheres makes it possible to understand the difference between relativistic and non-relativistic spacetimes without making reference to invariance groups.
%This is consistent with the idea of defining spacetime without reference to metrical-algebraic concepts.

Let us consider the dual affine indicatrix ${\bm p}\mapsto (-u^*({\bm p}), {\bm p})$ in observer coordinates. In the non-relativistic limit we can approximate $u^*=1+{\bm p}^2/2$ thus the dual indicatrix becomes ${\bm p}\mapsto (-1-{\bm p}^2/2, {\bm p})$ which is an improper affine sphere with affine normal $-\frac{\p}{\p p_0}$. This fact suggests to consider improper affine spheres in place of proper affine spheres.
In both the relativistic and non-relativistic theories the energy will be $-p_0$, that is, {\em minus} the zeroth component of the momenta in observer coordinates.

In order to understand the following construction let us recall the transformation rules for kinetic energy (over mass) and momentum (over mass) in classical physics under change of reference frame. Suppose to be in a reference frame $K$ and to change to the frame $K'$ which moves with velocity ${\bm v_0}$ with respect to $K$. If ${\bm p}$ and $T={\bm p}^2/2$ denote the momentum and kinetic energy of a particle in $K$, then those measured in $K'$ can be obtained with the affine transformation
\begin{align}
{\bm p}'&={\bm p}-{\bm v}_0, \label{ae1}\\
 T'&=T-{\bm p}\cdot {\bm v_0}+\frac{{\bm v}_0^2}{2}. \label{ae2}
\end{align}
Coordinates $(-T,{\bm p})$ can be thought  to be determined by a frame in a $n+1$-dimensional affine space. The equation $T={\bm p}^2/2$, being preserved by the frame change, defines an elliptic paraboloid on the affine space.

It is convenient to define the notion of non-relativistic spacetime in analogy with the characterization (d) (Prop.\ \ref{ksq}) for the relativistic case. The idea is to replace the proper affine spheres with {\em improper} affine spheres. While the former have a center and hence can be thought to belong to a vector space (or to an affine space with a selected special point) here we must necessarily work on an affine space.

Furthermore, it is necessary to work on the cotangent space rather than in the tangent space, for a parabolic affine sphere on the latter would determine a special vector field on $M$ given by the affine normal. This vector field would be interpreted as a privileged observer, a feature not present in classical theories. On the cotangent space the affine normal determines instead a one-form field $\psi$ which, if exact, could be written $\psi=\dd t$ where $t$ is the classical time foliation of $M$. Of course on the level sets $t=cnst.$ we would like to have defined a Riemannian (space) metric. This ingredient will be again a consequence of the parabolic sphere on the cotangent bundle.

Another difference is that proper affine spheres might be used to fix a volume form through the  normalization condition $H=-1$, which is why the volume form does not appear in (d). For improper spheres $H=0$, thus the volume form cannot be removed using this trick. Furthermore, parabolic affine spheres on the cotangent bundle obtained from each other by translations along the affine normal will be regarded as equivalent. Physically, this is related to the fact that in classical theories the energy is defined only up to an arbitrary constant.

%Let us  show that their difference is minimal as non-relativistic spacetimes can be defined very similarly replacing {\em proper} affine spheres with {\em improper} affine sphere.

%While the definition of relativistic spacetime Def.\ \ref{ksp} involves {\em hyperbolic} affine spheres, the definition non-relativistic spacetime involves {\em parabolic} (or improper) affine spheres

%Let us replace in the definition of (relativistic) spacetime   Def.\ \ref{ksp}, the ingredient given by  {\em hyperbolic} affine sphere with {\em parabolic} (or improper).

We recall that every affine space $E$ is associated to a vector space $V$, and every vector space has a natural structure of affine space.  Analogously, every affine bundle is associated to a vector bundle, and  every vector bundle can be regarded itself as an affine bundle.
%Every affine bundle admits a section so the affine bundle associated to the same vector bundle are isomorphic.
We also recall that a reference frame on $E$ is given by a pair $(p,\{ e_\alpha\})$ where $p\in E$ and $\{ e_\alpha\}$ is a basis on $V$, then every point $q\in E$ can be written $q=p+x^\alpha e_{\alpha}$ where $\{x^\alpha\}$ are the coordinates determined by the frame.
In the next definition the {\em affine cotangent bundle} is the cotangent bundle regarded as an affine space, or better said it is the affine bundle associated with the cotangent bundle.

\begin{definition} \label{nre}
A {\em non-relativistic spacetime} is a pair given by (a) a volume form on $M$, and (b)
%\begin{itemize}
%\item
an affine complete, definite, parabolic affine sphere subbundle  of the affine cotangent bundle (two parabolic affine spheres are regarded as equivalent if they are obtained through translation along the affine normal).
%\end{itemize}
\end{definition}

Let $E^*$ denote the affine bundle and let $E^*_x$ be its fiber at $x\in M$. $E^*_x$ is an affine space associated with $T^*_xM$. The volume form on $M$ induces a dual volume form $\rho^{-1}(x) \dd p_0 \dd p_1 \cdots \dd p_n$ on $T_x^*M$.
We know from J\"orgens-Calabi-Pogorelov-Cheng-Yau's Theorem \ref{ksw} that the affine sphere is an elliptic paraboloid. We now choose a frame on $E^*_x$. Let $p$ belong to the elliptic paraboloid and let $\{e^\alpha\}$ be a basis of $T_x^*M$ such that $-e^0$ is the affine normal (it is defined using the dual volume form). Denoting with ${p_\alpha}$ the induced coordinates on $E_x$ we have that the affine sphere is a graph  $p_0=b^i p_i-\frac{1}{2}\, c^{ij} p_i p_j$. Applying some linear transformations which consist in a redefinition of $e^i, i=1,\cdots, n$, we can always bring the equation of the elliptic paraboloid to the form $p_0=-\frac{1}{2} {\bm p}^2$.  As the affine normal is $-\p/\p p_0$, Eq.\ (\ref{moo}) is satisfied with $\rho =1$, thus in observer coordinates the dual volume form is $\vert\dd p_0 \dd p_1 \cdots \dd p_n\vert$.

Any coordinate system on $E^*_x$ which brings the equation of the paraboloid to this canonical form is called {\em observer coordinate system} (see Fig.\ \ref{ind}). Clearly, the origin of the coordinate system belongs necessarily to the parabolic sphere since the choice $p_0={\bm p}=0$ satisfies the equation $p_0=-\frac{1}{2} {\bm p}^2$.

\begin{proposition}
Any two observer coordinate systems on $E^*_x$ with origin on the same point of the paraboloid are
 related by $p_{\tilde 0}=p_0$, $p_{\tilde i}=O^{ i}_{\ \tilde i} p_i$ where $O$ is an orthogonal transformation.

  %In particular the volume form $\vert\dd p_^0\dd p_1\cdots\dd y^n\vert$ is well defined and with this choice the affine normal is $\p_0$.
\end{proposition}

\begin{proof}
For fixed origin observer coordinates are uniquely determined up to orthogonal transformations of the spatial part. To see this observe that any two choices, coming from the respective choices for the basis, are linearly related.
%Though we have not (yet) fixed a volume form we have that the affine normal and the affine metric are determined up to a global multiplication constant.
The affine normal and the affine metric of the parabolic sphere are independent of the coordinate system chosen.
Since the affine normal is an invariant we must have for any two choices of observer coordinates $\p/\p p_0= \p/\p p_{\tilde 0}$ which implies $\p p_{\tilde i}/\p p_0=0$ and $\p p_{\tilde 0}/\p p_0=1$. Similarly $\p p_{ i}/\p p_{\tilde 0}=0$. Since $p_{\tilde 0}=p_0+b^i p_i$ it must be $b^i=0$ otherwise we would have that the affine sphere graph $p_{\tilde 0}(\tilde {\bm p})$ has a linear term contrary to the definition of observer coordinates. Thus $p_0$ and $p_{\tilde 0}$ coincide. The spatial components linearly transform among themselves.
Since the affine metric is invariant, we have
 $\delta^{ij} \dd p_i \dd p_j= \delta^{\tilde i \tilde j}\dd p_{\tilde i} \dd p_{\tilde{j}}$, which implies that the transformation $\tilde {\bm p}({\bm p})$ is an orthogonal transformation. $\square$
 %The two undetermined constants $m_1,m_2$ are fixed to 1 by the invariance of the expression $y^0=-1+\frac{{\bm y}^2 }{2}$. The affine normal for the volume form $\vert\dd y^0\dd y^1\cdots\dd y^n\vert$ is $\p_0$ since Eq.\ (\ref{moo}) is satisfied with $\rho=c=1$.
 \end{proof}

The affine sphere represents once again the space of observers (massive particles) and the coordinate system depends furthermore on the orientation of the comoving laboratory. The just found $O(n)$ invariance of the canonical form of the paraboloid is expression of its isotropy. Contrary to the relativistic case {\em there is no anisotropic non-relativistic theory} since there is just one possible observer space, and this space is isotropic. This fact is a consequence of Theorem \ref{ksw} and physically should be expected. In fact if the speed of light is finite then its value can differ in different directions. On the contrary if it is infinite then it is isotropically so.

So far we have only obtained the kinematical aspects of the theory at the single point.

We observe that the normal to the paraboloid $\psi=-e^0=-\p/\p p_0\in T_x^*M$ and the metric $\gamma:=\delta^{ij} \dd p_i \dd p_j \in T_xM\otimes_M T_xM$, where $\{p_i\}$ are observer coordinates, are well defined as independent of the reference frame ($\gamma$ is an extension of the affine metric obtained imposing $\gamma(\cdot, e^0)=0$). We have therefore a triple $(M, \gamma, \psi)$ where $\gamma$ is a contravariant metric and $\psi$ is a non-vanishing one-form such that $\gamma^{\alpha \beta} \psi_\beta=0$. Any such triple is called {\em Galilei structure} \cite{toupin58,trautman63,kunzle72,dixon75,kunzle76,duval85,duval14}
and is precisely the geometrical structure at the basis of  classical (non-relativistic) physics.
%(see also \cite{cartan23,cartan24,toupin58,trautman63,dautcourt64,havas64,malament12} for the geometrization of Newtonian mechanics).
The fact that given a Galilei structure it is possible to reconstruct a paraboloid on the cotangent space (and hence on the affine space associated to it) is pretty obvious using coordinates $\{p_\alpha\}$ which diagonalize $\gamma$ and such that $\psi=-\p/\p p_0$.

\begin{proposition}
The non-relativistic spacetime given by Def.\ \ref{nre} is equivalent to a Galilei structure $(M, \gamma, \psi)$.
\end{proposition}
Once again we see that the definition through affine spheres does not involve metrical-algebraic concepts.

The metric $\gamma$ is equivalent to a metric on the quotient space $T^*_xM/\psi$, which gives an inverse metric $a$ on its dual $(T^*_xM/\psi)^*$. This dual is the vector space of linear functionals on $T^*_xM/\psi$, namely it is the vector space of linear functionals on $T^*_xM$ which vanish on $\psi$. This is precisely $S_x:=\ker \psi \subset T_xM$. Thus the metric $\gamma$ is equivalent to a metric $a$ on $S_x$ called {\em space metric}. Finally, if $\psi$ is locally (globally) exact then the Galilei structure is said to be locally (globally) integrable. There is a function $t\colon U\to \mathbb{R}$, $U\subset M$, such that $\psi=\dd t$. This function is the classical time of the theory.

\begin{remark}{\em K\"ahlerian characterization of non-relativistic  spacetimes}\\
We know that in the relativistic case the (local) affine sphere condition is equivalent to the (vertical) K\"ahler-Ricci flatness of the cotangent space (Sect.\  \ref{kst}). We have also argued that the cotangent indicatrix must be affine complete (Remark \ref{ric}), otherwise it would be possible to reach the maximum speed in finite proper time. These elements determined the characterization (c) of relativistic affine sphere spacetimes. Can non-relativistic spacetimes be characterized with similar conditions? The K\"ahler-Ricci flatness condition cannot be imposed at the level of a Lorentz-Finsler  Hamiltonian $\mathscr{H}$,  since working with an indicatrix asymptotic to a cone would imply an assumption on the finiteness of the speed of light. In the non-relativistic theory we do not have such Finsler Hamiltonian. Still we have the classical Hamiltonian $u^*$ which is used to express the indicatrix as a graph. As a consequence,  we expect that the next result should hold

\begin{theorem}
Assume that  the Blaschke (affine) metric\footnote{The Blaschke metric is $[\det (u^*)^{ij}]^{-1/(n+2)} (u^*)^{ij}$, since the indicatrix is a graph, see \cite[Ex.\ 3.3]{nomizu94}.} of the strictly convex smooth indicatrix ${\bm p}\mapsto (-u^*({\bm p}), {\bm p})$ is complete and the  K\"ahler-Ricci flatness condition
\begin{equation}
\frac{\p^2 }{\p p_i\p p_j } \log \det  \frac{\p^2 u^* }{\p p_i\p p_j }=0
\end{equation}
holds, then the indicatrix is an elliptic paraboloid.
\end{theorem}
In short  the condition on the parallel nature of the Blaschke normals used in the Def.\ \ref{nre} of non-relativistic spacetime, and related to the existence of a `local space' $\ker \psi$, can be replaced by the K\"ahler-Ricci flatness condition. In fact, this theorem along with many of the previously physically motivated result has been  proved. It can be found in a recent work by Li and Xu \cite{li09,li09b,xu15} in which they generalized the classical J\"orgens at al. theorem (Theor.\ \ref{ksw}). It would be nice if one could replace the condition on the affine completeness of the metric with the requirement of unboundedness of the velocity, namely with the request that the map ${\bm p}\mapsto {\bm v}:=\p u^*({\bm p})$ is surjective. This is again possible, see \cite[Main Th.\ $\!\!\!$]{li09b}. It is even possible to remove the condition on the unboundedness of the velocity provided $u\to \infty$ as we approach the maximum velocity. This condition roughly is demanding that though in principle the velocity could be bounded in some directions, the dual indicatrix should not be contained in any cone, namely the `relativistic' Finslerian framework should not apply.
\end{remark}

\begin{remark}{\em Minimal kinematical data}\\
It is convenient to take one step back and consider a minimal set of data which clarifies the connection between relativistic and non-relativistic theories. Let us consider a smooth strictly convex hypersurface $\mathscr{I}^*_x$ on the affine cotangent space. We assume that there are no distinct parallel hyperplanes tangent to $\mathscr{I}^*_x$. Let $\Omega_x$ be the collection of $y\in T_xM$ such that for some $q \in T_x^*M$, $\mathscr{I}^*_x\subset q+\{p\in T_x^*M, y(p)<0\}$.
In other words, we are considering all the tangent hyperplanes to $\mathscr{I}^*_x$ using those to construct a cone of vectors on the tangent space. To every point $p\in \mathscr{I}^*_x$ corresponds a half line (passing through the origin) in $\Omega_x$ given by those vectors whose kernel on $T_x^*M$ is parallel to the hyperplane tangent to $\mathscr{I}^*_x$ at $p$. Conversely, to every half-line in $\Omega_x$  corresponds a momenta on $\mathscr{I}^*_x$ obtained (Hilbert's trick) translating the kernel of the vector till it touches $\mathscr{I}^*_x$. It is also not difficult to show that $\Omega_x$  is open and convex. The strict convexity and smoothness of $\mathscr{I}^*_x$ implies that there is a bijection between momenta, i.e.\ points of  $\mathscr{I}^*_x$,  and velocities, i.e.\ half-lines of $\Omega_x$. Observe that so far   $\Omega_x$ could be non-sharp and even an half-space. The cone $\Omega_x$ gives the set of allowed directions for the timelike curves on $M$. The bijection with $\mathscr{I}^*_x$ provides (for unit mass particles) the momenta of the particle moving on such worldline. An hypersurface $\mathscr{I}^* \subset T^*M$ with the properties given above might be called a {\em minimal spacetime}.

We are going to show that in both relativistic and non-relativistic spacetimes $\mathscr{I}^*_x$ is regarded as an equiaffine hypersurface, the difference being that in the former case the center is `finitely placed' on the cotangent space while in the latter case it is at infinity.  The affine sphere condition implies that there is just one natural choice of center. To see this let us introduce affine coordinates on the cotangent affine space so that $\mathscr{I}^*_x$ is the image of a graph ${\bm p}\mapsto (-u^*({\bm p}), {\bm p})$, $u^*>0$, (observe that we have not yet introduced a star operation, this is just notation). The vector which corresponds to $p=(-u^*({\bm p}), {\bm p})$ is $y=s(1,\p u^*)=:s(1,{\bm v})$, $s>0$. This is a half-line. If we want to select just a representative we need to find an hypersurface $\mathscr{I}\subset \Omega_x$ transverse to the rays which serves as  normalization.

If the hypersurface $\mathscr{I}^*_x$ is regarded as centroaffine with `finitely placed' center and we had wisely selected the affine coordinates so that the center corresponds to the origin, then, since $p$ is transverse to $\mathscr{I}^*_x$, the most natural normalization is obtained imposing $y(p)=-1$, namely $s=-1/u$. In this way we get the representation of the relativistic indicatrix ${\bm v}\mapsto -\frac{1}{u({\bm v})} (1,{\bm v})$. If instead the center is at infinity then the affine coordinates are chosen so that $e_0=\p/\p p_0$ points towards the center, and since it is transverse to $\mathscr{I}^*_x$  the most natural normalization is $y(e_0)=1$ which gives the non-relativistic indicatrix ${\bm v}\mapsto  (1,{\bm v})$. Observe that even in the non-relativistic case the domain of  ${\bm v}$ could be bounded since it is just a section of $\Omega_x$. Note also that there could be several choices of center for $\mathscr{I}^*_x$. It is the affine sphere (or the K\"ahler) condition which implies that there is a natural choice. In the non relativistic case it implies that $\Omega_x$ is an half space, which will allow us to interpret  $\mathscr{I}$ as the first Jet bundle (see below). Another role of the affine sphere condition is the following: on a minimal spacetime given the distribution $x\mapsto \Omega_x$ we cannot recover $\mathscr{I}^*_x$ (up to translations), in fact any strictly convex compact deformation of $\mathscr{I}^*$ would give the same light cone distribution. The affine sphere condition makes it possible to recover the minimal data.
\end{remark}

\begin{remark}
There is a subtle point which requires some clarification. It is well known that ${\bm v} \mapsto (-u({\bm v}),{\bm v})$ is a parabolic affine sphere on $\mathbb{R}^{n+1}$ if and only if ${\bm p}\mapsto (-u^*({\bm p}),{\bm p})$ is a parabolic affine sphere on the dual space $\mathbb{R}_{n+1}$, where $u^*$ and $u$ are Legendre transform of each other. This fact is rather elementary as the Hessian of $u$ is the inverse of the Hessian of $u^*$, and the parabolic affine sphere equation is a condition on the unimodularity for the Hessian (\ref{moo}). One might therefore suspect that  a parabolic affine sphere on the  affine cotangent bundle should determine a parabolic affine sphere on the affine tangent bundle and conversely. We have seen a similar result in the hyperbolic case. This expectation is incorrect, for one reason, while in the proper affine sphere case one can work in dual {\em vector} spaces rather than in  {\em affine} spaces, and the correspondence can be made coordinate independent as it is evident in the Legendre approach passing through $\mathscr{L}$ and $\mathscr{H}$, see Eq.\ (\ref{mia}), here the duality map would depend on the coordinate systems placed on the affine cotangent space  and is therefore, ill defined.
\end{remark}

Still there is one question left. The observer space in the hyperbolic case can be given a geometrical interpretation using either a tangent or a cotangent approach. Here, in the non-relativistic parabolic theory we have given only the cotangent version of the observer space. What is the tangent counterpart?

The 1-form field $\psi$ not only determines its kernel $S_x$, it also determines the locus $\mathscr{I}_x=\{y\in T_xM\colon \psi(y)=1\}$ which is an affine space modeled over the vector space $S_x$. Given $\mathscr{I}_x$ one has a unique one form $\psi$ such that $\mathscr{I}_x=\psi^{-1}(1)$. The triple $(M,\gamma,\psi)$ is therefore equivalent to a triple $(M,\mathscr{I}, a)$ where $x\mapsto \mathscr{I}_x$ is a distribution of affine subspaces of the slit tangent bundle, and $a$ is a metric over them. The observer space in the tangent approach is $\mathscr{I}_x$.  If $\psi$ is exact  we have a time fibration $t\colon M \to T$, $T\subset \mathbb{R}$, and $\mathscr{I}_x$ is nothing but the Jet space $J_x^1(\mathbb{R},M)$, of sections (classical motions) $s\colon T\to M$.

Another way to reach the same conclusion is as follows.
The observer space is an elliptic paraboloid on the cotangent affine space $E^*_x$. The quotient $E^*_x/\psi$ where $\psi$ is the affine normal is in one-to-one correspondence with the paraboloid. It is an affine space modeled over the vector space $T_x^*M/\psi$. As a consequence, since the latter is nothing but the dual space to $S_x$ and $\gamma/\psi$ provides a bijection between the former and the latter, the observer space can be identified with an affine space $\mathscr{I}_x$ associated to $S_x$.

Let $p\in \mathscr{I}_x$ denote an observer. We can find a basis $\{e_\alpha\}$ of the tangent space such that $e_0=p$, $e_i\in S_x$ and $a(e_i, e_j)=\delta_{ij}$. Then every element of the tangent space can be written $y=(u,v_1,\cdots,v_n)$ where the locus $\mathscr{I}_x$ reads $u=1$. This choice of coordinates is an {\em observer coordinate system}, and it is clear that any two choices for the same observer $p$ differ only for the orientation of the axes, namely for an orthogonal transformation of ${\bm v}$. Changes of coordinates between different observers are linear, and since they must preserve both the equation $y^0=1$ and the diagonal form of $a$ they are necessarily of the form $u'=u$, $v^{\tilde i}=O^{\tilde i}_{\ j}(v^j-v_0^j u)$. Let $q_\alpha$ be coordinates in $T_x^*M$ dual to $y^\alpha$, then the coordinate change induces a coordinate change $-q_{\tilde 0}=-q_0+{\bm v}_0 \cdot {\bm q}$, $q_{\tilde i} O^{\tilde i}_{\ j}=q_j$, which is the linear part of the affine change (\ref{ae1})-(\ref{ae2}) where affine coordinates are used.
We remark that there is no Legendre connection between $u$ and $-p_0$.

\begin{figure}[ht!]
\centering
%\psfrag{m}{{\small  $M$}} \psfrag{j}{{\small  $I(x)$}} \psfrag{f}{{\small
%  $T_xM$}} \psfrag{x}{{\small  $x$}} \psfrag{i}{{\small  $\mathscr{I}_x$}} \psfrag{u}{{\small  $I_x$}}  \psfrag{q}{{\small  $T^*_xM$}}  \psfrag{r}{{\small  $\mathscr{I}_x^*$}} \psfrag{v}{{\small  ${\bm v}$}} \psfrag{c}{{\small  $u^*$}} \psfrag{g}{{\small  $p_0$}} \psfrag{w}{{\small  ${\bm p}$}} \psfrag{s}{{\small  $I_x^*$}} \psfrag{k}{{\small  $\psi$}} \psfrag{y}{{\small  $y^0$}} \psfrag{d}{{\small  $D$}} \psfrag{e}{{\small  $E^*_x$}} \psfrag{z}{{\small  ${\bm y}$}}
 \includegraphics[width=9cm]{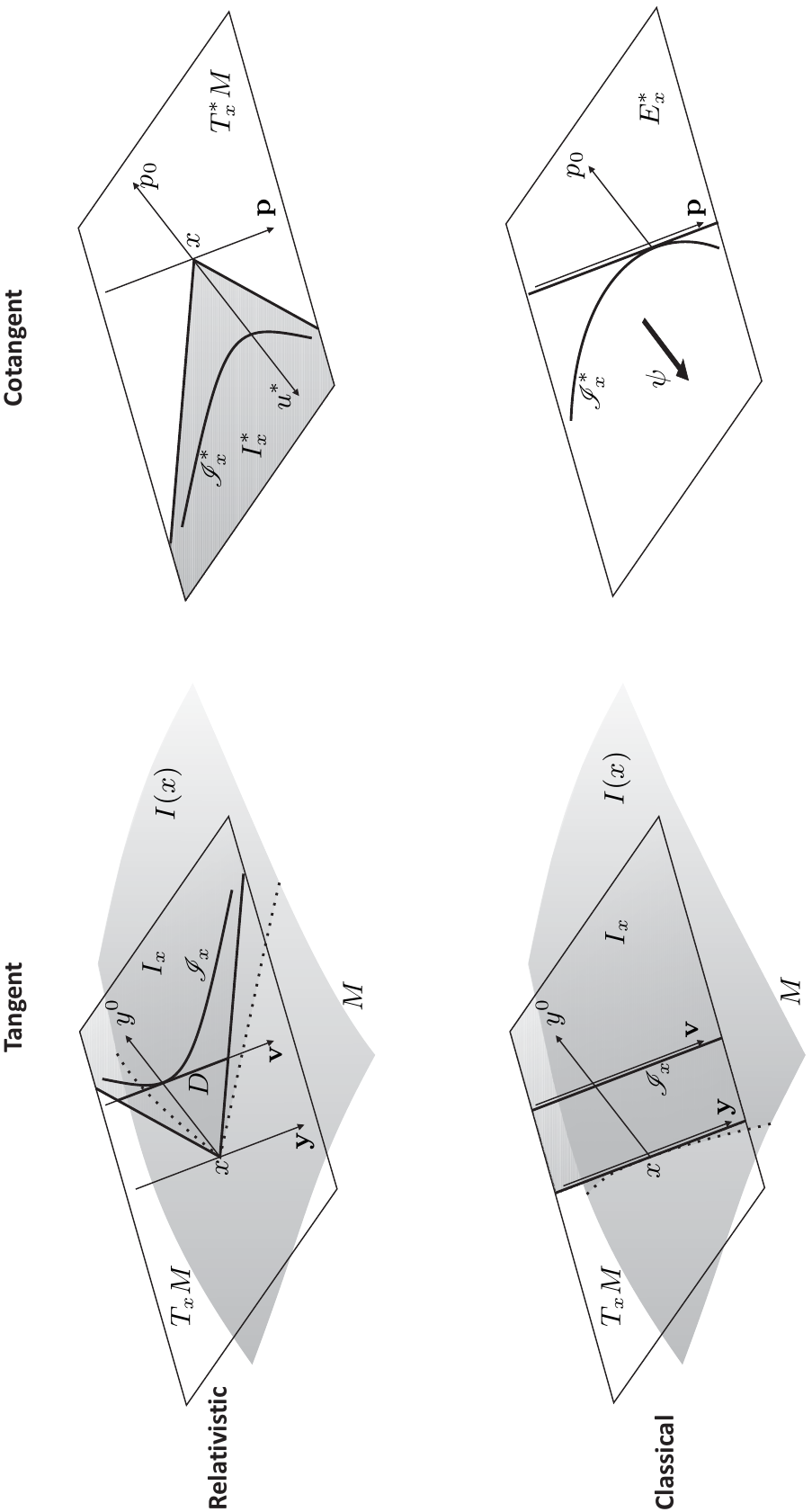}
 \caption{The observer coordinates in the tangent/cotangent space and in the relativistic/classical case. The observer is the intersection between the $y^0$ (or $p_0$) axis and the  indicatrix $\mathscr{I}_x$ (resp.\ $\mathscr{I}_x^*$). The space axes are tangent to the indicatrix at the observer point and diagonalize the affine metric. In the cotangent approach it becomes clear that the difference between relativistic and classical theories stays in the nature of the affine complete affine sphere, respectively hyperbolic or parabolic.} \label{ind}
\end{figure}

We have established that the velocity space is given by a hyperplane  $\mathscr{I}_x$ parallel to $S_x$ on $T_xM$. The cone of the relativistic theory degenerates to the open half space $I_x$  containing $\mathscr{I}_x$. A curve, $s \mapsto x(s)$ is timelike if its tangent vector belongs to $I_{x(s)}$ for every $s$. The following result is pretty easy to prove (take a piecewise timelike curve which approximates a parallelogram whose opposite sides belong to the integral lines of two  vector fields in $\ker \psi$)
\begin{proposition}
The non-relativistic spacetime is locally chronological if and only if $\ker \psi$ is locally integrable, namely $\dd \psi\wedge \psi=0$.
%Thus chronology coincides with the existence of a classical time function.
\end{proposition}

%We have yet to fully specify the
 In this work we considered just the kinematical aspects of the relativistic and non-relativistic theories. We can add further structure as follows.
% have not fully specified the relativistic and non-relativistic theories, as we have only considered the kinematical aspects so far. However, we can anticipate some elements to this section.
In order to compare the kinematical structure at different spacetime points we shall need a connection. In the classical case a natural possibility could be given by a torsionless linear connection $\nabla$ of the tangent bundle such that $\nabla \gamma=\nabla\psi=0$, namely a connection that respects the kinematical structure \cite{toupin58,kunzle72}. Let $T(X,Y)=\nabla_XY-\nabla_Y X-[X,Y]$, since
\begin{align*}
\dd \psi(X,Y)&=X (\psi(Y))-Y(\psi(X))-\psi([X,Y])\\
&=\psi(T(X,Y))+(\nabla_X\psi)(Y)-(\nabla_Y\psi)(X)=0 ,
\end{align*}
the one-form $\psi$ is closed and hence locally exact. Thus the existence of a classical time is really an outcome of the theory.

 %fact We recover the known fact that further

%\section*{Appendix: lightlike spray for Lagrangians non-differentiable on the light cone}

\section{Geodesic flow on the light cone bundle and causality} \label{qon}

Let $N=\p \Omega$ be the hypersurface of $TM\backslash 0$ given by the lightlike vectors.
In this section we assume that
\begin{quote}
($\star$)  $\qquad  N$ is $C^{2,1}$ and the open sharp cones $\Omega_x$ are strongly  convex\footnote{Actually we could just demand $C^{1,1}$ differentiability in the $x$ variable, and $C^{2,1}$ differentiability in mixed or vertical variables. The convexity assumption  means that the local graph of a section of $N_x$, being $C^2$, has positive second order term in the Taylor expansion. The convexity requirement on the cone has strong implications for differentiability (by Alexandrov's theorem convexity implies twice differentiability a.e.) so these conditions are likely to be relaxable.}
\end{quote}
and prove that these differentiability conditions guarantee a well defined notion of lightlike geodesic on spacetime.
 %distribution of light cones $x \to N_x$ induces the unparametrized lightlike geodesics provided it is $C^2$ (namely $N$ is sufficiently smooth).
  %This result is useful since
  Observe that due to the low regularity of $\mathscr{L}$ at the boundary of the cone, we cannot expect to construct the lightlike spray from the metric as done in the Lorentz-Finsler theories for which $\mathscr{L}$ is $C^2$ and Lorentzian on $\bar{\Omega}$ (for other results related to the differentiability of the Lagrangian at the light cone see \cite{minguzzi14h}).

\begin{proposition}
Assume ($\star$). There exist maps $\hat{\mathscr{L}}\colon U \to \mathbb{R}$, $U$ open set,  $N \subset U$, $\hat{\mathscr{L}}\vert_N=0$, $\dd  \hat{\mathscr{L}}_x\vert_N\ne 0$, $\hat{\mathscr{L}}<0$ on $\Omega\cap U$, $\hat{\mathscr{L}}$ is positive homogeneous of degree two in the fiber    and it has the degree of differentiability of $N$.
\end{proposition}

Here  $\hat{\mathscr{L}}_x(y)=\hat{\mathscr{L}}(x,y)$ and it is understood that $\dd  \hat{\mathscr{L}}_x=\frac{\p \hat{\mathscr{L}}}{\p y^\mu} \dd y^\mu$ while $\dd  \hat{\mathscr{L}}=\frac{\p \hat{\mathscr{L}}}{\p y^\mu} \dd y^\mu+\frac{\p \hat{\mathscr{L}}}{\p x^\mu} \dd x^\mu$, thus we have also $\dd \hat{\mathscr{L}}\ne 0$ on $N$.

 Notice that  $\hat{\mathscr{L}}$ is not asked to define  affine spheres on $\Omega$, in fact its domain is just an open neighborhood of $N$, nor it  is demanded to have invertible Hessian on $U$. Of course, there will be many functions with these properties; they might be called {\em subsidiary Lagrangians}. We denote with $\hat L$ the family of such functions.
\begin{proof}
Take a smooth section $x\mapsto P_x\cap N$ of the light cones introducing a distribution of affine hyperplanes $P_x$ on $TM$ (in suitable local coordinates this distribution $x \to P_x$ reads $y^0=1$). By the existence of tubular neighborhoods we can introduce a function $\hat u^2$ near the section so that it has the same degree of differentiability of $N$, vanishes and has non-zero differential on $P_x\cap N$. The function $\hat{\mathscr{L}}$ is obtained using Eq.\ (\ref{osd}). $\square$
\end{proof}

 %Taking a  section of the cones which depends smoothly on $x$ we can construct, by the existence of tubular neighborhoods, $C^{2,1}$  functions $\hat{\mathscr{L}}$ in a neighborhood  $U\supset N$ which are positive homogeneous of degree two and such that $\hat{\mathscr{L}}=0$, $\dd  \hat{\mathscr{L}}_x\ne 0$ on $N$,  and $\hat{\mathscr{L}}<0$ on $\Omega\cap U$ ().

 The  affine sphere Lagrangian $\mathscr{L}$ in general does not belong to $\hat L$. Here the idea is to define a dynamics through $\hat{ \mathscr{L}}$ and show that it is really independent of $\hat{\mathscr{L}}$. As consequence, the dynamics follows solely from the distribution of cones $N$. The result is analogous to the general relativistic result according to which unparametrized lightlike geodesics depend only on the light cone distribution. In that case the result is immediate given the fact that this distribution determines the conformal class of the metric. Here the proof  is totally different since we have the added difficulties of anisotropy and lack of twice differentiability at the cone.

%A way to regard the result could be this: though $\mathscr{L}$ is not differentiable the vary fact that it vanishes on a $C^2$ hypersurface provides that degree of differentiability which

As a first result we prove that every function in $\hat L$ has Lorentzian Hessian at $N$.

\begin{lemma}
Asssume $(\star)$. Every element of $\hat L$ has Lorentzian vertical Hessian on $N$.
%where for $(x,y)\in N$, $y$ is a lightlike vector for the Hessian at $y$.
\end{lemma}

\begin{proof}
For $y\in N_x$, let $g_y:=\frac{\p^2\hat{\mathscr{L}}}{\p y^\mu \p y^\nu}(x,y) \dd y^\mu \dd y^\nu$. By positive homogeneity $g_y(y,y)=2 \hat{\mathscr{L}}=0$
 which proves that the signature of $g_y$ has some zeros or both plus and minus signs. Let $w\in T_y N_x$, $w\not\propto y$, and let $y(t)$ be a curve on $N_x$, where $y=y(0)$, $w=\dot y(0)$. Since $\hat{\mathscr{L}}_x(y(t))=0$ we have differentiating twice and setting $t=0$, $g_y(w,w)+\dd \hat{\mathscr{L}}_x(\ddot y(0))=0$. But since the cone is convex $\ddot y(0)$ points towards the interior namely $\dd \hat{\mathscr{L}}_x(\ddot y(0))<0$  thus $g_y(w,w)>0$ which shows that the signature has at least  $n-1$ plus signs since any $n-1$-dimensional subspace which stays on $\ker \dd \hat{\mathscr{L}}_x\vert_y$ and is transverse to $y$ is a positive subspace for the Hessian. For the other two signs we are left with the possibilities $(-,+)$, $(-,-)$ $(0,0)$, $(\pm,0)$ since they are the only ones that admit a vector with null $g_y$-square. However, $(+,0)$ and $(0,0)$ are not acceptable since there the vectors with null $g_y$-square necessarily annihilate $g_y$, while we have by positive homogeneity $g_y(\cdot,y)=\dd \hat{\mathscr{L}}_x\vert_y\ne 0$ since $y\in N_x$. The choice $(-,-)$ is  not
 acceptable since we have just shown that $g_y$ restricted to $\ker \dd \hat{\mathscr{L}}_x\vert_y$ has signature $(0,+,\cdots,+)$, while for this choice there is no $n$-dimensional subspace with this signature (for these arguments it is useful to keep in mind Cauchy interlacing theorem). The choice $(-,0)$ is  not
 acceptable since again we have just shown that $g_y$ restricted to $\ker \dd \hat{\mathscr{L}}_x\vert_y$ has signature $(0,+,\cdots,+)$ where $y$ annihilates the restricted quadratic form. Here we do have a $n$-dimensional subspace with this signature but the only possibility is that $y$ annihilates $g_y$, which is impossible since $g_y(\cdot,y)=\dd \hat{\mathscr{L}}_x\vert_y\ne 0$. Thus the only possibility is that $g_y$ is Lorentzian. $\square$
\end{proof}

Let $\hat{\mathscr{L}} \in \hat L$, any $C^2$ solution   $t \mapsto x(t)$ of the Lagrange equation
\begin{equation} \label{fod}
\frac{\dd }{\dd t} \frac{\p \hat{\mathscr{L}}}{\p y^\mu}-\frac{\p \hat{\mathscr{L}}}{\p x^\mu}=0, \qquad y^\mu=\frac{\dd x^\mu}{\dd t},
\end{equation}
is such that $\frac{\dd }{\dd t} \hat{\mathscr{L}}(x(t),\dot x(t))=0$, thus it belongs to $N$ if any of its points belongs to it. This equation really defines a flow on $N$, in fact since the vertical Hessian is Lorentzian and hence invertible on $N$ it can be written in the (spray) form
\[
\frac{\dd y^\mu}{\dd t}+g_{y}^{\mu \sigma}\left(\frac{\p^2 \hat{\mathscr{L}}}{\p x^\nu\p y^\sigma} y^\nu -\frac{\p \hat{\mathscr{L}}}{\p x^\sigma}\right)=0, \qquad y^\mu =\frac{\dd x^\mu}{\dd t} .
\]
which is a  Lipschitz vector field over $N$. Existence and uniqueness of the solution is now standard from the theory of ODEs \cite{hartman64}. The map $t \mapsto x(t)$ will be at least $C^{2,1}$; more if $N$ has stronger differentiability properties. We now show that this unparametrized flow depends solely on the distribution of light cones.

\begin{proposition}
Assume $(\star)$.  The unparametrized lightlike solutions to (\ref{fod}) do not depend on the choice of $\hat{\mathscr{L}} \in \hat L$, and so depend solely on the distribution of light cones $N$.
%Let $t \mapsto x(t)$ be a regular  $(\dot x \ne 0)$ $C^2$ solution of the Lagrange equation
%for some $\hat{\mathscr{L}} \in \hat L$, under the assumption that $\dot x$ is lightlike for some $t$ (and hence for every $t$). Then up to reparametrizations the curve $x$ is also a solution of the same equation for any other function choice in $\hat L$.
\end{proposition}

%We shall say that they are the lightlike geodesics of the lightlike distribution $N$. They coincide with the usual unparametrized  lightlike geodesics for $\mathscr{L}$ $C^{2,1}$ at the boundary of the cone.

\begin{proof}
Let $\hat{\mathscr{L}}' \in \hat L$ be another choice. Since $\hat{\mathscr{L}}$ and $\hat{\mathscr{L}}'$ vanish on the hypersurface $N$ and are negative on $\Omega$ we must have $\dd \hat{\mathscr{L}}=\phi \,\dd \hat{\mathscr{L}}'$ for some positive function $\phi$ on $N$. This means that for  $(x,y)\in N$,
\begin{equation} \label{cji}
\frac{\p \hat{\mathscr{L}} }{\p x^\mu}= \phi \frac{\p \hat{\mathscr{L}}'}{\p x^\mu}, \qquad \frac{\p \hat{\mathscr{L}}}{\p y^\mu}=\phi \frac{\p \hat{\mathscr{L}}'}{\p y^\mu}.
\end{equation}
(The latter equation and the positive homogeneity of both $\hat{\mathscr{L}}$ and $\hat{\mathscr{L}}'$ prove that $\phi$ is positive homogeneous of degree zero.)
From (\ref{fod}) it is immediate that $ \hat{\mathscr{L}}$ is conserved and so that $\dot x$ is lightlike for every $t$, namely $(x(t),\dot x(t))\in N$.
Then over the curve
\[
0=\frac{\dd }{\dd t} \frac{\p \hat{\mathscr{L}}}{\p y^\mu}-\frac{\p \hat{\mathscr{L}}}{\p x^\mu}=\phi \left(\left(\frac{\dd }{\dd t} \log \phi \right)\frac{\p \hat{\mathscr{L}}'}{\p y^\mu}+\frac{\dd }{\dd t} \frac{\p \hat{\mathscr{L}}'}{\p y^\mu}-\frac{\p \hat{\mathscr{L}}'}{\p x^\mu}\right)(x(t),\dot x(t))
\]
Let $t'(t)$ be another parametrization to be determined. Since $\hat{\mathscr{L}}'$ is positive homogeneous of degree two in $y$, so is it partial derivative with respect to $x$, while the partial derivative with respect to $y$ is positive homogeneous of degree one, thus from
\[
\left(\!\left(\frac{\dd }{\dd t} \log \phi\right) \frac{\dd t'}{\dd t}\frac{\p \hat{\mathscr{L}}'}{\p y^\mu}+\frac{\dd }{\dd t} \left(\!\left(\frac{\dd t'}{\dd t}\right) \frac{\p \hat{\mathscr{L}}'}{\p y^\mu}\right)\!-\!\left(\frac{\dd t'}{\dd t}\right)^{\!\!2}\frac{\p \hat{\mathscr{L}}'}{\p x^\mu}\right)\left(\! x'(t'),\tfrac{\dd x'}{\dd t'}(t')\right)=0,
\]
where $x'(t'):=x(t)$, we get
\[
\left(\frac{\dd t'}{\dd t}\right)^{-1}\left(\frac{\dd }{\dd t} \log \left(\phi\frac{\dd t'}{\dd t} \right)\right)\frac{\p \hat{\mathscr{L}}'}{\p y^\mu}+ \frac{\dd }{\dd t'} \frac{\p \hat{\mathscr{L}}'}{\p y^\mu}-\frac{\p \hat{\mathscr{L}}'}{\p x^\mu}=0
\]
The choice
\begin{equation} \label{mkz}
t'= \int \frac{1}{\phi(x(t),\dot x(t))} \, \dd t
\end{equation}
shows that the reparametrization $x'(t')$ satisfies (\ref{fod}) with $\hat{\mathscr{L}}$ replaced by $\hat{\mathscr{L}}'$ and $t$ replaced by $t'$.  $\square$
\end{proof}

We conclude
\begin{theorem}
Assume $(\star)$, then we have a natural definition of lightlike geodesic which coincides with that for a subsidiary Lagrangian vertically $C^{2,1}$  at the light cone.
%For any chosen event and lightlike direction there is an unparametrized lightlike geodesic passing through it with that direction.
\end{theorem}

Let $x(t)$ be a lightlike solution to (\ref{fod}), we say that $t$ is a $\hat{\mathscr{L}}$-parametrization of the unparametrized geodesic $x$. Since $\hat{\mathscr{L}}$ is positive homogeneous of degree two, any two $\hat{\mathscr{L}}$-parametrizations are affinely related.
%Of course, the same holds for any other choice  $\hat{\mathscr{L}}'\in \hat L$.
We say that the $\hat{\mathscr{L}}$-parametrization $t$ is {\em syntonized} with the $\hat{\mathscr{L}}'$-parametrization $t'$ if at every point of $x$
\[
{\dd t}=\phi \,\dd t' \quad \textrm{ where } \phi \textrm{ is given by } \quad \dd \hat{\mathscr{L}}= \phi\, \dd \hat{\mathscr{L}}'
\]
%The proof of the previous theorem proves that the $\hat{\mathscr{L}}'$-parametrizations solve the Euler-Lagrange equation for $\hat{\mathscr{L}}'$.
The two parametrizations are syntonized at one event if they are at every event.

\begin{theorem} \label{hst}
Parametrizations relative to different Lagrangians in $\hat L$ assign the same momenta to the lightlike geodesic if and only if they are syntonized.
\end{theorem}

\begin{proof}
It follows from the identity
\begin{align}
%p_\mu&=
&\frac{\p \hat{\mathscr{L}}}{\p y^\mu}(x(t),\tfrac{\dd x}{\dd t}(t))=\phi \frac{\p \hat{\mathscr{L}}'}{\p y^\mu} (x(t),\tfrac{\dd x}{\dd t}(t))=\phi \frac{\dd t'}{\dd t}\frac{\p \hat{\mathscr{L}}'}{\p y^\mu} (x'(t'), \tfrac{\dd x'}{\dd t'}(t')) %\nonumber% \\
%&=\frac{\p \hat{\mathscr{L}}'}{\p y^\mu} (x'(t'), \frac{\dd x'}{\dd t'}(t'))=p'_\mu
\end{align} $\square$
\end{proof}

\subsubsection{Transport of momenta}
The notion of affine parameter on the lightlike geodesic might appear to be  necessary for the physical interpretation  of the theory. However, this is not so, what is mandatory is the possibility of transporting the momenta, while there could be no bijection between momenta and vectors. In fact,  observations are obtained coupling a momenta relative to a massive/massless particle with a velocity space vector (hence timelike) representing an observer. Null vectors are really just intermediate tools in theories like general relativity, used to express important concepts such as lightlike geodesics and momenta. They can be used because in those metrical theories there happen to be a one to one correspondence with lightlike momenta, a correspondence which is purely accidental and related to the fact that $g$ makes sense at the boundary of the cone. This correspondence is now lost in our theory without, however, leading to any loss of physical content. This is really a good feature of  our theory as  it helps us to identify and remove unnecessary mathematical constraints.

Let us elaborate on this important point.
We have seen that each Lagrangian $\hat{\mathscr{L}}$ has Lorentzian Hessian on the light cone, thus there is a  bijective map from lightlike vectors to lightlike momenta. However, this map depends on $\hat{\mathscr{L}}$ and so is not physically relevant. Without a Lagrangian we have just a bijection between lightlike momenta {\em directions} and lightlike vector {\em directions}, a map which follows from the duality between the strongly convex  $C^2$ cone $\Omega_x \subset T_xM\backslash 0$, and its polar $\Omega^*_x \subset T^*_x M\backslash 0$. It turns out that this is all we need for moving the momenta of  a lightlike particle along its geodesic.

%In fact, according to Eqs.\ (\ref{mkz}) and (\ref{cji}), we have using positive homogeneity and syntonized parametrizations
%\begin{align}
%p_\mu&=\frac{\p \hat{\mathscr{L}}}{\p y^\mu}(x(t),\frac{\dd x}{\dd t}(t))=\phi \frac{\p \hat{\mathscr{L}}'}{\p y^\mu} (x(t),\frac{\dd x}{\dd t}(t))=\phi \frac{\dd t'}{\dd t}\frac{\p \hat{\mathscr{L}}'}{\p y^\mu} (x'(t'), \frac{\dd x'}{\dd t'}(t')) \nonumber \\
%&=\frac{\p \hat{\mathscr{L}}'}{\p y^\mu} (x'(t'), \frac{\dd x'}{\dd t'}(t'))=p'_\mu
%\end{align}

So suppose we are given a photon of momenta $p$ at some event. We find first the lightlike direction that corresponds to it using the duality of the cones, then the unparametrized lightlike geodesic.  We choose an element  $\hat{\mathscr{L}}\in \hat L$ and a starting velocity which is mapped to $p$ by the Legendre map of  $\hat{\mathscr{L}}$. This assigns a parametrization to the lightlike geodesic. If we were to choose a different element of $\hat L$ we would still get a syntonized parametrization by Theorem \ref{hst}, thus  since parametrizations that are syntonized at one event remain syntonized, by the same theorem the momenta assigned to the other points of the lightlike geodesic does not  depend on the choice of $\hat{\mathscr{L}}$.

%For any other event over it we determine the momenta of the photon using any Lagrangian $\hat{\mathscr{L}}$. The previous equation shows that its determination is independent of the Lagrangian choice on $\hat L$.
We conclude
\begin{theorem}
Assume $(\star)$.  There is a natural unparametrized flow on $N^*$ such that every integral curve $\gamma$ projects on an unparametrized lightlike geodesic $\sigma$. If $(x, p)\in \gamma$ and $y$ is tangent to $\sigma$ at $x$ then $p_\mu y^\mu=0$.
\end{theorem}

\subsubsection{Affine parameter}
Some results and constructions might still require some notion of affine parameter. In order to introduce this concept the next argument could be of value.
Suppose to have been given not just the distribution of light cones but a Lorentz-Finsler Lagrangian which is Lipschitz on $N$ (this is basically the case of affine sphere relativity provided we assume that the $x$ dependence is not problematic and the vertical Lipschitz constant is locally uniform in $x$; recall the mentioned result by Cheng and Yau \cite[p.\ 53]{cheng77} which assures Lipschitzness in the vertical variable, or  Cor.\ \ref{hud}). We want to show that there is a natural parametrization on the lightlike geodesics determined up to affine transformations. Let $x(t)$ be a lightlike geodesic parametrized accordingly to the Lagrangian $\hat{\mathscr{L}} \in \hat L$.  The differential $\dd \mathscr{L}_x$ exists almost everywhere\footnote{It is really defined everywhere if the cones are smooth see Cor.\ \ref{hud}.} on $N_x$, thus we can find a locally bounded almost everywhere positive measurable function $\varphi:N \to \mathbb{R}$ such that at every $x$ and almost every $y\in N_x$, $\dd \mathscr{L}_x=\varphi \dd \hat{\mathscr{L}}_x$.
Defining
\[
\lambda=\int \varphi(x(t),\dot x(t))\dd t
\]
we obtain a parametrization which up to affine transformations it is independent of $\hat{\mathscr{L}}\in \hat L$ as it follows from Eqs. (\ref{cji}) and (\ref{mkz}). If $\mathscr{L}$ is just Lipschitz at the cone this parameter is defined only on almost every geodesic, since it is required that $\dd \hat{\mathscr{L}}_x$ exists on almost every point of the geodesic. The lightlike geodesic with images that do not comply with this requirement can be still approximated by lightlike geodesics for which the affine parameter is defined.

\subsubsection{Causality theory}
We know that the limit of an accumulating sequence of Finsler causal curves is still a Finsler causal curve, the argument being that given in \cite[Remark 2]{minguzzi14h}.
Let $y_n \in T_xM$, $y_n \to y$, where $y_n$ are timelike and $y$ is lightlike. We did not prove that the timelike geodesics with starting velocity $y_n$ converge in the limit to the lightlike geodesic with starting velocity $y$. We were not able to do so since the non-differentiability of the Lagrangian might imply that the limit is chronal though causal. Still the body of causality theory does not depend on this result but rather on the local achronality of lightlike geodesics (for instance, this property assures that  Cauchy horizons are generated by lightlike geodesics). Here we prove that the lightlike geodesics previously introduced are locally achronal (with respect to the timelike curves defined by $\mathscr{L}$).

\begin{proposition}
Assume ($\star$). The subsidiary Lorentzian Lagrangian $\mathscr{L}\colon U \to \mathbb{R}$, $N \subset U$, can be chosen so that $U=TM\backslash 0$, it has Lorentzian Hessian and   the same degree of differentiability of $N$.
\end{proposition}

%We shall be  interested on just the possibility of defining it in $\overline{\Omega}_x$.
\begin{proof}
On $T_xM$ let us consider the parametrization of the  indicatrix of $\hat{\mathscr{L}}$ (defined near the light cone) through the function $u$ (see Eq.\ (\ref{sop})). Here $u$ is defined just on a half neighborhood $U=\{v\in \bar D: -\epsilon<u(v)\le 0\}$, $ \p D\subset U \subset \bar D$, where $u=0$ on $\p D$. The Lorentzian nature of the Hessian of $\hat{\mathscr{L}}$ is equivalent to the convexity of $u$ (see Eqs.\ (\ref{kod}) and (\ref{bcp})). So we want to show that the $C^2$ function $u$ can be extended to the whole $D$ preserving convexity and its degree of differentiability. Due to a theorem by Min Yan \cite[Theor.\ 3.2]{minyan13} we have just to show that $u$ is convex on $U$ (which means $u(\alpha v_1+\beta v_2)\le \alpha u(v_1)+\beta u(v_2)$ for $\alpha,\beta\ge 0$, $\alpha+\beta=1$, whenever the linear combination belongs to $U$). Let $h$ be the extension of $u$ obtained setting $h=-\epsilon$ inside the region bounded by $\p U $. Its epigraph is convex since it is locally convex and connected (Tietze-Nakajima's theorem) thus $u\vert_U$ was convex (alternatively apply directly the theorem by Ghomi \cite{ghomi02}). We have extended the subsidiary Lagrangian to $\bar{\Omega}_x$. The extension over $T_xM$ follows from the main result of \cite{minguzzi14h}. $\square$
\end{proof}

%Let $N$ be $C^{3,1}$, let $\hat{\mathscr{L}}$ be a  $C^{3,1}$ subsidiary Lagrangian and let us consider the indicatrix of $\hat{\mathscr{L}}$ (defined near the light cone). We can complete it to a convex $C^{3,1}$ indicatrix whose spanned cone coincides with $\Omega_x$. This indicatrix being convex induces a subsidiary Lorentz-Finsler Lagrangian (Eq.\ (\ref{bcp})) which in the neighborhood of the cone coincides with the starting subsidiary Lagrangian. Moreover, this Lagrangian  can be $C^{3,1}$ extended beyond the cone (Whitney extension theorem) keeping the Lorentzian Hessian in a neighborhood of the cone (actually one can extend it all over $TM\backslash 0$ \cite[Theor.\ 1]{minguzzi14h}).

Let $N$ be $C^{3,1}$. As the original and the subsidiary Lagrangian share the same timelike cones and lightlike geodesics the causality of our original Finsler theory coincides with that of the subsidiary Lagrangian. Causality theory for Lagrangians $C^{3,1}$ defined on $TM\backslash 0$ has been developed in \cite{minguzzi13d} where we proved that lightlike geodesics are achronal \cite[Theor.\ 6]{minguzzi13d}. Thus the same holds for the present theory
\begin{corollary}
The lightlike geodesics just introduced are locally achronal with respect to the chronological relation determined by the light cones.
\end{corollary}
As a consequence, standard results of causality theory follow  (e.g.\ \cite[Lemma 2]{minguzzi13d}).
Actually, we have shown that  causality theory does not depend on a Lagrangian as its main results follow just from the definition of $N$.

\section{Conclusions}
Finsler gravity theory has proved successful in generalizing some standard results of causality theory including singularity theorems. %\cite{minguzzi13d,minguzzi14c,minguzzi15}.
In this work we have argued that the family of Finsler spacetimes might be too broad. We showed that in order to establish a one-to-one correspondence with the family of pairs given by (a) sharp convex cone structure and (b) spacetime volume form, we have to restrict ourselves to {\em affine sphere spacetimes}, namely to Finsler spacetimes having affine sphere indicatrices, or equivalently, having vanishing mean Cartan torsion.  Only through this restriction we can preserve the deep correspondence between spacetimes and pairs of spacetime measures and conic orders.

We obtained this result taking advantage of some deep mathematical theorems on affine differential geometry proved in the sixties and seventies by Calabi, Cheng, Yau and other mathematicians. In particular, we showed that for a consistent physical interpretation projective coordinates should be used on the tangent bundle and homogeneous coordinates should be used on the cotangent bundle. These are the coordinate systems which  simplify the Monge-Amp\`ere equation of the affine sphere and which  admit the simplest physical interpretation in terms of velocity and momentum variables. We have also shown that the non-relativistic theory admits the same formulation; it is sufficient to replace proper affine spheres with improper affine spheres. In fact, we have shown that improper affine sphere spacetimes are equivalent to Galilei structures, namely to the kinematical structures representing non-relativistic physics.

In the last section we have returned to the more general context of Lorentz-Finsler  theories. We have shown that the unparametrized lightlike geodesic flow and the transport of lightlike momenta over lightlike geodesics can be consistently defined without using the differentiability of the Lagrangian at the cone, but using just differentiability and convexity conditions on the distribution of light cones. In fact, we have shown that the lightlike geodesics and the transport of lightlike momenta depend solely on the light cone distribution, not on the Lagrangian.
This  is a surprising and very satisfactory result. The affine parameter could be defined in some cases but we  argued that it does not seem necessary for most of the physical interpretation of the theory. Finally, we have shown that the basic results of causality theory do not depend on some Lagrangian differentiability condition at the light cone.
%In general relativity itself it could be considered an  accidental and unnecessary concept.

Much of the difficulties connected with the Finslerian generalization of general relativity are due to an overabundance of mathematical objects.
Fortunately, in the last years a less tensorial, more geometrical approach has helped to clarify several aspects of this theory. Its ability to involve some beautiful but somewhat overlooked chapters of mathematics, such as affine differential geometry, might signal that we could be on the right path for the development of a physical extension of general relativity. Indeed, the overall feeling is that we might be rapidly progressing towards a mature gravitational theory.

%Some old problems have yet to be solved though: (a) it is necessary to  clarify once and for all, which are the dynamical field equations of the theory, possibly passing through a clarification of the problem of energy-momentum conservation;  (b) we should find examples of non-trivial exact solutions; (c) we should study the theory perturbatively. Clearly, in order to proceed on firm ground (b) and (c) would have to follow (a), and after these problems are solved we could really contrast the theory with observations.

\section*{Acknowledgments}
I thank Xu-Jia Wang for suggesting to pass through the regularity of the graph of $u$ to estimate the regularity of $u^2$ (cf.\ Prop.\ \ref{qoi}). I announced the main idea of this  work  in \cite{minguzzi15c}, and  explored other aspects of the condition $I_\alpha=0$ in \cite{minguzzi14c,minguzzi15d}.
This work has been partially supported by GNFM of INDAM.

%\bibliography{../../bibliografie/simultaneity,../../bibliografie/libri,../../bibliografie/miei,../../bibliografie/mieiPrep,../../bibliografie/mieiProc}
%\bibliographystyle{cmp}

\end{document}